\documentclass[UKenglish]{lipics-v2016}
\usepackage{amsmath}
\usepackage{listings}
\usepackage{array}
\usepackage{paralist}
\usepackage{url}
\usepackage{xspace}
\usepackage{smartref}
\usepackage{hyperref}

\DeclareMathOperator{\vaddr}{ADDR}
\DeclareMathOperator{\vexp}{EXP}
\DeclareMathOperator{\vnew}{NEW}

\DeclareMathOperator{\vstate}{STATE}

\newcommand{\comnospace}{\mbox{$\triangleright$}}
\newcommand{\com}{\mbox{\comnospace\ }}
\newcommand{\func}{\textit} 

\newtheorem{Observation}{Observation}

\newcommand{\fakeparagraph}[1]{\noindent\textbf{\textit{#1}.\ }}

\usepackage{etoolbox}
\usepackage{environ}

\newtoggle{isfullver}
\togglefalse{isfullver}
\NewEnviron{fullver}{\iftoggle{isfullver}{\BODY\xspace}{}}
\NewEnviron{shortver}{\iftoggle{isfullver}{}{\BODY\xspace}}

\newtoggle{isthesis}
\togglefalse{isthesis}
\NewEnviron{thesisonly}{\iftoggle{isthesis}{\BODY\xspace}{}}
\NewEnviron{thesisnot}{\iftoggle{isthesis}{}{\BODY\xspace}}

\usepackage[usenames,dvipsnames,svgnames,table]{xcolor}
\definecolor{pdfbgcolor}{RGB}{180,180,180}

\usepackage{textcomp}

\usepackage{tabto}

\usepackage{framed}

\newcommand{\true}{\textsc{True}}
\newcommand{\false}{\textsc{False}}

\newcommand{\op}{\mbox{\sct -record}}
\newcommand{\rec}{\mbox{Data-record}}

\newcommand{\info}{\textit{des}}

\newcommand{\llt}{\func{LLX}}

\newcommand{\sct}{\func{SCX}}

\newcommand{\fail}{\textsc{Fail}}
\newcommand{\finalized}{\textsc{Finalized}}

\lstdefinestyle{nonumbers}{numbers=none}

\begin{document}


\title{Reuse, don't Recycle: Transforming Lock-free Algorithms that Throw Away Descriptors}

\author{Maya Arbel-Raviv}
\author{Trevor Brown}
\affil{Technion, Computer Science Department, Haifa, Israel\\
    \texttt{mayaarl@cs.technion.ac.il, me@tbrown.pro}}
\authorrunning{M. Arbel-Raviv and T. Brown} 

\Copyright{Maya Arbel-Raviv and Trevor Brown}


\maketitle

\begin{abstract}
In many lock-free algorithms, threads help one another, and each operation creates a \textit{descriptor} that describes how other threads should help it.
Allocating and reclaiming descriptors introduces significant space and time overhead.
We introduce the first \textit{descriptor} abstract data type (ADT), which captures the usage of descriptors by lock-free algorithms.
We then develop a \textit{weak descriptor} ADT which has weaker semantics, but can be implemented significantly more efficiently.
We show how a large class of lock-free algorithms can be transformed to use weak descriptors, and demonstrate our technique by transforming several algorithms, including the leading $k$-compare-and-swap ($k$-CAS) algorithm.
The original $k$-CAS algorithm allocates at least $k+1$ new descriptors \textit{per $k$-CAS}.
In contrast, our implementation allocates two descriptors \textit{per process}, and each process simply reuses its two descriptors.
Experiments on a variety of workloads show significant performance improvements over implementations that reclaim descriptors, and reductions of up to three orders of magnitude in peak memory usage.
\end{abstract}

\section{Introduction}

\begin{fullver}
As core counts continue to rise in modern processors, it is increasingly important for applications to be scalable.
Designing scalable concurrent software is notoriously difficult (even for experts), and programmers must rely on efficient concurrent library code to be effective.
Concurrent data structures represent some of the most fundamental building blocks in libraries, and are important in both theory and practice.
\end{fullver}

Many concurrent data structures use locks, but locks have downsides, such as susceptibility to convoying, deadlock and priority inversion.
Lock-free data structures avoid these downsides, and can be quite efficient.
They guarantee that some process will always makes progress, even if some processes halt unexpectedly.
This guarantee is typically achieved with \textit{helping}, which allows a process to harness any time that it would otherwise
spend waiting for another operation to complete. 
Specifically, whenever a process $p$ is prevented from making progress by another operation, it attempts to perform some (or all) of the work of the other operation, on behalf of the process that started it.
This way, even if the other process has crashed, its operation can be completed, so that it no longer blocks $p$. 

In simple lock-free data structures (e.g., \cite{Valois:1995,Harris:2001,Michael:2002,Natarajan:2014}), a process can determine how to help an operation that blocks it by inspecting a small part of the 
data structure.
\begin{thesisonly}
In more complex lock-free data structures (such as~\cite{Ellen:2010,Howley:2012,Shafiei:2013,Brown:2014}, and those implemented with \llt\ and \sct), processes publish \textit{descriptors} for their operations, and helpers look at these descriptors to determine how to help.
\end{thesisonly}
\begin{thesisnot}
In more complex lock-free data structures~\cite{Ellen:2010,Howley:2012,Shafiei:2013,Brown:2014}, processes publish \textit{descriptors} for their operations, and helpers look at these descriptors to determine how to help.
\end{thesisnot}
\xspace
A descriptor typically encodes a sequence of steps that a process should follow in order to complete the operation that created it.

Since lock-free algorithms cannot use mutual exclusion, many helpers can simultaneously help an operation, potentially long after the operation has terminated.
\begin{fullver}
Thus, to avoid situations where helpers read inconsistent data in a descriptor and corrupt the data structure, or try to access a descriptor that has been freed to the operating system and crash, each descriptor must remain consistent and accessible until it can be determined that no helper will ever access it again.
\end{fullver}
\begin{shortver}
Thus, to avoid situations where helpers read inconsistent data in a descriptor and corrupt the data structure, each descriptor must remain consistent and accessible until no helper will ever access it again.
\end{shortver}
This leads to \textit{wasteful algorithms} which allocate a new descriptor for each operation. 

\begin{thesisnot}
In this work, 
\end{thesisnot}
\begin{thesisonly}
In this chapter, 
\end{thesisonly}
we introduce two simple abstract data types (ADTs) that capture the way descriptors are used by wasteful algorithms (in Section~\ref{sec-wasteful}).
The \textit{immutable descriptor} ADT 
provides two operations, \func{CreateNew} and \func{ReadField}, which respectively create and initialize a new descriptor, and read one of its fields.
The \textit{mutable descriptor} ADT 
extends the immutable descriptor ADT by adding two operations: \func{WriteField} and \func{CASField}.
These allow a helper to modify fields of the descriptor (e.g., to indicate that the operation has been partially or fully completed).
We also give examples of wasteful algorithms whose usage of descriptors is captured by these ADTs.

The natural way to implement the immutable and mutable descriptor ADTs is to have \func{CreateNew} allocate memory and initialize it, and to have \func{ReadField}, \func{WriteField} and \func{CASField} perform a read, write and CAS, respectively.
%
Every implementation of one of these ADTs must eventually reclaim the descriptors it allocates.
\begin{fullver}
Otherwise, the algorithm would eventually exhaust memory, and either cause the system to crash, or block while waiting for more memory to become available, violating lock-free progress.
Usually, a lock-free memory reclamation algorithm is used for the reclamation of descriptors.
We briefly explain why reclaiming descriptors this way is expensive.
\end{fullver}
\begin{shortver}
Otherwise, the algorithm would eventually exhaust memory.
We briefly explain why reclaiming descriptors is expensive.
\end{shortver}

In order to safely free a descriptor, a process must know that the descriptor is no longer \textit{reachable}.
This means no other process can reach the descriptor by following pointers in shared memory \textit{or} in its private memory.
State of the art lock-free memory reclamation algorithms such as hazard pointers~\cite{Michael:2004} and DEBRA+~\cite{Brown:2015} can determine when no process has a pointer in its \textit{private} memory to a given object, but they typically require the underlying algorithm to identify a time $t$ after which the object is no longer reachable from \textit{shared} memory. 
%
In an algorithm where each operation removes all pointers to its descriptor from shared memory, $t$ is when $O$ completes.
However, in some algorithms (e.g.,~\cite{Brown:2013}), pointers to descriptors are ``lazily'' cleaned up by subsequent operations, so $t$ may be difficult to identify.
The overhead of reclaiming descriptors comes both from identifying $t$, and from actually running a lock-free memory reclamation algorithm.

Additionally, in some applications, such as embedded systems, it is important to have a small, predictable number of descriptors in the system. 
In such cases, one must use memory reclamation algorithms that aggressively reclaim memory to minimize the number of objects that are waiting to be reclaimed at any point in time. 
Such algorithms incur high overhead.
For example, hazard pointers can be used to maintain a small memory footprint, but a process must perform costly memory fences \textit{every} time it tries to access a new descriptor.


To circumvent the aforementioned problems, we introduce a \textit{weak descriptor} ADT (in Section~\ref{sec-weak-descriptors}) that has slightly \textit{weaker semantics} than the mutable descriptor ADT, but can be implemented \textit{without memory reclamation}. 
The crucial difference 
is that each time a process invokes \func{CreateNew} to create a new descriptor, it \textit{invalidates} all of its previous descriptors.
An invocation of \func{ReadField} on an invalid descriptor \textit{fails} and returns a special value $\bot$.
Invocations of \func{WriteField} and \func{CASField} on invalid descriptors have no effect.
We believe the weak descriptor ADT can be useful in designing new lock-free algorithms, since an invocation of \func{ReadField} that returns $\bot$ can be used to inform a helper that it no longer needs to continue helping (making further accesses to the descriptor unnecessary).

We also identify a class of lock-free algorithms that use the descriptor ADT, and which can be \textit{transformed} to use the weak descriptor ADT (in Section~\ref{sec-weak-transformation}). 
At a high level, these are algorithms in which (1) each operation creates a descriptor and invokes a \func{Help} function on it, and (2) \func{ReadField}, \func{WriteField} and \func{CASField} operations occur only inside invocations of \func{Help}.
Intuitively, the fact that these operations occur only in \func{Help} makes it easy to determine how the transformed algorithm should proceed when it performs an invalid operation: the operation being helped must have already terminated, so it no longer needs help.
We prove correctness for our transformation, and demonstrate
its use by transforming a wasteful implementation of a double-compare-single-swap (DCSS) primitive~\cite{Harris:2002}.

We then present an extension to our weak descriptor ADT, and show how algorithms that perform \func{ReadField} operations \textit{outside} of \func{Help} can be transformed to use this extension (in Section~\ref{sec-adt-extended}).
We prove correctness for the transformation, and demonstrate its use by transforming wasteful implementations of a $k$-compare-and-swap ($k$-CAS) primitive~\cite{Harris:2002} and the LLX and SCX primitives of Brown et~al.~\cite{Brown:2013}.
These primitives can be used to implement a wide variety of advanced lock-free data structures.
For example, LLX and SCX have been used to implement lists, chromatic trees, relaxed AVL trees, relaxed $(a,b)$-trees, relaxed $b$-slack trees and weak AVL trees~\cite{Brown:2014,BrownPhD,He:2016}.

We use mostly known techniques to produce an efficient, provably correct implementation of our extended weak descriptor ADT (in Section~\ref{sec-extended-impl}).
The high level idea is to (1) store a sequence number in each descriptor, (2) replace pointers to descriptors with \textit{tagged sequence numbers}, which contain a process name and a sequence number, and (3) increment the sequence number in a descriptor each time it is reused.

With this implementation, the transformed algorithms for $k$-CAS, and LLX and SCX, have some desirable properties.
In the original $k$-CAS algorithm, \textit{each operation attempt} allocates at least $k+1$ new descriptors.
In contrast, the transformed algorithm allocates only two descriptors \textit{per process, once, at the beginning of the execution}, and these descriptors are reused.
Similarly, in the original algorithm for LLX and SCX, each SCX operation creates a new descriptor, but the transformed algorithm allocates only one descriptor per process, at the beginning of the execution.
%
This entirely eliminates dynamic allocation \textit{and} memory reclamation for descriptors (significantly reducing overhead), and results in an extremely small descriptor footprint.

We present extensive experiments on a 64-thread AMD system and a 48-thread Intel system (in Section~\ref{sec-exp}).
Our results show that transformed implementations always perform at least as well as their wasteful counterparts, and \textit{significantly} outperform them in some workloads.
In a $k$-CAS microbenchmark, our implementation outperformed wasteful implementations using fast distributed epoch-based reclamation~\cite{Brown:2015}, hazard pointers~\cite{Michael:2004} and read-copy-update (RCU)~\cite{Desnoyers:2012} by up to 2.3x, 3.3x and 5.0x, respectively.
In a microbenchmark using a binary search tree (BST) implemented with LLX and SCX, our transformed implementation is up to 57\% faster than the next best wasteful implementation.

The crucial observation in this work is that, in algorithms where descriptors are used only to facilitate helping, a descriptor is no longer needed once its operation has terminated.
This allows a process to reuse a descriptor as soon as its operation finishes, instead of allocating a new descriptor for each operation, and waiting considerably longer (and incurring much higher overhead) to reclaim it using standard memory reclamation techniques.
The challenge in this work is to characterize the set of algorithms that can benefit from this observation, and to design and prove the correctness of a transformation that takes such algorithms and produces new algorithms that simply reuse a small number of descriptors.
As a result of developing this transformation, we also produce significantly faster implementations of 
$k$-CAS, and LLX and SCX.

\begin{thesisnot}
\begin{fullver}
\section{Model} \label{sec-model}


We study a shared memory system with $n$ processes numbered $1..n$.
Each process has a private memory, and there is a shared memory accessible by all processes.
Shared memory consists of base objects (such as read/write registers, compare-and-swap objects, or descriptors), each of which offers a set of operations.
%


A \textit{descriptor} object is \textit{valid} when it is first created, and becomes \textit{invalid} when it is invalidated by a \func{CreateNew} operation.
All other base objects are \textit{always valid}.
Any operation on a valid object (resp., an invalid object) is a valid operation (resp., an invalid operation).
In the following, the terms \textit{operation} and \textit{valid operation} refer to low-level atomic operations on base objects, and the term \textit{high-level operation} refers to an ADT operation that is being implemented \textit{using} base objects as atomic building blocks.

The remainder of the model is relevant mainly when we prove the correctness of our transformations.
	
A \emph{configuration} describes the state of all processes and base objects.
The \textit{state} of a process consists of \textit{private memory}, which is accessible only to the process, and a \textit{program counter}, which is a pointer to the step of the algorithm that the process will execute when it next takes a step.
In the \emph{initial} configuration, each process and base object is in its initial state.
An \emph{execution} $\pi$ is an alternating sequence of configurations
and steps, $C_0 \cdot s_1 \cdots s_i \cdot C_i\cdots$,
where $C_0$ is the initial configuration,
and each configuration $C_i$ is the result of executing step $s_i$
in configuration $C_{i-1}$.


An execution is \emph{linearizable}~\cite{Herlihy:1990} if it is possible to identify, for each high-level operation, a \emph{linearization point} that occurs during the high-level operation, so that the response of each high-level operation is the same as if it were performed atomically at its linearization point.
An algorithm is \emph{linearizable} if all possible executions are linearizable.
\end{fullver}

\end{thesisnot}

\section{Wasteful Algorithms} \label{sec-wasteful}
In this section, we describe two classes of lock-free wasteful algorithms, and give descriptor ADTs that capture their behaviour.
First, we consider algorithms with \textit{immutable} descriptors, which are not changed after they are initialized.
We then discuss algorithms with \textit{mutable} descriptors, which are modified by helpers.

For the sake of illustration, we start by describing one common way that lock-free wasteful algorithms are implemented. 
Consider a lock-free algorithm that implements a set of \textit{high-level} operations.
Each high-level operation consists of one or more \textit{attempts}, which either succeed, or fail due to contention.
Each high-level operation attempt accesses a set of objects (e.g., individual memory locations or nodes of a tree).
Conceptually, a high-level operation attempt locks a subset of these objects and then possibly modifies some of them. 
These locks are special: instead of providing exclusive access to a \textit{process}, they provide exclusive access to a \textit{high-level operation attempt}.
Whenever a high-level operation attempt by a process $p$ is unable to lock an object because it is already locked by another high-level operation attempt $O$, $p$ first \textit{helps} $O$ to complete, before continuing its own attempt or starting a new one.
By helping $O$ complete, $p$ effectively removes the locks that prevent it from making progress.
Note that $p$ is able to access objects locked for a different high-level operation attempt (which is not possible in traditional lock-based algorithms), but only for the purpose of helping the other high-level operation attempt complete.

We now discuss how helping is implemented.
Each high-level operation or operation attempt allocates a new \textit{descriptor} object, and fills it with information that describes any modifications it will perform.
This information will be used by any processes that help the high-level operation attempt.
For example, if the lock-free algorithm performs its modifications with a sequence of CAS steps, then the descriptor might contain the addresses, expected values and new values for the CAS steps.

A high-level operation attempt locks each object it would like to access by publishing pointers to its descriptor, typically using CAS.
Each pointer may be published in a dedicated field for descriptor pointers, or in a memory location that is also used to store application values.
For example, in the BST of Ellen~et~al., nodes have a separate field for descriptor pointers~\cite{Ellen:2010}, but in Harris' implementation of multi-word CAS from single-word CAS, high-level operations temporarily replace application values with pointers to descriptors~\cite{Harris:2002}.

When a process encounters a pointer $ptr$ to a descriptor (for a high-level operation attempt that is not its own), it may decide to help the other high-level operation attempt by invoking a function \textit{Help}$(ptr)$.
Typically, \textit{Help}$(ptr)$ is also invoked by the process that started the high-level operation.
That is, the mechanism used to help is the same one used by a process to perform its own high-level operation attempt.

\begin{fullver}
Wasteful algorithms typically assume that, whenever an operation attempt allocates a new descriptor, it uses fresh memory that has never previously been allocated.
In many algorithms, if this assumption is violated, then an ABA problem can occur.
For example, suppose a process $p$ reads an address $x$ and sees $A$, then performs a CAS that changes $x$ from $A$ to $C$, and interprets the success of this CAS to mean that $x$ contained $A$ at all times between the read and CAS.
If other processes change $x$ from $A$ to $B$ and from $B$ back to $A$ between $p$'s read and CAS, then $p$'s interpretation is invalid, and we say an ABA problem has occurred.
Of course, if memory used to store a descriptor is reclaimed, and later reallocated and used to store another descriptor, this will violate the assumption that each descriptor is allocated new memory which has never previously been allocated.
Note, however, that safe memory reclamation algorithms will reclaim a descriptor only if no process has, or can obtain, a pointer to it.
Thus, no process can tell whether a descriptor is allocated fresh memory or memory that was previously reclaimed, and the ABA problem cannot occur.
\end{fullver}

\begin{shortver}
Wasteful algorithms typically assume that, whenever an operation attempt allocates a new descriptor, it uses fresh memory that has never previously been allocated.
If this assumption is violated, then an \textit{ABA problem} may occur.
Suppose a process $p$ reads an address $x$ and sees $A$, then performs a CAS to change $x$ from $A$ to $C$, and interprets the success of the CAS to mean that $x$ contained $A$ at all times between the read and CAS.
If another process changes $x$ from $A$ to $B$ and back to $A$ between $p$'s read and CAS, then $p$'s interpretation is invalid, and an ABA problem has occurred.
Note that safe memory reclamation algorithms will reclaim a descriptor only if no process has, or can obtain, a pointer to it.
Thus, no process can tell whether a descriptor is allocated fresh or reclaimed memory. 
So, safe memory reclamation will not introduce ABA problems.
\end{shortver}


\subsection{Immutable descriptors} \label{sec-adt-immutable}

%
%

We give a straightforward \textit{immutable descriptor} ADT that captures the way that descriptors are used by the class of wasteful algorithms we just described.
A \textit{descriptor} has a set of fields, and each field contains a value. 
The ADT offers two operations: \func{CreateNew} and \func{ReadField}. 
\func{CreateNew} takes, as its arguments, a descriptor type and a sequence of values, one for each field of the descriptor. 
It returns a unique descriptor pointer $des$ that has never previously been returned by \func{CreateNew}.
Every descriptor pointer returned by \func{CreateNew} represents a new immutable descriptor object.
\func{ReadField} takes, as its arguments, a descriptor pointer $des$ and a field $f$, and returns the value of $f$ in $des$.

In wasteful algorithms, whenever a process wants to create a new descriptor, it simply invokes \func{CreateNew}.
Whenever a helper wants to access a descriptor, it invokes \func{ReadField}.

\subsubsection{Progress}
If the immutable descriptor ADT is implemented so that \func{CreateNew} allocates and initializes a new descriptor, and \func{ReadField} reads and returns a field of a descriptor, then its operations will be \textit{wait-free} (i.e., each operation will terminate after a finite number of its own steps).
However, wait-free descriptor operations are not necessary to guarantee lock-freedom for high-level operations that use descriptors. 
Instead, we simply require descriptor operations to be lock-free.
We now explain why this is sufficient to implement lock-free data structures.

Consider a lock-free algorithm that 
uses a wait-free implementation of the immutable descriptor ADT.
Suppose we transform this algorithm by replacing the wait-free implementation of the descriptor ADT with a \textit{lock-free} implementation.
We argue that the transformed algorithm remains lock-free.
In other words, we show that, if processes take infinitely many steps in the transformed algorithm, then infinitely many high-level operations complete.

In the original algorithm, if processes take infinitely many steps, then infinitely many high-level operations will complete.
The only steps we change to obtain the transformed algorithm are invocations of \func{CreateNew} and \func{ReadField}, some of which might no longer terminate.
Therefore, the only way the transformed algorithm can \textit{fail} to satisfy lock-freedom is if, eventually, all processes take steps only in non-terminating invocations of \func{CreateNew} and \func{ReadField}.
(Otherwise, processes take infinitely many steps of the original algorithm, so infinitely many high-level operations will succeed.)
In this case, only finitely many invocations of \func{CreateNew} and \func{ReadField} will terminate.
However, since \func{CreateNew} and \func{ReadField} are lock-free, infinitely many invocations of \func{CreateNew} and/or \func{ReadField} must terminate. 
Thus, a lock-free implementation of the immutable descriptor ADT is sufficient to implement lock-free algorithms.

\subsubsection{Example Algorithm: DCSS} \label{sec-dcss}
We use the double-compare single-swap (\func{DCSS}) algorithm of Harris et al.~\cite{Harris:2002} as an example of a lock-free algorithm that fits the preceding description.
Its usage of descriptors is easily captured by the immutable descriptor ADT.
A \func{DCSS}($a_1,e_1,a_2,e_2,n_2$) operation does the following \textit{atomically}.
It checks whether the values in addresses $a_1$ and $a_2$ are equal to a pair of expected values, $e_1$ and $e_2$.
If so, it stores the value $n_2$ in $a_2$ and returns $e_2$.
Otherwise it returns the current value of $a_2$.

\begin{figure}[tb]
\begin{minipage}{0.4825\textwidth}
\begin{lstlisting}[name=dcss,frame=single]
 $\func{DCSS}(a_1, e_1, a_2, e_2, n_2):$
   $des := \func{CreateNew}(\func{DCSSdes}, a_1, e_1, a_2, e_2, n_2)$//\label{code-dcss-throw-allocate}
   $fdes := flag(des)$// \label{code-dcss-throw-flag}
   loop //\label{code-dcss-throw-do}
     $r := \func{CAS}(a_2, e_2, fdes)$ //\label{code-dcss-throw-publish-cas}
     if $r\ \mbox{is flagged}$ then $\func{DCSSHelp}(r)$ //\label{code-dcss-throw-dcss-help}
     else exit loop //\label{code-dcss-throw-while}
   if $r = e_2$ then $\func{DCSSHelp}(fdes)$ //\label{code-dcss-throw-dcss-finish}
   return $r$ //\label{code-dcss-throw-dcss-return}

 $\func{DCSSRead}(addr):$
   loop
     $r := *addr$ //\label{code-dcss-throw-read-read}
     if $r\ \mbox{is flagged}$ then $\func{DCSSHelp}(r)$ //\label{code-dcss-throw-read-help}
     else exit loop
   return $r$ //\label{code-dcss-throw-read-return}
\end{lstlisting}
\end{minipage}
\hspace{0.05\linewidth}
\begin{minipage}{0.46\linewidth}
\begin{lstlisting}[name=dcss,frame=single]
 type $\func{DCSSdes}:$ $\{\vaddr_1, \vexp_1, \vaddr_2, \vexp_2, \vnew_2\}$



 $\func{DCSSHelp}(fdes):$
   $des := unflag(fdes)$
   $a_1 := \func{ReadField}(des, \vaddr_1)$ //\label{code-dcss-throw-help-read-a1}
   $a_2 := \func{ReadField}(des, \vaddr_2)$
   $e_1 := \func{ReadField}(des, \vexp_1)$
   if $*a_1 = e_1$ then //\label{code-dcss-throw-help-compare-a1}
     $n_2 := \func{ReadField}(des, \vnew_2)$
     $\func{CAS}(a_2, fdes, n_2)$ //\label{code-dcss-throw-help-cas-new}
   else
     $e_2 := \func{ReadField}(des, \vexp_2)$
     $\func{CAS}(a_2, fdes, e_2)$ //\label{code-dcss-throw-help-cas-old}
\end{lstlisting}
\end{minipage}
    \vspace{-4.5mm}
	\caption{Code for the \func{DCSS} algorithm of Harris et~al.~\cite{Harris:2002} using the \textit{immutable descriptor} ADT.}
    \label{code-throw-DCSS}
\end{figure}

Pseudocode for the \func{DCSS} algorithm appears in 
Figure~\ref{code-throw-DCSS}.
At a high level, \func{DCSS} creates a descriptor, and then attempts to lock $a_2$ by using CAS to replace the value in $a_2$ with a pointer to its descriptor.
Since the \func{DCSS} algorithm replaces values with descriptor pointers, it needs a way to distinguish between values and descriptor pointers (in order to determine when helping is needed).
So, it steals a bit from each memory location and uses this bit to \textit{flag} descriptor pointers. 

We now give a more detailed description.
\func{DCSS} starts by creating and initializing a new descriptor $des$ at line~\ref{code-dcss-throw-allocate}.
It then flags $des$ 
at line~\ref{code-dcss-throw-flag}.
We call the result \textit{fdes} a \textit{flagged pointer}.
\func{DCSS} then attempts to lock $a_2$ in the loop at lines~\ref{code-dcss-throw-do}-\ref{code-dcss-throw-while}.
In each iteration, it tries to store its flagged pointer in $a_2$ using CAS. 
If the CAS is successful, then the operation attempt invokes \func{DCSSHelp} to complete the operation (at line~\ref{code-dcss-throw-dcss-finish}).
Now, suppose the CAS fails.
Then, the \func{DCSS} checks whether its CAS failed because $a_2$ contained another \func{DCSS} operation's flagged pointer (at line~\ref{code-dcss-throw-dcss-help}).
If so, it invokes \func{DCSSHelp} to help the other \func{DCSS} complete, and then retries its CAS.
\func{DCSS} repeatedly performs its CAS (and helping) until the \func{DCSS} either succeeds, or fails because $a_2$ did not contain~$e_2$.

\func{DCSSHelp} takes a flagged pointer $fdes$ as its argument, and begins by unflagging $fdes$ (to obtain the actual descriptor pointer for the operation).
Then, it reads $a_1$ and checks whether it contains $e_1$ (at line~\ref{code-dcss-throw-help-compare-a1}).
If so, it uses CAS to change $a_2$ from $fdes$ to $n_2$, completing the \func{DCSS} (at line~\ref{code-dcss-throw-help-cas-new}).
Otherwise, it uses CAS to change $a_2$ from $fdes$ to $e_2$, effectively aborting the \func{DCSS} (at line~\ref{code-dcss-throw-help-cas-old}).
Note that this code is executed by the process that created the descriptor, and also possibly by several helpers.
Some of these helpers may perform a CAS at line~\ref{code-dcss-throw-help-compare-a1} and some may perform a CAS at line~\ref{code-dcss-throw-help-cas-new}, but only the first of these CAS steps can succeed.

When a program uses DCSS, some addresses can contain either values or descriptor pointers. 
So, each read of such an address must be replaced with an invocation of a function called \func{DCSSRead}.
\func{DCSSRead} takes an address $addr$ as its argument, and begins by reading $addr$ (at line~\ref{code-dcss-throw-read-read}).
It then checks whether it read a descriptor pointer (at line~\ref{code-dcss-throw-read-help}) and, if so, invokes \func{DCSSHelp} to help that \func{DCSS} complete.
\func{DCSSRead} repeatedly reads and performs helping until it sees a value, which it returns (at line~\ref{code-dcss-throw-read-return}).

\subsection{Mutable descriptors} \label{sec-adt-mutable}
In some more advanced lock-free algorithms, each descriptor also contains information about the \textit{status} of its high-level operation attempt, and this status information is used to coordinate helping efforts between processes.
Intuitively, the status information gives helpers some idea of what work has already been done, and what work remains to be done.
Helpers use this information to direct their efforts, and update it as they make progress.
For example, the state information might simply be a bit that is set (by the process that started the high-level operation, or a helper) once the high-level operation succeeds.

As another example, in an algorithm where high-level operation attempts proceed in several phases, the descriptor might store the current phase, which would be updated by helpers as they successfully complete phases.
Observe that, since lock-free algorithms cannot use mutual exclusion, helpers often use CAS to avoid making conflicting changes to status information, which is quite expensive.
Updating status information may introduce contention.
Even when there is no contention, it adds overhead.
Lock-free algorithms typically try to minimize updates to status information.
Moreover, status information is usually simplistic, and is encoded using a small number of bits.

Status information might be represented as a single field in a descriptor, or it might be distributed across several fields.
Any fields of a descriptor that contain status information are said to be \textit{mutable}.
All other fields are called \textit{immutable}, because they do not change during an operation.

\fakeparagraph{Mutable descriptor ADT}
%
%
We now extend the immutable descriptor ADT to provide operations for changing (mutable) fields of descriptors.
The \textit{mutable descriptor} ADT offers four operations: \func{CreateNew}, \func{WriteField}, \func{CASField} and \func{ReadField}.
The semantics for \func{CreateNew} and \func{ReadField} are the same as in the immutable descriptor ADT.
%
\func{WriteField} takes, as its arguments, a descriptor pointer $des$, a field $f$ and a value $v$.
It stores $v$ in field $f$ of $des$. 
\func{CASField} takes, as its arguments, a descriptor pointer $des$, a field $f$, an expected value $exp$ and a new value $v$.
    Let $v_f$ be the value of $f$ in $des$ just before the \func{CASField}.
    If $v_f = exp$, then \func{CASField} stores $v$ in $f$.
    \func{CASField} returns $v_f$.
%
As in the immutable descriptor ADT, we require the operations of the mutable descriptor ADT to be lock-free.

\subsubsection{Example Algorithm: $k$-CAS} \label{sec-kcas}
\begin{figure}[th!]
\begin{lstlisting}[name=kcas,frame=single]
 type $\func{k-CASdes}: \{\vstate, \vaddr_1, \vexp_1, \vnew_1,$ $\vaddr_2, \vexp_2, \vnew_2, \ldots, \vaddr_k, \vexp_k, \vnew_k \}$
  
 //\com \textbf{k-CAS ADT operations}
 $\func{k-CAS}(a_1, e_1, n_1, a_2, e_2, n_2, \ldots, a_k, e_k, n_k):$
   $des := \func{CreateNew}(\func{k-CASdes}, Undecided, a_1, e_1, n_1,\ldots )$ //\label{line-kcasthrow-kcas-createnew}
   $fdes :=$// flagged version of $des$\label{line-kcasthrow-kcas-flag}
   return $\func{k-CASHelp}(fdes)$//\label{line-kcasthrow-kcas-help}
  
 $\func{k-CASRead}(addr):$
   loop
     $r := \func{DCSSRead}(addr)$//\label{line-kcasthrow-kcasread-dcssread}
     if $r \mbox{ is flagged}$ then $\func{k-CASHelp}(r)$ //\label{line-kcasthrow-kcasread-help}
     else exit loop
   return $r$ 
 
 //\com \textbf{Private procedures}
 $\func{k-CASHelp}(fdes):$
   $des :=$// remove the flag from $fdes$ \label{line-kcasthrow-kcashelp-unflag}
   //\com Use DCSS to store $fdes$ in each of $a_1, a_2, \ldots, a_k$ 
   //\com \textit{only} if $des$ has $\vstate$ $\textit{Undecided}$ and $a_i = e_i$ for all $i$
   if $\func{ReadField}(des,\vstate) = \textit{Undecided}$ then//\label{line-kcasthrow-kcashelp-phase1-ifundecided}
     $state := Succeeded$//\label{line-kcasthrow-kcashelp-phase1-setstatesucceeded}
     for $i=1 ... k$ do//\label{line-kcasthrow-kcashelp-phase1-start}
 $retry\_entry$: //\label{line-kcasthrow-kcashelp-phase1-labelretry}
       $a_1 := \func{ReadField}(des,\vstate)$ //\label{line-kcasthrow-kcashelp-phase1-readstate}
       $a_2 := \func{ReadField}(des,\vaddr_i)$ //\label{line-kcasthrow-kcashelp-phase1-readaddr}
       $e_2 := \func{ReadField}(des,\vexp_i)$ //\label{line-kcasthrow-kcashelp-phase1-readexp}
       $val := \func{DCSS}(\langle des,\vstate \rangle,\textit{Undecided},a_2,e_2,fdes)$ //\label{line-kcasthrow-kcashelp-phase1-dcss}
       if $val \mbox{ is flagged}$ then //\label{line-kcasthrow-kcashelp-phase1-iskcas}
         if $val \neq fdes$ then //\label{line-kcasthrow-kcashelp-phase1-ifdifferentkcas}
           $\func{k-CASHelp}(val)$//\label{line-kcasthrow-kcashelp-phase1-help}
           goto $retry\_entry$//\label{line-kcasthrow-kcashelp-phase1-gotoretry}
       else
         if $val \neq e_2$ then //\label{line-kcasthrow-kcashelp-phase1-ifnotexpected}
           $state := Failed$ //\label{line-kcasthrow-kcashelp-phase1-setstatefailed}
           break //\label{line-kcasthrow-kcashelp-phase1-break}
     $\func{CASField}(des,\vstate,\textit{Undecided},state)$ //\label{line-kcasthrow-kcashelp-phase1-casstate}\label{line-kcasthrow-kcashelp-phase1-end}

   //\com Replace $fdes$ in $a_1, ..., a_k$ with $n_1, ..., n_k$ or $e_1, ..., e_k$
   $state : = \func{ReadField}(des,\vstate)$ //\label{line-kcasthrow-kcashelp-phase2-readstate}
   for $i=1 ... k$ do //\label{line-kcasthrow-kcashelp-phase2-start}
     $a = \func{ReadField}(des,\vaddr_i)$ //\label{line-kcasthrow-kcashelp-phase2-readaddr}
     if $state = Succeeded$ then //\label{line-kcasthrow-kcashelp-ifsucceeded}
       $new := \func{ReadField}(des,\vnew_i)$ //\label{line-kcasthrow-kcashelp-phase2-readnew}
     else
       $new := \func{ReadField}(des,\vexp_i)$ //\label{line-kcasthrow-kcashelp-phase2-readexp}
     $\func{CAS}(a,fdes,new)$ //\label{line-kcasthrow-kcashelp-phase2-casnew}\label{line-kcasthrow-kcashelp-phase2-end}
   return $(state = Succeeded)$ //\label{line-kcasthrow-kcashelp-return}
\end{lstlisting}
\vspace{-4.5mm}
\caption{Code for the \func{k-CAS} algorithm of Harris et~al.~\cite{Harris:2002} using the \textit{mutable descriptor} ADT.}
\label{code-throw-kCAS}
\vspace{-3mm}
\end{figure}
%
A \textit{$k$-CAS}($a_1, ..., a_k,$ $e_1, ..., e_k,$ $n_1, ..., n_k$) operation atomically does the following.
First, it checks if each address $a_i$ contains its expected value $e_i$.
If so, it writes a new value $n_i$ to $a_i$ for all $i$ and returns true.
Otherwise it returns false.

The $k$-CAS algorithm of Harris~et~al.~\cite{Harris:2002} is an example of a lock-free algorithm that has descriptors with mutable fields.
At a high level, 
a $k$-CAS operation $O$ in this algorithm 
starts by creating a descriptor that contains its arguments. 
It then tries to lock each location $a_i$ \textit{for the operation} $O$ by changing the contents of $a_i$ from $e_i$ to $des$, where $des$ is a pointer to $O$'s descriptor.
If it successfully locks each location $a_i$, then it changes each $a_i$ from $des$ to $n_i$, and returns true.
If it fails because $a_i$ is locked for another operation, then it helps the other operation to complete (and unlock its addresses), and then tries again.
If it fails because $a_i$ contains an application value different from $e_i$, then the $k$-CAS fails, and unlocks each location $a_j$ that it locked by changing it from $des$ back to $e_j$, and returns false.
(The same thing happens if $O$ fails to lock $a_i$ because the operation has already terminated.)

We now give a more detailed description of the algorithm.
Pseudocode appears in Figure~\ref{code-throw-kCAS}.
A $k$-CAS operation creates its descriptor at line~\ref{line-kcasthrow-kcas-createnew}, and then 
invokes a function \func{k-CASHelp} to complete the operation.
%
In addition to the arguments to its $k$-CAS operation, a $k$-CAS descriptor contains a 2-bit \textit{state} field that initially contains \textit{Undecided} and is changed to \textit{Succeeded} or \textit{Failed} depending on how the operation progresses.
This \textit{state} field is used to coordinate helpers.

Let $p$ be a process performing (or helping) a $k$-CAS operation $O$ that created a descriptor $d$.
If $p$ fails to lock some address $a_i$ in $d$, then $p$ attempts to change the \textit{state} of $d$ using CAS from \textit{Undecided} to \textit{Failed}.
On the other hand, if $p$ successfully locks each address in $d$, then $p$ attempts to change the \textit{state} of $d$ using CAS from \textit{Undecided} to \textit{Succeeded}.
Since the \textit{state} field changes only from \textit{Undecided} to either \textit{Failed} or \textit{Succeeded}, only the first CAS on the state field of $d$ will succeed.
The $k$-CAS implementation then uses a lock-free DCSS primitive (the one presented in Section~\ref{sec-dcss}) to ensure that $p$ can lock addresses for $O$ \textit{only} while $d$'s \textit{state} is \textit{Undecided}.
This prevents helpers from erroneously performing successful CAS steps after the $k$-CAS operation is already over.

Recall that the DCSS algorithm allocates a descriptor for each DCSS operation.
A $k$-CAS operation performs potentially \textit{many} DCSS operations (at least $k$ for a successful $k$-CAS), and also allocates its own $k$-CAS descriptor.
The $k$-CAS algorithm need not be aware of DCSS descriptors (or of the bit reserved in each memory location by the DCSS algorithm to flag values as DCSS descriptor pointers), since it can simply use the \func{DCSSRead} procedure described above whenever it accesses a memory location that might contain a DCSS descriptor.
However, 
the $k$-CAS algorithm performs DCSS on the \textit{state} field of a $k$-CAS descriptor, which is 
accessed using the $k$-CAS descriptor's \func{ReadField} operation.
To allow DCSS to access the \textit{state} field, we must modify DCSS slightly.
First, instead of passing an address $a_1$ to DCSS, we pass a pointer to the $k$-CAS descriptor and the name of the \textit{state} field (at line~\ref{line-kcasthrow-kcashelp-phase1-dcss} of Figure~\ref{code-throw-kCAS}).
Second, we replace the read of $addr_1$ in DCSS (at line~\ref{code-dcss-throw-help-compare-a1} of Figure~\ref{code-throw-DCSS}) 
with an invocation of \func{ReadField}.

\newpage

Since $k$-CAS descriptor pointers are temporarily stored in memory locations that normally contain application values, the $k$-CAS algorithm needs a way to determine whether a value in a memory location is an application value or a $k$-CAS descriptor pointer.
In the DCSS algorithm, the solution was to reserve a bit in each memory location, and use this bit to \textit{flag} the value contained in the location as a pointer to a DCSS descriptor.
Similarly, the $k$-CAS algorithm reserves a bit in each memory location to flag a value as a $k$-CAS descriptor pointer.
The $k$-CAS and DCSS algorithms need not be aware of each other's reserved bits, but they should not reserve the same bit (or else, for example, a DCSS operation could encounter a $k$-CAS descriptor pointer, and interpret it as a DCSS descriptor pointer).

When the $k$-CAS algorithm is used, some memory addresses may contain either values or descriptor pointers, so reads of such addresses must be replaced by a \func{k-CASRead} operation.
This operation reads an address, and checks whether it contains a $k$-CAS descriptor pointer.
If so, it helps the $k$-CAS operation to complete, and tries again.
Otherwise, it returns the value it read.
For further details, on the $k$-CAS algorithm refer to~\cite{Harris:2002}.

\section{Weak descriptors} \label{sec-weak-descriptors}
In this section we present a \textit{weak descriptor} ADT that has weaker semantics than the mutable descriptor ADT, but can be implemented more efficiently (in particular, without requiring any memory reclamation for descriptors).
We identify a class of algorithms that use the mutable descriptor ADT, and which can be transformed to use the weak descriptor ADT, instead.

We first discuss a restricted case where operation attempts only create a single descriptor, and we give an ADT, transformation and proof for that restricted case.
(In the next section, we describe how the ADT and transformation can be modified slightly to support operation attempts that create multiple descriptors.)


\subsection{Weak descriptor ADT} \label{sec-adt-weak}

The weak descriptor ADT is a variant of the mutable descriptor ADT that allows some operations to \textit{fail}. 
%
To facilitate the discussion, we introduce the concept of descriptor validity.
Let $des$ be a pointer returned by a \func{CreateNew} operation $O$ by a process $p$, and $d$ be the descriptor pointed to by $des$.
In each configuration, $d$ is either \textbf{valid} or \textbf{invalid}.
Initially, $d$ is valid.
If $p$ performs another \func{CreateNew} operation $O'$ \textit{after} $O$, then $d$ becomes invalid immediately after $O'$ (and will never be valid again).

We say that a \func{ReadField}$(des, ...)$, \func{WriteField}$(des, ...)$ or \func{CASField}$(des, ...)$ operation is performed \textbf{on a descriptor} $d$, where $des$ is a pointer to $d$.
An operation on a valid (resp., invalid) descriptor is said to be valid (resp., invalid).
Invalid operations have no effect on any base object, and return a special value $\bot$ (which is never contained in a field of any descriptor) instead of their usual return value.
We say that a \func{CreateNew} operation $O$ is performed \textbf{on a descriptor} $d$ if $O$ returns a pointer to $d$.
Observe that a \func{CreateNew} operation is always valid.
We say that a process $p$ \textbf{owns} a descriptor $d$ if it performed a \func{CreateNew} operation that returned a pointer $des$ to $d$.

The semantics for \func{CreateNew} are the same as in the mutable descriptor ADT.
The semantics for the other three operations are the same as in the mutable descriptor ADT, except that they can be invalid.
%
As in the previous ADTs, these operations must be lock-free.

\subsection{Transforming a class of algorithms to use the weak descriptor ADT} \label{sec-weak-transformation}

We now formally define a class of lock-free algorithms that use the mutable descriptor ADT, and can easily be transformed so that they use the weak descriptor ADT, instead.
We say that a step $s$ of an execution is \textit{nontrivial} if it changes the state of an object $o$ in shared memory, and \textit{trivial} otherwise.
In particular, all invalid operations are trivial, and an unsuccessful CAS or a CAS whose expected and new values are the same are both trivial.
In the following, we abuse notation slightly by referring interchangeably to a descriptor and a pointer to it.



\begin{definition} \label{def:WCA}
	Weak-compatible algorithms (WCA) are lock-free wasteful algorithms that use the mutable descriptor ADT, and have the following properties:  
	
	\begin{enumerate}
		\item Each high-level operation attempt $O$ by a process $p$ may create (and initialize) a single descriptor $d$.
        Inside $O$, $p$ may perform at most one invocation of a function \func{Help}$(d)$ (and $p$ may not invoke \func{Help}$(d)$ outside of $O$).
        \label{def:WCA:prop:alg-steps}
        \item A process may help any operation attempt $O'$ by another process by invoking \func{Help}$(d')$ where $d'$ is the descriptor that was created by $O'$.
        \label{def:WCA:prop:helping}
		\item If $O$ terminates at time $t$, 
        then any steps taken in an invocation of \func{Help}$(d)$ after time $t$ are \textit{trivial} (i.e., do not \textbf{change} the state of \textbf{any} shared object, incl. $d$).\label{def:WCA:prop:not-change-mem}
		\item While a process $q \neq p$ is performing \func{Help}$(d)$, $q$ cannot change any variables in its private memory that are still defined once \func{Help}$(d)$ terminates (i.e.,~variables that are local to the process $q$, but are not local to \func{Help}).
        \label{def:WCA:prop:return-help}
		\item All accesses (read, write or CAS) to a field of $d$ occur inside either \func{Help}$(d)$ or $O$.
        \label{def:WCA:prop:desc-ops-only-in-help}
	\end{enumerate}
\end{definition}

At a high level, properties~\ref{def:WCA:prop:alg-steps} and~\ref{def:WCA:prop:helping} of WCA describe how descriptors are created and helped.
Property~\ref{def:WCA:prop:return-help} intuitively states that, whenever a process $q$ finishes helping another process perform its operation attempt, $q$ knows only that it finished helping, and does not remember anything about what it did while helping the other process.
In particular, this means that $q$ cannot pay attention to the return value of \func{Help}.
We explain why this behaviour makes sense.
If $q$ creates a descriptor $d$ as part of a high-level operation attempt $O$ and invokes \func{Help}$(d)$, then $q$ might care about the return value of \func{Help}, since it needs to compute the response of $O$.
However, if $q$ is just helping another process $p$'s high-level operation attempt $O$, then it does not care about the response of \func{Help}, since it does not need to compute the response of $O$.
The remaining properties, \ref{def:WCA:prop:not-change-mem} and~\ref{def:WCA:prop:desc-ops-only-in-help}, allow us to argue that the contents of a descriptor are no longer needed once the operation that created it has terminated (and, hence, it makes sense for the descriptor to become invalid).
%
%
In Section~\ref{sec-adt-extended}, we will study a larger class of algorithms with a weaker version of property~\ref{def:WCA:prop:desc-ops-only-in-help}. 

\fakeparagraph{The transformation}
Each algorithm in WCA can be transformed in a straightforward way into an algorithm that uses the weak descriptor ADT as follows.
Consider any \func{ReadField} or \func{CASField} operation $op$ performed by a high-level operation attempt $O$ in an invocation of \func{Help}$(d)$, where $d$ was created by a \textit{different} high-level operation attempt $O'$.
Note that $op$ is performed while $O$ is \textit{helping} $O'$.
After $op$, a check is added to determine whether $op$ was invalid, in which case $p$ returns from \func{Help} immediately.
(In this case, \func{Help} does not need to continue, since $op$ will be invalid only if $O'$ has already been completed by the process that owns $d$ or a helper.)

\fakeparagraph{Example Algorithm: DCSS}
\begin{figure}[t]
\begin{lstlisting}[name=dcsstransform,frame=single]
 $\func{DCSSHelp}(fdes):$
   $des := \func{Unflag}(fdes)$
   $a_1 := \func{ReadField}(des, \vaddr_1)$ //\label{code-dcss-weak-ugly-help-read-a1}
   if $a_1 = \bot$ then return //\label{code-dcss-weak-ugly-help-checkreturn-v}
   $a_2 := \func{ReadField}(des, \vaddr_2)$
   if $a_2 = \bot$ then return //\label{code-dcss-weak-ugly-help-checkreturn-addr2}
   $e_1 := \func{ReadField}(des, \vexp_1)$ //\label{code-dcss-weak-ugly-help-read-exp1}
   if $e_1 = \bot$ then return //\label{code-dcss-weak-ugly-help-checkreturn-exp1}
   if $*a_1 = e_1$ then //\label{code-dcss-weak-ugly-help-compare-a1}
     $n_2 := \func{ReadField}(des, \vnew_2)$ //\label{code-dcss-weak-ugly-help-read-new2}
     if $n_2 = \bot$ then return //\label{code-dcss-weak-ugly-help-checkreturn-new2}
     $\func{CAS}(a_2, fdes, n_2)$ //\label{code-dcss-weak-ugly-help-cas-new}
   else
     $e_2 := \func{ReadField}(des, \vexp_2)$ //\label{code-dcss-weak-ugly-help-read-exp2}
     if $e_2 = \bot$ then return //\label{code-dcss-weak-ugly-help-checkreturn-exp2}
     $\func{CAS}(a_2, fdes, e_2)$ //\label{code-dcss-weak-ugly-help-cas-old}
\end{lstlisting}
\vspace{-4.5mm}
\caption{Applying the transformation to \func{DCSS}.} 
\label{code-weak-ugly-DCSS}
\end{figure}
Figure~\ref{code-weak-ugly-DCSS} shows code for the \func{DCSS} algorithm in Figure~\ref{code-throw-DCSS} 
that has been \textit{transformed} to use the weak descriptor ADT. 
There, we include only the \func{DCSSHelp} procedure, since it is the only one that differs from Figure~\ref{code-throw-DCSS}.
The transformation adds lines~\ref{code-dcss-weak-ugly-help-checkreturn-v}, \ref{code-dcss-weak-ugly-help-checkreturn-addr2}, \ref{code-dcss-weak-ugly-help-checkreturn-exp1}, \ref{code-dcss-weak-ugly-help-checkreturn-new2} and \ref{code-dcss-weak-ugly-help-checkreturn-exp2} to \textit{check} whether the preceding invocations of \func{ReadField} are invalid.

\subsection{Correctness}
We argue that our transformation takes a linearizable algorithm $\mathcal{A} \in$ WCA that uses mutable descriptors and produces a linearizable algorithm $\mathcal{A}'$ that uses weak descriptors.
Consider any execution $e'$ of the transformed algorithm $\mathcal{A'}$.
We prove there exists an execution $e$ of the original algorithm $\mathcal{A}$ that performs the \textit{same} high-level operations, in the same order, and with the same responses, as in $e'$.
We explain how this helps.
Since $e$ is a correct execution of the original algorithm $\mathcal{A}$, the high-level operations performed in $e$ must respect the sequential specification(s) of the object(s) implemented in $\mathcal{A}$.
Furthermore, since $e'$ performs the same high-level operations, in the same order, and with the same responses, the high-level operations in $e'$ must also respect the sequential specification(s) of the same object(s).
Therefore, the transformed algorithm $\mathcal{A'}$ is correct.

We construct 
$e$ 
as follows.
By Property~\ref{def:WCA:prop:desc-ops-only-in-help} of WCA, all \func{ReadField}, \func{WriteField} and \func{CASField} operations occur in \func{Help}.
Whenever a check by a process $p$ follows a \func{ReadField} or \func{CASField} in $e'$ that returns $\bot$ (because the operation attempt $O$ being helped by $p$ has already terminated), we replace that check by a consecutive sequence of steps in which $p$ finishes its invocation of \func{Help}.
All other checks immediately following \func{ReadField} or \func{CASField} are simply removed.

By Property~\ref{def:WCA:prop:not-change-mem} of WCA, none of the steps added to $e$ change the state of any shared object.
So, these steps will not change the behaviour of any other process.
We also argue that none of these steps make any changes to $p$'s private memory that persist after $p$ finishes its invocation of \func{Help}.
(I.e., any changes these steps make to $p$'s private memory are \textit{reverted} by the time $p$ finishes its invocation of \func{Help}, so $p$'s private memory is the same just after the invocation of \func{Help} as it was just before the invocation of \func{Help}.)
So, these steps will not change the behaviour of $p$ after it finishes its invocation of \func{Help}.
Observe that, whenever a process performs a \func{ReadField} or \func{CASField} operation on a descriptor that it created, this operation will return a value different from $\bot$.
This is due to Property~\ref{def:WCA:prop:alg-steps} of WCA, and the definition of the weak descriptor ADT, which states that $d$ becomes invalid only after $O$ has terminated.
Since $p$'s invocation of \func{ReadField} or \func{CASField} returns $\bot$, $p$ must therefore be performing \func{Help}$(d)$ where $d$ was created by a \textit{different} process.
Thus, Property~\ref{def:WCA:prop:return-help} of WCA implies that, after $p$ performs the sequence of steps to finish its invocation of \func{Help}$(d)$, its private memory has the same state as it did just before it invoked \func{Help}. 


\subsection{Reading immutable fields efficiently}
If an invocation of \func{Help}$(des)$ accesses many immutable fields of a descriptor, then we can optimize it by replacing many \func{ReadField} operations with a single, more efficient operation called \func{ReadImmutables}.
This operation reads and returns \textit{all} of a descriptor's immutable fields, unless the descriptor is invalid, in which case it returns $\bot$.
%


To use \func{ReadImmutables} in \func{Help}$(des)$, one can simply perform, at the beginning of \func{Help}, a \func{ReadImmutables} operation, followed by an \textit{if}-statement that checks whether it the operation invalid, and, if so, returns immediately.
Then, in the body of \func{Help}$(des)$, each invocation of \func{ReadField}$(des, f)$, where $f$ is immutable, is replaced with a direct read from the set of values returned by \func{ReadImmutables}.
\begin{figure}[tb]
\begin{lstlisting}[name=dcssreadimmutables,frame=single]
 $\func{DCSSHelp}(fdes):$
   $des := \func{Unflag}(fdes)$
   $values := \func{ReadImmutables}(des)$ //\label{code-weak-nice-DCSS-readimmutables}
   if $values = \bot$ then return
   $\langle a_1, e_1, a_2, e_2, n_2 \rangle := values$

   if $*a_1 = e_1$ then
     $\func{CAS}(a_2, fdes, n_2)$
   else
     $\func{CAS}(a_2, fdes, e_2)$
\end{lstlisting}
\vspace{-4.5mm}
\caption{Using \func{ReadImmutables} to optimize and streamline the transformed DCSS algorithm.}
\label{code-weak-nice-DCSS}
\end{figure}
We demonstrate this approach on the transformed pseudocode for DCSS in Figure~\ref{code-weak-ugly-DCSS}.
Figure~\ref{code-weak-nice-DCSS} shows the result. 
Since all fields of a DCSS descriptor are immutable, \textit{every} invocation of \func{ReadField} can be replaced with a direct read from the result of the \func{ReadImmutables} operation performed at line~\ref{code-weak-nice-DCSS-readimmutables}.
(This will not be the case in an algorithm where the \func{Help} procedure reads mutable fields.)
Since \func{ReadImmutables} replaces several invocations of \func{ReadField}, it has the added benefit of making code simpler and shorter.

\section{Extended Weak Descriptors} \label{sec-adt-extended}

In this section, we describe an extended version of the weak descriptor ADT, and an extended version of the transformation in Section~\ref{sec-weak-transformation}.
This extended transformation weakens property~\ref{def:WCA:prop:desc-ops-only-in-help} of WCA so that \func{ReadField} operations on a descriptor $d$ can also be performed \textit{outside} of \func{Help}$(d)$.
%
At a high level, we handle \func{ReadField} operations performed outside of \func{Help} as follows. 
For \func{ReadField}s performed inside \func{Help}, we have seen that we can simply stop helping when $\bot$ is returned. 
However, for \func{ReadField}s performed outside of \func{Help}, it is not clear, in general, how we should respond if $\bot$ is returned. 
Intuitively, the goal is to find a value that \func{ReadField} can return so that the algorithm will behave the same way as it would if the descriptor were still valid. 
In some algorithms, just knowing that an operation has been completed gives us enough information to determine what a \func{ReadField} operation should return (as we will see below). 

\fakeparagraph{Extended weak descriptor ADT}
This ADT is the same as the weak descriptor ADT, except that \func{ReadField} is extended to take, as an additional argument, a default value $dv$ that is returned instead of $\bot$ when the operation is invalid.
Observe that the weak descriptor ADT is a special case of the extended weak descriptor ADT where each argument $dv$ to an invocation of \func{ReadField} is $\bot$.


\fakeparagraph{The extended transformation}
\func{CASField} and \func{WriteField} operations are handled the same way as in the WCA transformation.
However, an invocation of \func{ReadField}$(des, f)$ is handled differently depending on whether it occurs inside an invocation of \func{Help}$(des)$.
If it does, it is replaced with an invocation of \func{ReadField}$(des, f, \bot)$ followed by the check, as in the WCA transformation.
If not, it is replaced with an invocation of \func{ReadField}$(des, f, dv)$, where the choice of $dv$ is specific to the algorithm being transformed.

Let $\mathcal{A}$ be any algorithm that uses mutable descriptors, and satisfies properties \ref{def:WCA:prop:alg-steps}-\ref{def:WCA:prop:return-help} of WCA algorithms (see Definition~\ref{def:WCA}), as well as a weaker version of property~\ref{def:WCA:prop:desc-ops-only-in-help}, called property~\ref{def:WCA:prop:desc-ops-only-in-help}$'$, which states: every write or CAS to a field of a descriptor $d$ must occur in an invocation of \func{Help}$(d)$.
Let $e$ be an execution of $\mathcal{A}$ and let $e'$ be an execution that is the same as $e$, except that one (arbitrary) descriptor $d$ becomes invalid at some point $t$ after the high-level operation attempt $O$ that created $d$ terminates.
(When we say that $d$ becomes invalid at time $t$, we mean that after $t$, each invocation of \func{ReadField}$(d, f, dv)$ that is performed outside of \func{Help}$(d)$ returns its default value $dv$.)

Let $O'$ be any high-level operation attempt in $e'$ which, after $t$, performs \func{ReadField} on $d$ outside of \func{Help}$(d)$.
We say that an extended transformation is \textit{correct for} $\mathcal{A}$ if, for all choices of $e$, $e'$, $d$, $t$, and $O'$, the exact same changes are performed by $O'$ in $e$ and $e'$ to any variables that are still defined once $O'$ terminates (i.e.,~variables that are local to the process performing $O'$, but are not local to $O'$, and variables in shared memory), and $O'$ returns the same response in both executions.
An algorithm $\mathcal{A}$ is an \textit{extended weak-compatible algorithm} (and is in the class \textit{EWCA}) if there is an extended transformation that is correct for $\mathcal{A}$.

\subsection{Correctness}
Consider any extended transformation which is correct for a linearizable algorithm $\mathcal{A}$ that uses mutable descriptors.
We prove the result of applying this transformation to $\mathcal{A}$ is a linearizable algorithm $\mathcal{A}'$ that uses extended weak descriptors.
Specifically, let $e'$ be any execution of $\mathcal{A}'$.
We prove there is an execution $e$ of $\mathcal{A}$ that performs the same high-level operations, in the same order, with the same responses, as in $e'$.

%

First, we define an execution $e_0$.
Whenever a check in $e'$ by a process $p$ in \func{Help}$(d)$ determines that the preceding \func{ReadField} or \func{CASField} on a descriptor $d$ is \textit{invalid} (which means that the operation attempt being helped by $p$ has already terminated), we replace that check by a consecutive sequence of steps in which $p$ finishes its invocation of \func{Help}$(d)$.
By Property~\ref{def:WCA:prop:not-change-mem} of WCA, none of these added steps change the state of any shared variable.
Moreover, by Property~\ref{def:WCA:prop:return-help} of WCA, $p$ does not change any variable that is still defined after its invocation of \func{Help}, so $p$ has the same local state after \func{Help} in $e_0$ and $e'$.
Whenever such a check determines that the preceding \func{ReadField} or \func{CASField} is \textit{valid}, we simply remove this check.
Observe that each invalid operation in $e_0$ is an invalid \func{ReadField} operation on some descriptor $d$ performed outside of \func{Help}$(d)$.

Let $d_1, d_2, ...$ be the sequence of descriptors created in $e_0$.
We inductively construct a sequence $e_1, e_2, ...$ of executions such that $e_i$ differs from $e_{i-1}$ only in that descriptor $d_i$ never becomes invalid in $e_i$.
Specifically, for each high-level operation attempt $O'$ that performs an invalid \func{ReadField} operation on descriptor $d_i$ outside of \func{Help}$(d_i)$, consider the first such \func{ReadField} operation $R$.
All of the steps of $O'$ prior to $R$ are the same in $e_i$ as in $e_{i-1}$.
After $R$, $O'$ continues to take steps in $e_i$, but each \func{ReadField} operation that $O'$ performs on a field $f$ of $d_i$ returns the contents of $f$ (instead of a default value).
This may result in $O'$ executing completely different code paths in $e_{i-1}$ and $e_i$.
However, by the definition of an extended transformation that is correct for $\mathcal{A}$, $O'$ returns the same response in $e_i$ and $e_{i-1}$ and performs the \textit{exact same changes} to any variables that are still defined once $O'$ terminates.
Thus, for each variable $v$ that is still defined once $O'$ terminates, we can schedule the sequence of changes to $v$ in the exact same way in $e_i$ and $e_{i-1}$ (which implies that any reads in $e_{i-1}$ which see these changes can be scheduled appropriately in $e_i$). 

Since the claim holds for all $i$, there is an execution $e$ in which no descriptor becomes invalid (so $e$ is an execution of $\mathcal{A}$), and the same high-level operation attempts are performed, in the same order, and with the same responses.

\subsection{Multiple descriptors per operation attempt}
In some lock-free algorithms, a high-level operation attempt can create several different descriptors, and potentially invoke a different \func{Help} procedure for each descriptor.
We describe how to adjust the definitions above to support these kinds of algorithms.
For simplicity, we think of there being a single \func{Help} procedure that checks the type of the descriptor passed to it, and behaves differently for different types.

In order to allow a high-level operation attempt to create multiple descriptors without simply invalidating the ones it previously created, we update the definition of valid and invalid descriptors.
Let $des$ be a pointer to a descriptor $d$ of type $T$ returned by a \func{CreateNew} operation $C$ performed by process $p$.
Initially, $d$ is valid.
If $p$ performs another \func{CreateNew} operation $C'$ with the \textit{same descriptor type} $T$ after $C$, then $d$ becomes invalid immediately after $C'$ (and will never be valid again).

With this definition of valid and invalid descriptors, it might initially seem like an operation cannot create multiple descriptors of the same type $T$.
However, this turns out not to be a problem.
If an operation should create multiple descriptors of type $T$, we can simply imagine creating multiple \textit{clone} types $T_1, T_2, ...$ that have the exact same fields as $T$.
To create $k$ descriptors of type $T$, one would then create $k$ clone types, and have an operation invoke \func{CreateNew} once for each clone type.
(However, we are unaware of any algorithms in which a high-level operation attempt creates multiple descriptors of the same type.)

We also slightly modify Property~\ref{def:WCA:prop:alg-steps} of (extended) weak-compatible algorithms, as follows, to accommodate the use of multiple descriptors.
Each high-level operation attempt $O$ by a process $p$ may create (and initialize) a sequence $D$ of descriptors, each with a \textbf{unique type}.
Inside $O$, $p$ may perform at most one invocation of a function \func{Help}$(d)$ for each $d \in D$ (and $p$ may not invoke \func{Help}$(d)$ outside of $O$).
Note that the proof for the extended weak transformation goes through unchanged.

\subsection{Example Algorithm: k-CAS}
In this section, we explain how the extended transformation is applied to the $k$-CAS algorithm presented in Section~\ref{sec-kcas}.
Note that no invocations of \func{ReadField} on a DCSS descriptor $des$ are performed outside of \func{HelpDCSS}$(des)$.
\begin{fullver}
There is only one place in the code where an invocation $I$ of \func{ReadField} on a $k$-CAS descriptor $des$ is performed \textit{outside} of \func{Help}$(des)$ (the \func{Help} procedure for $k$-CAS).
\end{fullver}
\begin{shortver}
There is only one place in the algorithm where an invocation $I$ of \func{ReadField} on a $k$-CAS descriptor $des$ is performed \textit{outside} of \func{Help}$(des)$ (the \func{Help} procedure for $k$-CAS).
\end{shortver}
Specifically, $I$ reads the \textit{state} field of a $k$-CAS descriptor inside the modified version of 
\func{HelpDCSS}.
Recall that the $k$-CAS algorithm passes a $k$-CAS descriptor pointer and the name of the \textit{state} field as the first argument to DCSS at line~\ref{line-kcasthrow-kcashelp-phase1-dcss} of Figure~\ref{code-throw-kCAS}, and the DCSS algorithm is modified to use \func{ReadField} at line~\ref{code-dcss-throw-help-compare-a1} of Figure~\ref{code-throw-DCSS} to read this \textit{state} field.
We choose the default value $dv = \textit{Succeeded}$ for this invocation of \func{ReadField}.
We explain why this extended transformation of the $k$-CAS algorithm is correct.

When $I$ is performed at line~\ref{code-dcss-throw-help-compare-a1} of \func{DCSSHelp} (in Figure~\ref{code-throw-DCSS}), its response is compared with $e_1$, which contains \textit{Undecided}.
If $I$ returns \textit{Undecided}, then the CAS at line~\ref{code-dcss-throw-help-cas-new} is performed, and the process $p$ performing $I$ returns from \func{HelpDCSS}.
Otherwise, the CAS at line~\ref{code-dcss-throw-help-cas-old} is performed, and $p$ returns from \func{HelpDCSS}.

Suppose $I$ is invalid.
Then, we know the $k$-CAS operation attempt that created $des$ has been completed.
We use the following algorithm specific knowledge.
After a $k$-CAS operation attempt has completed, its $k$-CAS descriptor has \textit{state} \textit{Succeeded} or \textit{Failed} (and is never changed back to \textit{Undecided}).
(This can be determined by inspection of the code.)
Thus, if $I$ were valid, its response would \textit{not} be \textit{Undecided}, and $p$ would perform the CAS at line~\ref{code-dcss-throw-help-cas-old} and return from \func{HelpDCSS}.
Since $dv = Succeeded$, $p$ does exactly the same thing when $I$ is invalid.
(Note that the exact value of \textit{state} is unimportant.
It is only important that it is not \textit{Undecided}.)

\smallskip
\subsection{Example Algorithm: LLX and SCX}
\begin{figure*}[ph]
    \vspace{-32mm}\hspace{-35mm}\includegraphics[scale=1]{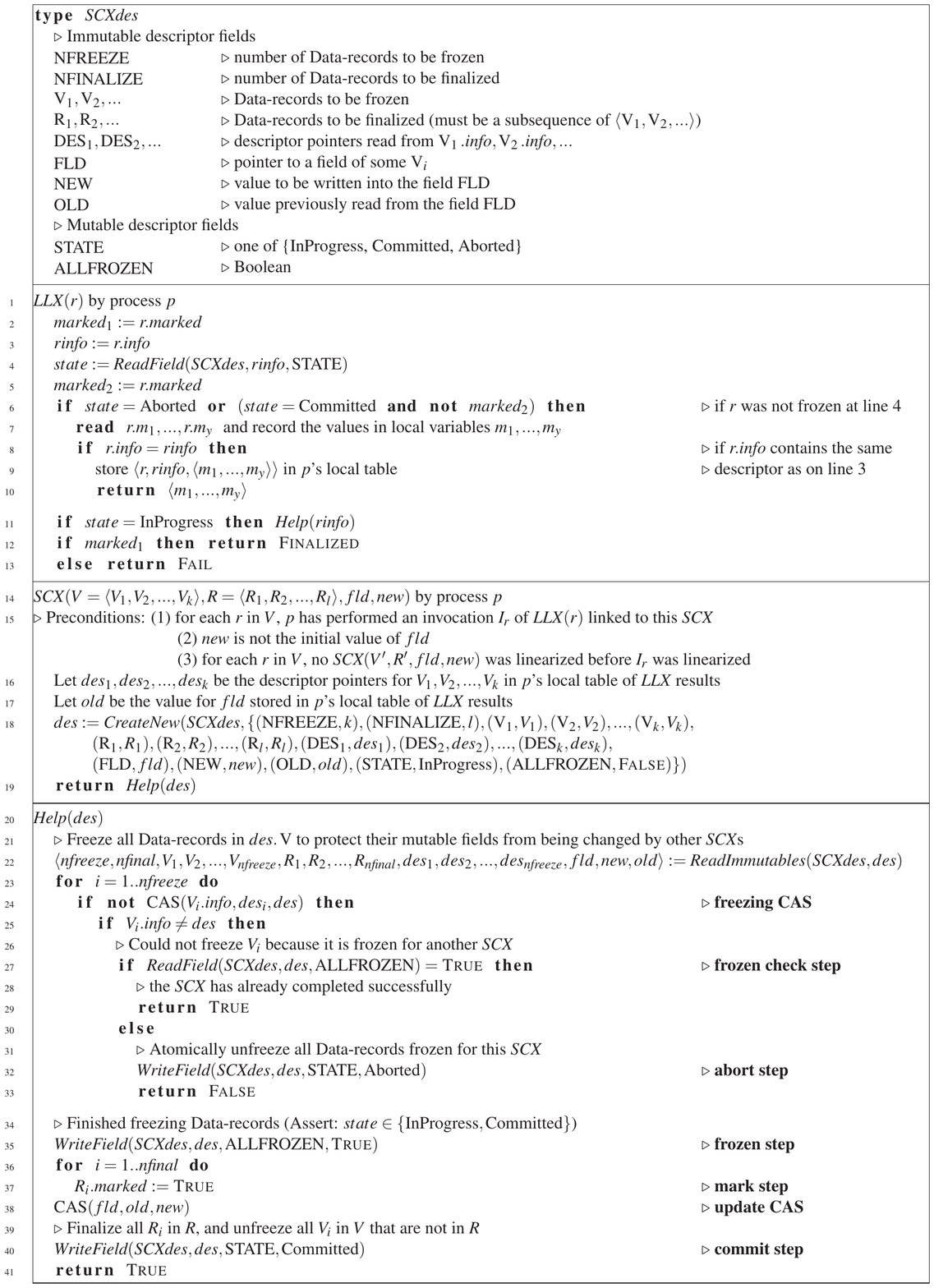}
    \vspace{-24mm}\caption{Code for the \llt\ and \sct\ algorithm 
        using the \textit{mutable descriptor} ADT.}
    \label{code-desc-scx}
\end{figure*}

\begin{thesisnot}
In this section, we explain how the extended transformation is applied to the multiword synchronization primitives load-linked-extended (\llt) and store-conditional-extended (\sct) of Brown et~al.~\cite{Brown:2013}.
Note that Brown et~al.~\cite{Brown:2014} also used these primitives to design a tree update template that can be followed to produce a fast lock-free implementation of any data structure based on a down-tree (a directed acyclic graph where each node has indegree one).
Thus, by optimizing \llt\ and \sct, we also optimize the tree update template, and all of the data structures that have been implemented with it.
\end{thesisnot}
\begin{thesisonly}
In this section, we explain how the extended transformation is applied to the \llt\ and \sct\ implementation in Chapter~\ref{chap-scx}.
\end{thesisonly}
 Pseudocode for \llt\ and \sct\ using mutable descriptors is presented in Figure~\ref{code-desc-scx}.
\begin{thesisonly}
Recall that each \sct\ creates a new descriptor called an \op, which has two mutable fields: a 2-bit \textit{state} field and an \textit{allFrozen} bit.
The \textit{state} field contains one of three values: \textit{InProgress}, \textit{Committed} and \textit{Aborted}.
\end{thesisonly}

\begin{thesisnot}
\llt\ and \sct\ operate on multi-field \textit{data records}, which can be used to represent, e.g., nodes in a tree, or records in a table.
Like descriptors, \textit{data records} contain mutable and immutable fields.
However, whereas descriptors are used only to facilitate helping, and are not part of a sequential data structure, data records are.

\llt$(r)$ attempts to take a snapshot of the mutable fields of a \rec\ $r$.
If it is concurrent with an \sct\ involving~$r$, it may return \fail, instead.
Individual fields of a \rec\ can also be read directly.
An \sct$(V,R,fld,new)$ takes as its arguments a sequence $V$ of \rec s, a subsequence $R$ of $V$, a pointer $fld$ to a mutable field of one \rec\ in~$V$, and a new value $new$ for that field.
The \sct\ tries to atomically: store the value $new$ in the field that $fld$ points to and {\it finalize} each \rec\ in $R$.
Once a \rec\ is finalized, its mutable fields cannot be changed by any subsequent \sct, and any \llt\ of the \rec\ will return \finalized\ instead of a snapshot.

Before a process invokes \sct, it must perform an \llt$(r)$ on each \rec\ $r$ in $V$.
The last such \llt\ by the process is said to be {\it linked} to the \sct, and the linked \llt\ must return a snapshot of $r$ (not \fail\ or \finalized).
An \sct($V, R, fld, new$) by a process modifies the data structure and returns \true\ only if no \rec\ $r$ in $V$ has changed since its linked \llt($r$); otherwise the \sct\ fails and returns \false.
Although \llt\ and \sct\ can fail, their failures are limited in such a way that they can be used to build data structures with lock-free progress.
See \cite{Brown:2013} for a more formal specification of these primitives.

Each \sct\ operation creates a new descriptor called an \op.
\llt\ and \sct\ requires each \rec\ $r$ to have a dedicated field $r.\info$ that stores a pointer to an \op, and this field is only ever accessed by \llt\ and \sct\ operations.
Each \rec\ also has a \textit{marked} bit which is accessed only by \llt\ and \sct.
This field is used by \sct\ to finalize \rec s.
We say that a \rec\ is \textit{marked} if its \textit{marked} bit is set.
\op s have two mutable fields: a 2-bit \textit{state} field and an \textit{allFrozen} bit.
The \textit{state} field contains one of three values: \textit{InProgress}, \textit{Committed} and \textit{Aborted}.
\end{thesisnot}

The following properties of the \llt\ and \sct\ algorithm are relevant for our purposes.
\begin{compactenum}[P1.]
	\item Before the first invocation of \func{Help}$(des)$ for an \sct\ $O$ (performed by $O$ or a helper) has been completed, the \op\ $des$ created by $O$ has its \textit{state} field set to \textit{Committed} or \textit{Aborted}, and, after this, the \textit{state} field of $des$ is never changed again.
    \item A marked \rec\ remains marked forever.
    \item A marked \rec\ cannot point to an \op\ with $\textit{state} = \textit{Aborted}$.
    \item Each time the $\info$ field of a \rec\ changes, it changes to a new value that has never previously been stored there (to avoid the ABA problem).
\end{compactenum}

There is only one place in the code where an invocation $I$ of $\func{ReadField}(\op, d,$ $f, dv)$ can occur outside of \func{Help}$(des)$: at line~4
of \llt\ in Figure~\ref{code-desc-scx}.
$I$ reads the \textit{state} field of $d$.
We choose the default value $dv = \textit{Committed}$ for $I$.
We give a rigorous, but straightforward, proof that this extended transformation of \llt\ and \sct\ is correct.

Let $e$ be an execution of the original \llt\ and \sct\ algorithm $\mathcal{A}$, and let $e'$ be an execution that is the same as $e$, except that one arbitrary \op\ $d$ becomes invalid at some point $t$ after the \sct\ operation attempt $O$ that created $d$ terminates.
Let $O'$ be any \llt\ in $e$ which, after $t$, performs an invocation $I$ of \func{ReadField} on $d$ outside of \func{Help}$(d)$.
We must prove that $O'$ performs the exact same changes in $e$ and $e'$ to any variables that are still defined after $O'$ terminates, and returns the same response in both executions.

Since $I$ is invalid in $e'$, by definition, the \sct\ $O$ that created $d$ must have terminated before $I$.
Thus, by P1, $I$ must return \textit{Committed} or \textit{Aborted} in $e$.
If $I$ returns \textit{Committed} in $e$, then $I$ returns the same response in $e$ and $e'$, so $O'$ is exactly the same in both executions.
Now, suppose $I$ returns \textit{Aborted} in $e$.
We consider three cases, depending on where $O'$ returns in $e$.

\textit{Case 1:} $O'$ returns at line~10 
in $e$.
If $marked_2 = \false$, then $O'$ behaves exactly the same way in $e$ and $e'$.
So, suppose $marked_2 = \true$.
Then, $O'$ will enter the if-statement at line~6 
 in $e$, but not in $e'$.
In this case, $O'$ saw that the \rec\ $r$ pointed to an \op\ with $\textit{state} = \textit{Aborted}$ when it performed line~3, 
and that $r$ was marked when it performed line~5. 
By P3, $r$ cannot simultaneously be marked and point to an \op\ with $\textit{state} = \textit{Aborted}$, so $r.\info$ must change between these two lines.
By P4, it must change to a value different from $r\info$, so the if-block at line~8 
will not be executed in $e$.
However, this contradicts our assumption that $O'$ returns at line~10. 

\textit{Case 2:} $O'$ returns \finalized\ at line~12 
in $e$.
Observe that $O'$ does not execute line~9 
in $e$ (since it would then return at the following line). 
We first prove that $O'$ does not execute line~9 
in $e'$.
Since $O'$ sees $marked_1 = \true$ just before returning at line~12 
in $e$, P2 implies that $marked_2 = \true$ (in both $e$ and $e'$).
Since $I$ returns \textit{Committed} in $e'$, $O'$ will not enter the if-block at line~6 
in $e'$.
Thus, $O'$ reaches line~11 
in both $e$ and $e'$.

Since $I$ returns \textit{Committed} in $e'$, and we have assumed $I$ returns \textit{Aborted} in $e$, $O'$ will not invoke \func{Help} at line~11 
in $e$ or $e'$.
Therefore, $O'$ does not change any variable that is still defined after it terminates.
So, it suffices to prove that $O'$ returns \finalized\ (at line~12) 
in $e'$.
However, this is immediate from the fact that $marked_1 = \true$ in $O'$ in $e$ (and, hence, in $e'$).

\textit{Case 3:} $O'$ returns \fail\ at line~13 
in $e$.
The proof is similar to the previous case, except $marked_1 = \false$ in $O'$ in $e$, so when $O'$ reaches line~12, 
it will enter the \textit{else}-block and return \fail\ in both $e$ and $e'$.

%
%
%

\section{Implementing the extended weak descriptor ADT} \label{sec-extended-impl}
%
%
%
We give a brief high-level overview of our implementation. 
It uses largely known techniques (similar to~\cite{Marathe:2008}), and is not the main contribution of this work.
Each process $p$ uses a \textit{single} descriptor object $D_{T,p}$ in shared memory to represent \textit{all} descriptors of type $T$ that it ever creates.
The descriptor object $D_{T,p}$ conceptually represents $p$'s \textit{current} descriptor of type $T$.
At different times in an execution, $D_{T,p}$ represents different \textit{abstract descriptors} created by $p$.
We store a sequence number in $D_{T,p}$ that is incremented every time $p$ performs \func{CreateNew}$(T, -)$. 
Instead of using traditional descriptor pointers, we represent each descriptor pointer as a pair of fields stored in a single word.
These fields contain the name of the process who owns the descriptor, and a sequence number that indicates which invocation of \func{CreateNew} conceptually created this descriptor.
When a descriptor pointer is passed to an operation $O$ on the abstract descriptor, $O$ compares the sequence number in $des$ with the current sequence number in $D_{T,p}$ to determine whether the operation is valid or invalid.
Thus, incrementing the sequence number in $D_{T,p}$ effectively makes all abstract descriptors of type $T$ that were previously created by $p$ \textit{invalid}.
Mutable fields are stored in a single word alongside a sequence number, so they can be updated with CAS, preventing invalid operations from making changes.
(If the mutable fields and a sequence number cannot fit in one word, then one can use multiple words and attach the sequence number to each word.)

\begin{figure}[ph!]
\begin{lstlisting}[name=extendedweak,frame=single]
 //\com \textbf{Data types}
 //Descriptor of type $T:$
   $mutables = \langle seq, mut_1, mut_2, ... \rangle$ //\hfill\com Mutable fields
   $imm_1, imm_2, ... $ //\hfill\com Immutable fields

 //\com \textbf{Shared variables}
   $D_{T,p}$ //for each descriptor type $T$ and process $p$

 //\com \textbf{ADT operations}
 //$\func{CreateNew}(T, v_1, v_2, ...)$ by process $p:$
   $oldseq$ := $D_{T,p}.mutables.seq$
   $D_{T,p}.mutables.seq$ := $oldseq+1$ //\label{code-impl-create-new-lin}
   for each //field $f$ \textbf{in} $D_{T,p}$
     let $value$ //be the corresponding value in $\{v_1, v_2, ...\}$
     if $f \mbox{ is immutable}$ then
       $D_{T,p}.f := value$
     else 
       $D_{T,p}.mutables.f := value$
   $D_{T,p}.mutables.seq$ := $oldseq+2$ //\label{code-extended-weak-before-fence}
   return $\langle p, oldseq+2 \rangle$ //\label{code-extended-weak-after-fence}
   
 //$\func{ReadField}(des, f, dv):$
   $\langle q, seq \rangle := des$
   if $f \mbox{ is immutable}$ then
     $result := D_{T,q}.f$ //\label{code-impl-read-field-immutable-lin}
   else
     $result := D_{T,q}.mutables.f$ //\label{code-impl-read-field-mutable-lin}
   if $seq \neq D_{T,q}.mutables.seq$ then return $dv$ //\label{code-impl-read-field-invalid-lin}
   return $result$ //\label{code-impl-read-field-valid-return}

 //$\func{ReadImmutables}(des):$
   $\langle q, seq \rangle := des$
   for each $f$ in $des$
     if $f \mbox{ is immutable}$ then //add $D_{T,q}.f$ to $result$
   if $seq \neq D_{T,q}.mutables.seq$ then return $\bot$ //\label{code-impl-read-immutables-lin}
   return $result$//\label{code-impl-read-immutables-valid-return}

 //$\func{WriteField}(des, f, value):$
   $\langle q, seq \rangle := des$
   loop
     $exp := D_{T,q}.mutables$ //\label{code-impl-write-field-read}
     if $exp.seq \neq seq$ then return //\label{code-impl-write-field-invalid-lin}
     $new := exp$
     $new.f := value$
     if $\func{CAS}(\&D_{T,q}.mutables, exp, new)$ then return//\label{code-impl-write-field-success-lin}

 //$\func{CASField}(des, f, \textit{fexp}, \textit{fnew}):$
   $\langle q, seq \rangle := des$
   loop
     $exp := D_{T,q}.mutables$ //\label{code-impl-cas-field-read}
     if $exp.seq \neq seq$ then return $\bot$ //\label{code-impl-cas-field-invalid-lin}
     if $exp.f \neq \textit{fexp}$ then return $exp.f$ //\label{code-impl-cas-field-failed-lin}
     $new := exp$
     $new.f := \textit{fnew}$
     if $\func{CAS}(\&D_{T,q}.mutables, exp, new)$ then //\label{code-impl-cas-field-success-lin}
       return $fnew$ //\label{code-impl-cas-field-valid-return}
\end{lstlisting}
\vspace{-3mm}
\caption{Pseudocode for the \textit{extended weak descriptor} ADT implementation.}
\label{code-reuse-impl}
\end{figure}

\subsection{Detailed description}
Complete pseudocode 
appears in Figure~\ref{code-reuse-impl}.
We start by describing the data types and shared variables.
Each descriptor contains zero or more immutable fields, and zero or more mutable fields (which are determined by the \textit{descriptor type}), as well as a sequence number field \textit{seq}.
Recall that $D_{T,p}$ represents different abstract descriptors at different times. 
Note that the immutable fields of $D_{T,p}$ are only immutable for as long as $D_{T,p}$ represents the same abstract descriptor.
When $D_{T,p}$ is reused, so that it represents a different abstract descriptor, its immutable fields can be reinitialized.
Usually very few bits are required for the mutable fields, since they exist solely to capture the state of an ongoing operation (and it is inefficient to frequently change the state of a descriptor).
(Every lock-free algorithm we are aware of uses at most a small constant number of bits for its mutable fields.)
Consequently, we think of the sequence field and the mutable fields of a descriptor $d$ as being packed together in a single word $mutables$ of $d$ (with subfields for the sequence field and each mutable field).
(Note that, if more space is needed for mutable fields in some future algorithm, we can eliminate this assumption about the size of $mutables$, as we explain below.)
We use $d.f$ to denote an immutable field $f$ of $d$, $d.mutables$ to denote the field $mutables$ of $d$, and $d.mutables.f$ to denote a mutable field $f$ of $d$.

Since $mutables$ fits in a single word, it can be modified atomically using CAS.
By having CAS atomically operate on a mutable field and the sequence number, 
we can ensure that a descriptor changes only if its sequence number has not changed.

We now describe the operations.
An invocation of \func{CreateNew}$(T, ...)$ by process $p$ first increments the sequence number of $D_{T,p}$, then initializes all of its fields, then increments the sequence number again and returns a new descriptor pointer (with the up-to-date sequence number).
Observe that the descriptor pointers returned by \func{CreateNew} always have even sequence numbers, and the sequence number of a descriptor is odd while it is being initialized by \func{CreateNew}.
Consequently, while a descriptor is being initialized, its sequence number does not match any descriptor pointer in the system, so no process can read or modify the descriptor's fields.

Note that this approach of incrementing a sequence number twice has been used in different contexts such as in transactional memory, where the least significant bit represents whether the sequence number is locked or unlocked.
Here, the idea is slightly different, since the least significant bit represents whether the descriptor is currently being reused and initialized, or is safe to access.
(Nevertheless, in some sense, one can think of the bit indicating whether the descriptor is currently being initialized as a sort of lock.
It does not prevent other processes from making progress (since operations on the descriptor will terminate, but will simply be invalid), but it does prevent them from accessing fields of the descriptor as they are being changed.)

An invocation of \func{ReadField}$(des, f, \textit{default})$ by $p$ 
reads the value $v$ of the mutable or immutable field $f$ from $D_{T,p}$ followed by its sequence number $s$.
If $s$ matches the sequence number in the descriptor pointer $des$, then $v$ is returned.
Otherwise, \textit{default} is returned.

\func{ReadImmutables} is similar to \func{ReadField}, except it reads all immutable fields, instead of a single field, and it returns $\bot$ instead of \textit{default}.

An invocation $I$ of \func{WriteField}$(des, f, value)$ by $p$ performs a sequence of one or more \textit{attempts}.
In each attempt, it reads the contents $old$ of $mutables$, including the sequence number $s$, from $D_{T,p}$, then checks whether $s$ matches the sequence number in the descriptor pointer $des$.
If the sequence numbers do not match, then the abstract descriptor represented by $des$ is invalid, so $I$ returns without changing $f$.
Otherwise, $I$ uses CAS to try to change $D_{T,p}.mutables$ from $old$ to $new$, which is a copy of $old$ in which the contents of field $f$ have been changed (locally) to contain $value$.
Observe that this CAS will succeed only if the sequence number in $D_{T,p}.mutables$ matches the sequence number in $des$.
If the CAS succeeds, then $I$ returns.
Otherwise, $I$ performs another attempt. 

Note that \func{WriteField} is less efficient than performing a direct write to memory.
However, since mutable fields are used merely to encode the status of an ongoing operation, there are usually very few changes to a descriptor.


\func{CASField} is quite similar to \func{WriteField}. 
The only differences are (1) \func{CASField} has different return values and, (2) in each attempt, it performs an additional check to determine whether $old.f$ is equal to $\textit{fexp}$, and, if not, returns $old.f$.

\subsection{Practical considerations}
One might wonder, in an algorithm with multiple types of descriptors, why the type of a descriptor is not also encoded in descriptor pointers.
In algorithms that use multiple descriptor types, any time the original algorithm accesses a field of a descriptor, it typically must know what kind of descriptor it is accessing (if, for no other reason, to compute the address of the desired field within the descriptor).
In such algorithms, it would not be necessary for descriptor pointers to carry this extra information.
For algorithms that access descriptors without knowing their exact types, one can include the descriptor type in descriptor pointers.

	Some lock-free algorithms ``steal'' up to three bits from pointers to encode additional information, typically to distinguish between application values and (potentially, various types of) descriptors.~
To accommodate such algorithms, one can slightly shrink the sequence number in our descriptor pointers, and reserve the three lowest-order bits for use by other algorithms. 

One obvious way to store the descriptors for each thread is to create an array for each descriptor type, with a slot containing a descriptor for each process.
In this kind of implementation, it is extremely important to pad each slot to avoid false sharing~\cite{scott1993false}.
We suggest allocating at least two cache lines for each descriptor (128 bytes on modern Intel and AMD machines).

To improve efficiency, modern Intel and AMD processors implement a relaxed memory model called total store order (TSO) that allows certain steps in a program to be executed out of order.
Specifically, a read that occurs after a write in a program can actually be executed \textit{before} the write, as long as the read and write are not accessing the same address.
This can render a concurrent algorithm incorrect if it requires a write by a process $p$ to be visible to other processes \textit{before} $p$ performs a subsequent read.
One can prevent this reordering by placing a memory fence (or barrier) between the write and read.
CAS instructions also act as memory fences.
Our implementation does not require any memory fences (beyond those implied by CAS instructions).
This is an attractive property, since memory fences incur high overhead.


Our implementation uses unbounded sequence numbers.
However, in practice, sequence numbers are bounded, and they may wrap around.
If wraparound occurs, then two invocations of \func{CreateNew} might return the same descriptor pointer.
This can cause an \textit{ABA problem} if the high-level algorithm that uses descriptors relies on the uniqueness of descriptor pointers returned by \func{CreateNew}.

We argue that the sequence number can be made sufficiently large on modern systems for this to be a non-issue.
Consider a system with a 64-bit word size.
Recall that a sequence number appears both in each descriptor pointer, and also in the $mutables$ field of each descriptor.
A descriptor pointer contains only a process name and a sequence number, so if $n$ bits are reserved for the process name, then $64-n$ bits remain for the sequence number.
The $mutables$ field contains the descriptor's mutable fields and a sequence number, so if $m$ bits are reserved for mutable fields, then $64-m$ bits remain for the sequence number.
Thus, if we use 14-bit process names (as the Linux kernel does), and the mutable fields of each descriptor fit in at most 14 bits, then 50 bits remain for the sequence number.
We are unaware of any algorithm that requires more than three bits for mutable fields in its descriptors, so this is realistic.
In this case, a single process must perform $2^{50}$ operations to trigger even a single wraparound.
If we assume that a single process can perform one million operations per second, this will take 35 years of continuous execution.
If this is still a concern, then one can use double-wide CAS (DWCAS), which is implemented on modern Intel and AMD systems, instead of CAS, to atomically operate on two adjacent words (containing a much larger sequence number).

Although we are unaware of any current lock-free algorithms that use more than three bits for mutable fields in descriptors, some future algorithm may use more.
If the mutable fields of a descriptor cannot fit in the same word as a sequence number, then our approach must be modified.
If the mutable fields and a sequence number can fit in two adjacent words, then one can simply use DWCAS instead of CAS.
Otherwise, one can store mutable fields in their own separate words, and \textit{replicate} the sequence number, storing a copy in the word adjacent to each mutable field.
To change a mutable field, one would then perform DWCAS on the word containing the mutable field, and its adjacent sequence number.
When the descriptor is reused, instead of incrementing a single sequence number, one would increment all sequence numbers.


In order to choose how many bits should be devoted to the process name in descriptor pointers, one must know an upper bound on the number of processes.
We stress that this is not an onerous constraint, because the upper bound does not need to be tight.
Note that one need not initially allocate descriptors for all processes that \textit{could} be running in the system.
It is straightforward to allocate a descriptor for a process the first time it invokes \func{CreateNew} (potentially even in batches, to amortize the cost and improve control over memory layout).

\newpage

\subsection{Correctness}
We now prove that our implementation 
is linearizable. 
We first give the linearization points for all operations. 
\begin{itemize}
\item Each invocation of \func{CreateNew} is linearized at the increment of the sequence number at line~\ref{code-impl-create-new-lin}.
\item If an invocation $I$ of \func{ReadField}$(des, f, dv)$ returns at line~\ref{code-impl-read-field-invalid-lin}, then it is linearized at the read of the sequence number at the same line.
If $I$ returns at line~\ref{code-impl-read-field-valid-return}, then it is linearized at the preceding read of the field $f$: for immutable fields this is line~\ref{code-impl-read-field-immutable-lin}, and for mutable fields this is line~\ref{code-impl-read-field-mutable-lin}.
\item Each invocation of \func{ReadImmutables} is linearized at the read of the sequence number at line~\ref{code-impl-read-immutables-lin}.
\item If an invocation $I$ of \func{WriteField}$(des, f, value)$ returns at line~\ref{code-impl-write-field-invalid-lin}, then it is linearized at the last read of the sequence number at the same line.
If $I$ returns at line~\ref{code-impl-write-field-success-lin}, then it
is linearized at the successful CAS at the same line.
\item If an invocation $I$ of \func{CASField}$(des, f, \textit{fexp}, \textit{fnew})$ returns at line~\ref{code-impl-cas-field-invalid-lin}, then it is linearized at the last read of the sequence number at the same line.
If $I$ returns at line~\ref{code-impl-cas-field-valid-return}, then it is linearized at the successful CAS at the previous line.
If $I$ returns at line~\ref{code-impl-cas-field-failed-lin}, then it
is linearized at the last read at the same line.
\end{itemize}

\begin{Observation}
	The sequence number of $D_{T,p}$ (also denoted $D_{T,p}.mutables.seq$) is written only by $p$ in invocations of \func{CreateNew}($T,-$).\label{impl-observ-write-to-seq}
\end{Observation}

\begin{Observation}
	Every descriptor pointer returned by \func{CreateNew} has an even sequence number, and the linearization point of \func{CreateNew} always changes the sequence number of the descriptor to an odd number. \label{impl-observ-even-vs-odd}
\end{Observation}

\begin{Observation}
	The sequence number returned by a \\
	\func{CreateNew}($T,-$) operation by p is $2+v$ where $v$ is the sequence number returned by $p$'s previous \func{CreateNew}($T,-$) operation, or $v = 0$ if $p$ has not performed \func{CreateNew}($T,-$).\label{impl-observ-unique-des}
\end{Observation}

We now prove that the above linearization points are correct. 
Let $e$ be an execution of our implementation of extended weak descriptors.
Let $O_1,O_2\cdots O_k$ be the extended weak descriptor operations executed in $e$ in the order they are linearized. 
\begin{shortver}
Note that we prove correctness assuming unbounded sequence numbers.
The implications of bounded sequence numbers were considered above. 
\end{shortver}

\begin{theorem}
	The responses of $O_1,O_2\cdots O_k$ respect the semantics of the extended weak descriptor ADT.
\end{theorem}  
  
\begin{proof}
By strong induction on the sequence of extended weak descriptor operations that terminate in $e$.
Base case: the claim vacuously holds when no operations have returned.
%
%
%
Induction step: assume the return values of $O_{1}, O_{2}\cdots O_{i-1}$ follow the semantics of the extended weak descriptor ADT. 
Let $p$ be the process that performs $O_i$, and $T$ be the type of descriptor on which $O_i$ is performed.

Suppose $O_i$ is a \func{CreateNew}($T,-$) operation.
By Observation~\ref{impl-observ-unique-des}, 
$O_i$ returns a unique descriptor pointer.


In each of the following cases, $O_i$ takes a descriptor pointer $des$ as one of its arguments.
Let $q$ and $seq$ be the process name sequence number in $des$, respectively.
Let $O_{init}$ be the \func{CreateNew}($T,-$) by $q$ that returned $des$. 
Since $des$ is returned by $O_{init}$ before it is passed to any operation, $O_{init}$ is linearized before $O_i$. 

Suppose $O_i$ is a \func{ReadField} that returns the default value at line~\ref{code-impl-read-field-invalid-lin}, a \func{CASField} that returns $\bot$ at line~\ref{code-impl-cas-field-invalid-lin} or a \func{ReadImmutables} that returns $\bot$ at line~\ref{code-impl-read-immutables-lin}.
We argue that $des$ is invalid when $O_i$ is linearized.
In each case, $O_i$ returns after seeing that the sequence number of $D_{T,q}$ no longer contains $seq$.
Thus, this sequence number must change after $des$ is returned by $O_{init}$, and before $O_i$ is linearized. 
By Observation~\ref{impl-observ-write-to-seq}, this change to the sequence number of $D_{T,q}$ must be performed by a \func{CreateNew}$(T, -)$ operation $O_{change}$ by $q$ (which occurs after $O_{init}$, and before $O_i$ is linearized).
$O_{change}$ changes the sequence number twice, and is linearized at the first change.
Thus, $O_{change}$ is linearized after $O_{init}$ and before $O_i$, which means that $des$ is \textit{invalid} when $O_i$ is linearized.
%
%

Now suppose $O_i$ is a \func{ReadField} that returns at line~\ref{code-impl-read-field-valid-return}, a \func{CASField} that returns at line~\ref{code-impl-cas-field-valid-return}, or a \func{ReadImmutables} that returns at line~\ref{code-impl-read-immutables-valid-return}.
We first argue that $des$ is valid when $O_i$ is linearized.
In each case, $O_i$ sees that the sequence number of $D_{T,q}$ is $seq$ at some time $t$, which is either when $O_i$ is linearized, or is after $O_i$ is linearized.
By Observation~\ref{impl-observ-write-to-seq} and Observation~\ref{impl-observ-unique-des}, whenever the sequence number of $D_{T,q}$ is changed, it is changed to a new value that it never previously contained.
Thus, since the sequence number of $D_{T,q}$ contains $seq$ when $O_{init}$ terminates, and it contains $seq$ at time $t$, it contains $seq$ at all times after $O_{init}$ terminates and before $t$.
Hence, the sequence number of $D_{T,q}$ contains $seq$ when $O_i$ is linearized.
By Observation~\ref{impl-observ-write-to-seq}, $q$ does not perform any \func{CreateNew}$(T, -)$ after $O_{init}$ and before $t$, so $des$ is \textit{valid} when $O_i$ is linearized.

We now argue that the response of $O_i$ is correct if it is a \func{ReadField}$(des, f, dv)$.
The proof is similar when $O_i$ is a \func{CASField} or \func{ReadImmutables}.

%

If $f$ is immutable, then it is changed only by \func{CreateNew}$(T, -)$ operations by $q$.
Since $des$ is valid when $O_i$ is linearized, $O_{init}$ performs the last change to $f$ before $O_i$ is linearized.
Recall that $O_i$ start after $O_{init}$ terminates.
Thus, the write of $f$ in $O_{init}$ happens before the invocation of $O_i$, and $O_i$ will return the value written to $f$ by $O_{init}$. 

If $f$ is mutable, then let $O_{change}$ be the operation that performs the last change to $f$ before $O_i$ is linearized, and $v$ be the value that it stores in $f$.
Observe that $O_i$ returns $v$.
We show that $O_{change}$ is the last operation that changes $f$ and is linearized before $O_i$.
If $O_{change}$ is the same as $O_{init}$, then we are done.
Otherwise, since we have argued that $q$ does not perform any \func{CreateNew}$(T, -)$ after $O_{init}$ and before $t$, $O_{change}$ must be a \func{WriteField} or \func{CASField}.
In each case, $O_{change}$ can change $f$ only once, with a successful CAS (at line~\ref{code-impl-cas-field-success-lin} or line~\ref{code-impl-write-field-success-lin}).
Since $O_{change}$ is linearized at this CAS, it is linearized before $O_i$.
Moreover, since we have assumed that $O_{change}$ is the last operation to change $f$ before $O_i$, no other operation that changes $f$ linearized after $O_{change}$ and before $O_i$.
\end{proof} 

\subsection{Progress}
%
Suppose, to obtain a contradiction, that there is an execution in which processes take infinitely many steps, but only finitely many (extended weak descriptor) operations terminate.
Then, after some time $t$, no operation terminates, which means there is at least one operation $O$ in which a process takes infinitely many steps.
By inspection of Figure~\ref{code-reuse-impl}, $O$ must be a \func{WriteField} or \func{CASField} operation.
Suppose $O$ is a \func{WriteField} operation.
Then, each time $O$ executes line~\ref{code-impl-write-field-invalid-lin}, it sees $old.seq = seq$, and each time it executes line~\ref{code-impl-write-field-success-lin}, its CAS fails and returns $old$ without changing $D_{T,q}.mutables$.
Observe that the CAS will fail only if $D_{T,q}.mutables$ changes after it is read at line~\ref{code-impl-write-field-read} and before the CAS at line~\ref{code-impl-write-field-success-lin}.
Thus, $D_{T,q}.mutables$ changes infinitely many times in the execution.
Since $D_{T,q}.mutables$ can be changed only by \func{WriteField} or \func{CASField} operations, and any operation that changes $D_{T,q}.mutables$ terminates, there must be infinitely many terminating \func{WriteField} or \func{CASField} operations, which is a contradiction.
The proof is similar when $O$ is a \func{CASField}.

\smallskip
\section{Experiments} \label{sec-exp}


Our experiments were run on two large-scale systems.
The first is a 2-socket Intel E7-4830 v3 with 12 cores per socket and 2 hyperthreads (HTs) per core, for a total of 48 threads.
Each core has a private 32KB L1 cache and 256KB L2 cache (which is shared between HTs on a core).
All cores on a socket share a 30MB L3 cache.
The second is a 4-socket AMD Opteron 6380 with 8 cores per socket and 2 HTs per core, for a total of 64 threads.
Each core has a private 16KB L1 data cache and 2MB L2 cache (which is shared between HTs on a core).
All cores on a socket share a 6MB L3 cache.

Since both machines have multiple sockets and a non-uniform memory architecture (NUMA), in all of our experiments, we pinned threads to cores so that the first socket is filled first, then the second socket is filled, and so on.
Furthermore, within each socket, each core has one thread pinned to it before hyperthreading is engaged.
Consequently, our graphs clearly show the effects of hyperthreading and NUMA.

For example, on the Intel machine, from thread counts 1 to 12 all threads are running on a single socket and at most one thread is pinned to each core. \textbf{(socket 1: no HTs; socket 2: empty)}.
From 13 to 24, all threads are running on a single socket and cores either have one or two threads pinned to them \textbf{(socket 1: HTs; socket 2: empty)}.
From 25 to 36, each core on the first socket has two threads pinned to it, and the remaining threads are each pinned to unique cores on the second socket \textbf{(socket 1: HTs; socket 2: no HTs)}.
Finally, from 37 to 48, each core on the first socket has two threads pinned to it, and cores on the second socket have one or two threads pinned to them \textbf{(socket 1: HTs; socket 2: HTs)}.

Both machines have 128GB of RAM.
Each runs Ubuntu 14.04 LTS.
All code was compiled with the GNU C++ compiler (G++) 4.8.4 with build target x86\_64-linux-gnu and compilation options \texttt{-std=c++0x -mcx16 -O3}.
Thread support was provided by the POSIX Threads library.
~We used the Performance Application Programming Interface (PAPI) library~\cite{Browne:2000} to collect statistics from hardware performance counters to determine cache miss rates, stall times, instructions retired, and so on.

The system (glibc) allocator was found to have poor scaling and overall performance.
Instead, we used jemalloc 4.2.1, a fast user-space allocator designed to minimize contention and improve scalability~\cite{Evans:2006}.
The library was dynamically linked with \texttt{LD\_PRELOAD}, which is the recommended method.
This allocator was found to yield vastly superior performance for all algorithms, in all benchmarks.
We also tried the tcmalloc allocator from Google's Perftools library, which is another common choice for concurrency-friendly allocation.
However, performance with tcmalloc was substantially worse for all algorithms than with jemalloc.

On the AMD machine, transparent huge-pages were disabled manually in the jemalloc implementation by changing the default allocation chunk size from $2^{21}$ to $2^{19}$ using the environment parameter setting \texttt{MALLOC\_CONF=lg\_chunk:19}.
This maintained or improved the performance for all algorithms in all workloads, and did not change the performance relationship between any pair of algorithms.
The same change did not improve performance on the Intel machine (for any algorithm or workload), so the original chunk size was used.

For read-heavy workloads, it was necessary to force distribution of pages across NUMA nodes to get consistently high performance.
To achieve this, we used \texttt{numactl --interleave=all} for all workloads.
(Doing this did not negatively impact the performance of any workload, but its benefit was less noticeable for write-heavy workloads.)


\subsection{$k$-CAS microbenchmark} \label{sec-exp-kcas}
In order to compare our reusable descriptor technique with algorithms that reclaim descriptors, we implemented $k$-CAS with several memory reclamation schemes.
Specifically, we implemented a lock-free memory reclamation scheme that aggressively frees memory called \textit{hazard pointers}~\cite{Michael:2004}, a (blocking) epoch-based reclamation scheme called \textit{DEBRA}~\cite{Brown:2015}, and reclamation using the read-copy-update (RCU) primitives~\cite{Desnoyers:2012} (also blocking).
We use \textit{Reuse} as shorthand for our reusable descriptor based algorithm, and \textit{DEBRA}, \textit{HP} and \textit{RCU} to denote the 
\begin{shortver}
other algorithms.
\end{shortver}
\begin{fullver}
algorithms that use DEBRA, hazard pointers and RCU, respectively.
\end{fullver}

The paper by Harris~et~al. also describes an optimization to reduce the number of DCSS descriptors that are allocated by embedding them in the $k$-CAS descriptor.
We applied this optimization, and found that it did not significantly improve performance.
Furthermore, it complicated reclamation with hazard pointers.
Thus, we did not use this optimization.

\fakeparagraph{Methodology}
We compared our implementations of $k$-CAS using a simple array-based microbenchmark.
For each algorithm $A \in \{$\textit{Reuse}, \textit{DEBRA}, \textit{HP}, \textit{RCU}$\}$, array size $S \in \{2^{14}, 2^{20}, 2^{26}\}$ and $k$-CAS parameter $k \in \{2, 16\}$, we run ten timed \textit{trials} for several thread counts $n$.
In each trial, an array of a fixed size $S$ is allocated and each entry is initialized to zero.
Then, $n$ concurrent threads run for one second, during which each thread repeatedly chooses $k$ uniformly random locations in the array, reads those locations, and then performs a $k$-CAS (using algorithm $A$) to increment each location by one.

As a way of validating correctness in each trial, each thread keeps track of how many successful $k$-CAS operations it performs.
At the end of the trial, the sum of entries in the array must be $k$ times the total number of successful $k$-CAS operations over all threads.

\begin{shortver}
\begin{figure}[t]
    \centering
    \setlength\tabcolsep{0pt}
    \begin{minipage}{0.49\linewidth}
    \begin{tabular}{m{0.05\linewidth}m{0.465\linewidth}m{0.465\linewidth}}
        &
        \multicolumn{2}{c}{
            \fcolorbox{black!80}{black!40}{\parbox{\dimexpr 0.93\linewidth-2\fboxsep-2\fboxrule}{\centering\textbf{2x 24-thread Intel E7-4830 v3}}}
        }
        \\
        &
        \fcolorbox{black!50}{black!20}{\parbox{\dimexpr \linewidth-2\fboxsep-2\fboxrule}{\centering {\footnotesize 2-CAS}}} &
        \fcolorbox{black!50}{black!20}{\parbox{\dimexpr \linewidth-2\fboxsep-2\fboxrule}{\centering {\footnotesize 16-CAS}}}
        \\
        \rotatebox{90}{Array size $2^{26}$} &
        \includegraphics[width=\linewidth]{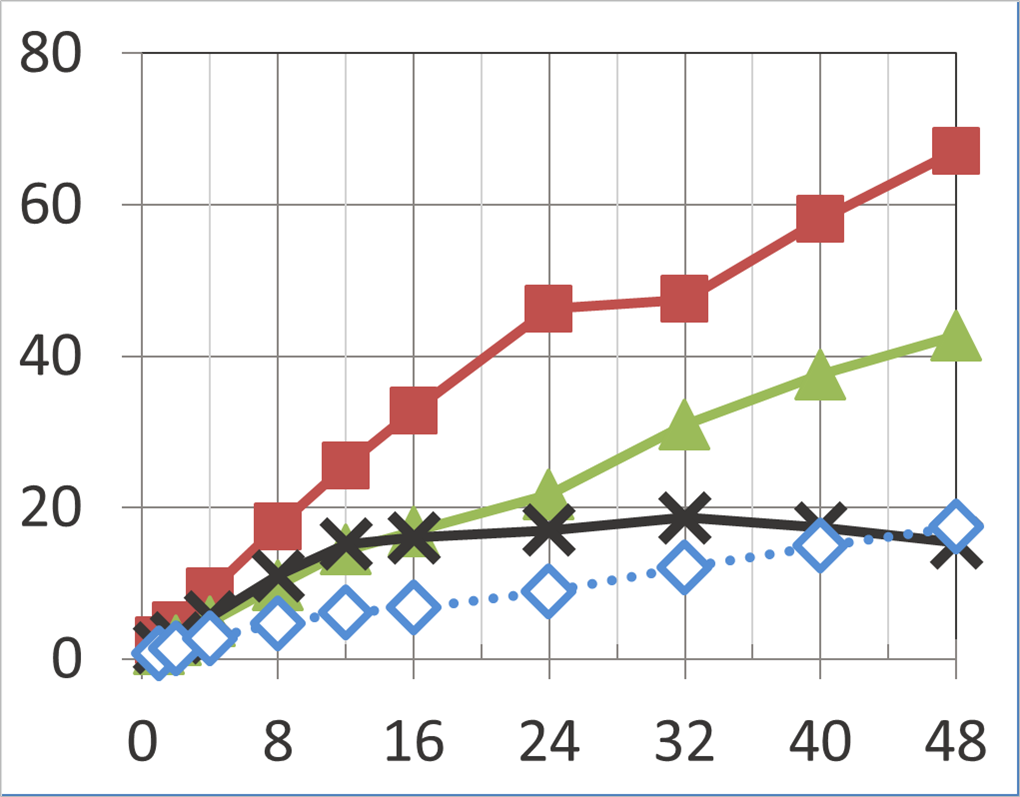} &
        \includegraphics[width=\linewidth]{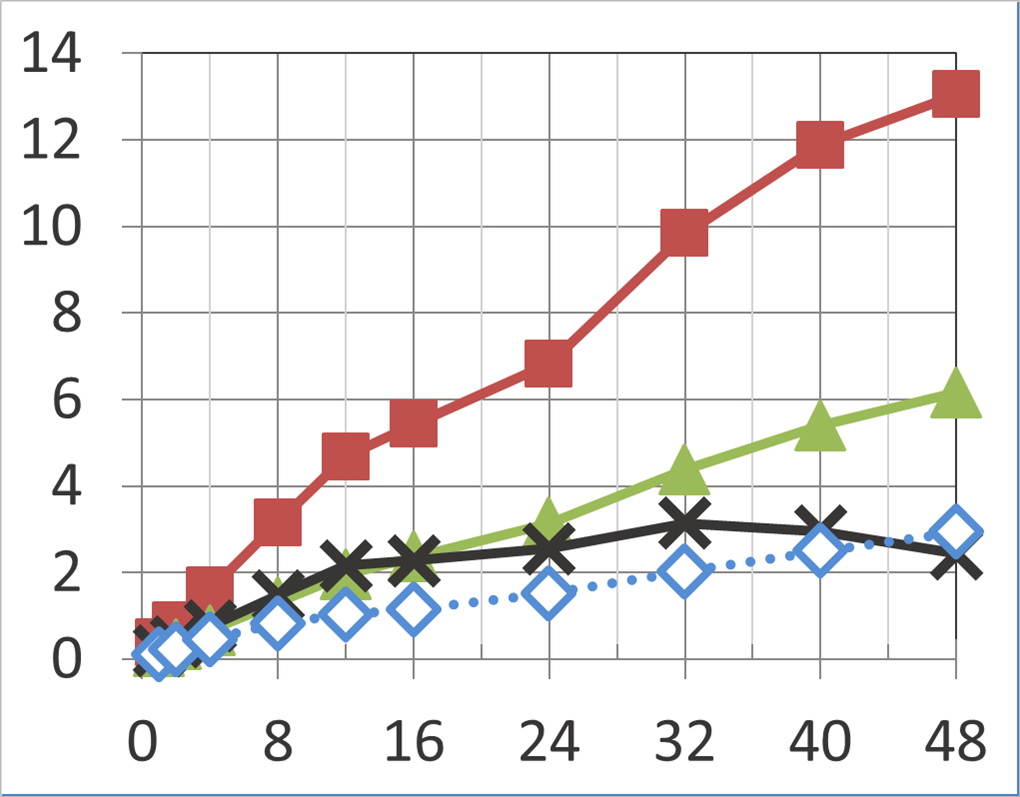}
        \\
        \vspace{-1.5mm}\rotatebox{90}{Array size $2^{20}$} &
        \vspace{-1.5mm}\includegraphics[width=\linewidth]{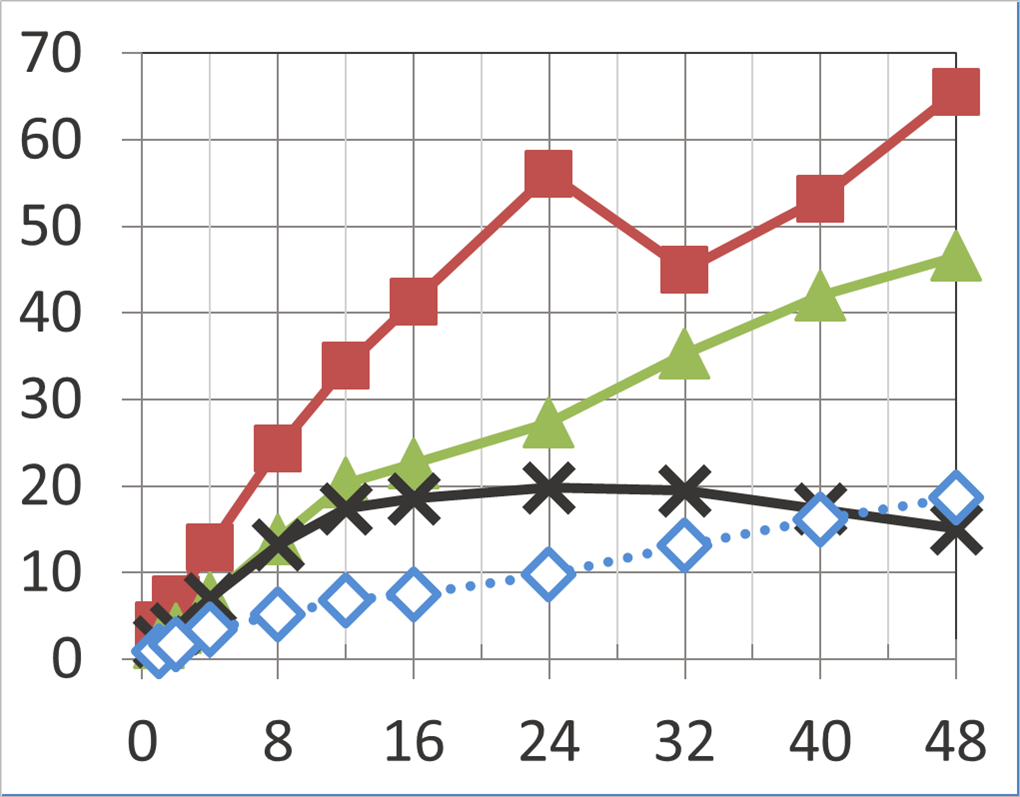} &
        \vspace{-1.5mm}\includegraphics[width=\linewidth]{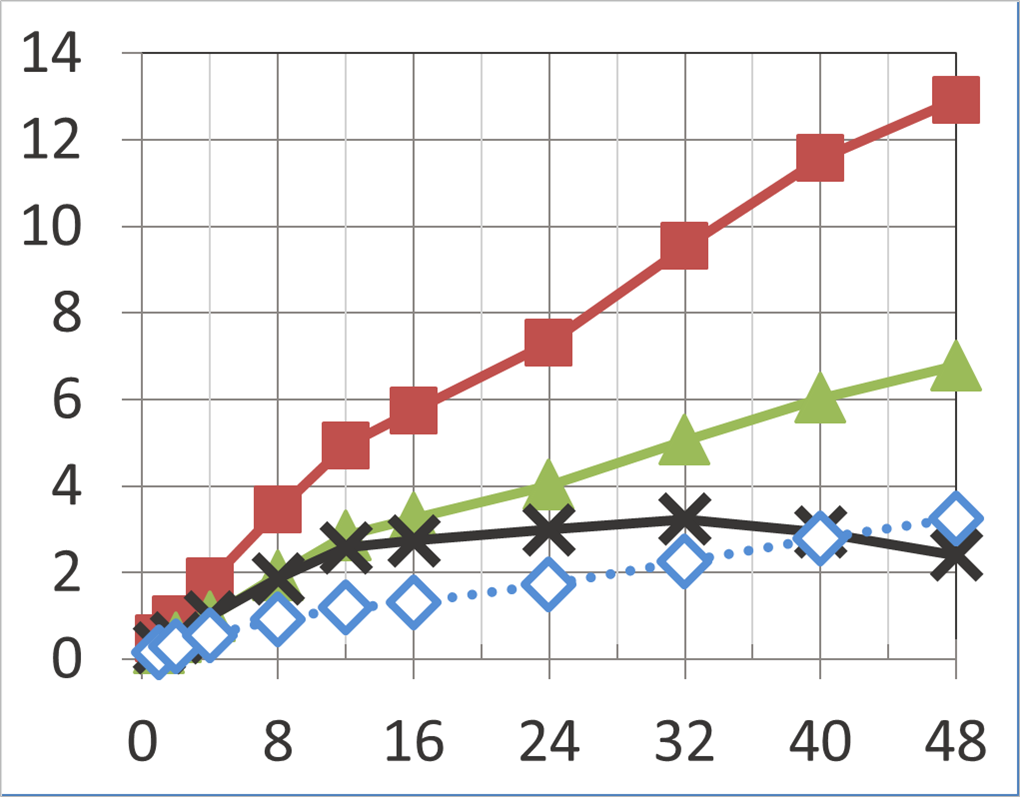}
        \\
        \vspace{-1.5mm}\rotatebox{90}{Array size $2^{14}$} &
        \vspace{-1.5mm}\includegraphics[width=\linewidth]{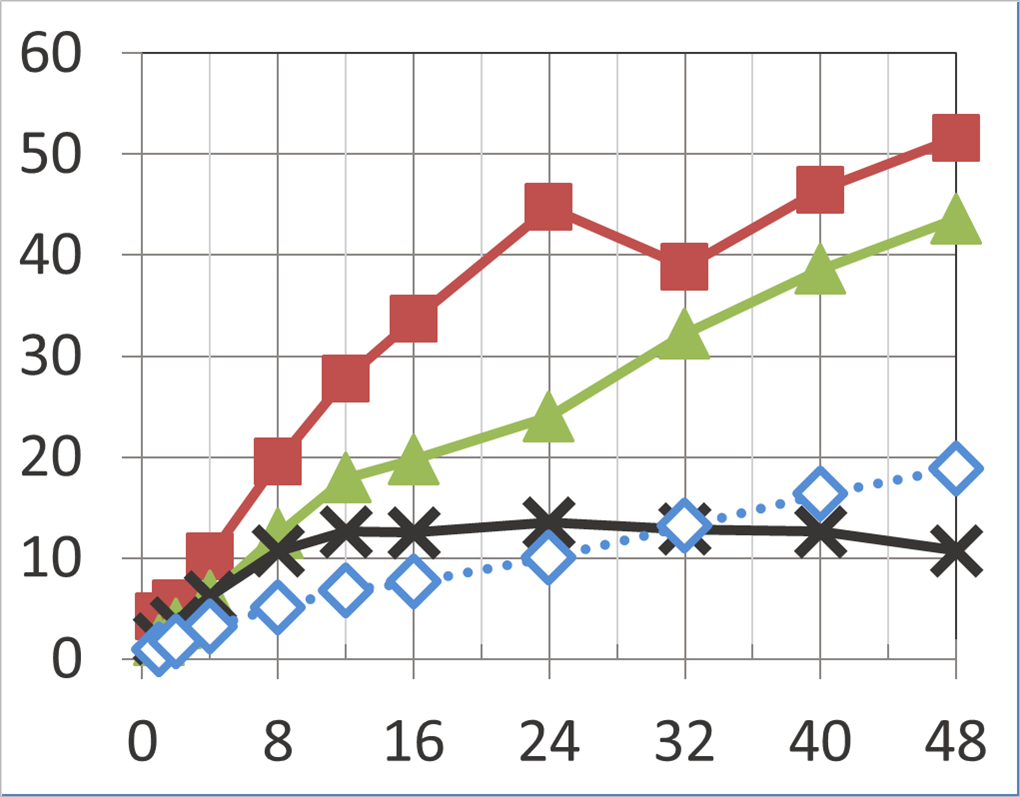} &
        \vspace{-1.5mm}\includegraphics[width=\linewidth]{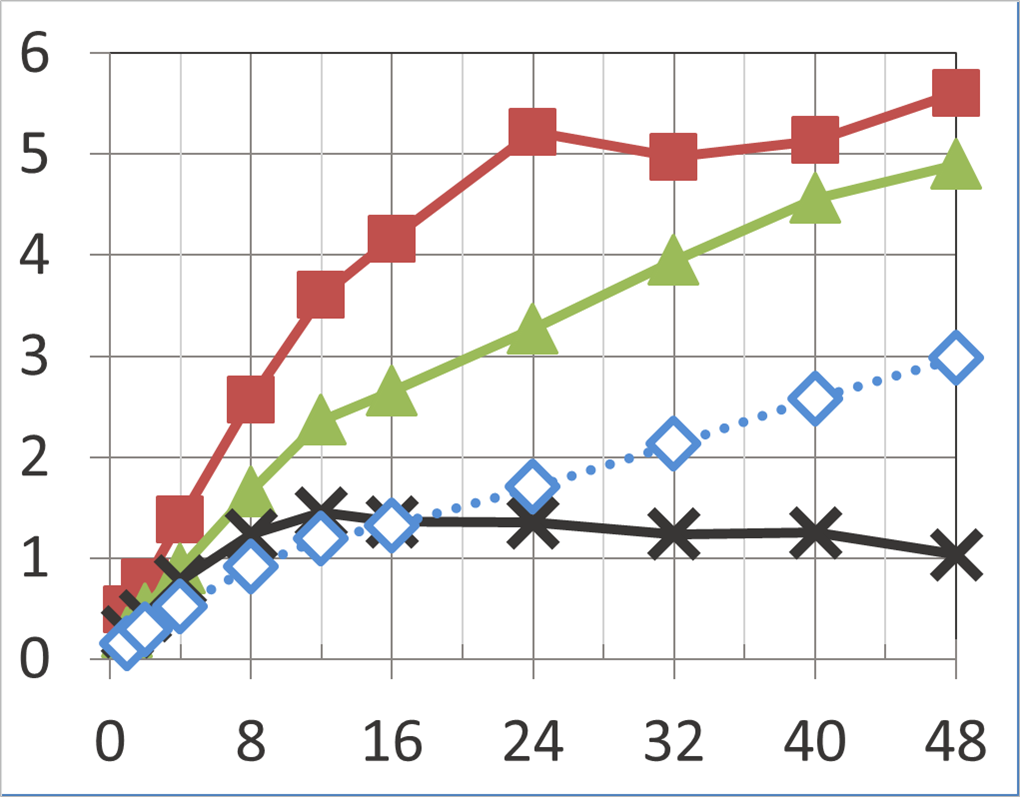}
        \\
    \end{tabular}
    \end{minipage}
    \begin{minipage}{0.49\linewidth}
    \begin{tabular}{m{0.05\linewidth}m{0.465\linewidth}m{0.465\linewidth}}
        &
        \multicolumn{2}{c}{
            \fcolorbox{black!80}{black!40}{\parbox{\dimexpr 0.93\linewidth-2\fboxsep-2\fboxrule}{\centering\textbf{4x 16-thread AMD Opteron 6380}}}
        }
        \\
        &
        \fcolorbox{black!50}{black!20}{\parbox{\dimexpr \linewidth-2\fboxsep-2\fboxrule}{\centering {\footnotesize 2-CAS}}} &
        \fcolorbox{black!50}{black!20}{\parbox{\dimexpr \linewidth-2\fboxsep-2\fboxrule}{\centering {\footnotesize 16-CAS}}}
        \\
        \rotatebox{90}{Array size $2^{26}$} &
        \includegraphics[width=\linewidth]{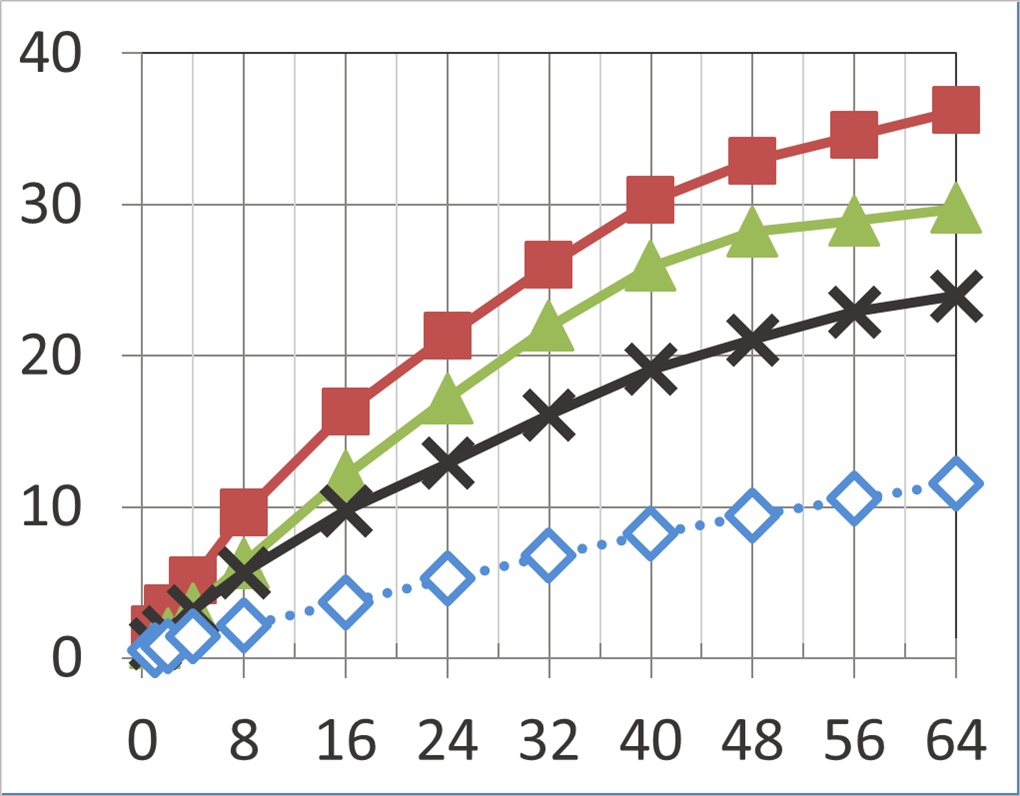} &
        \includegraphics[width=\linewidth]{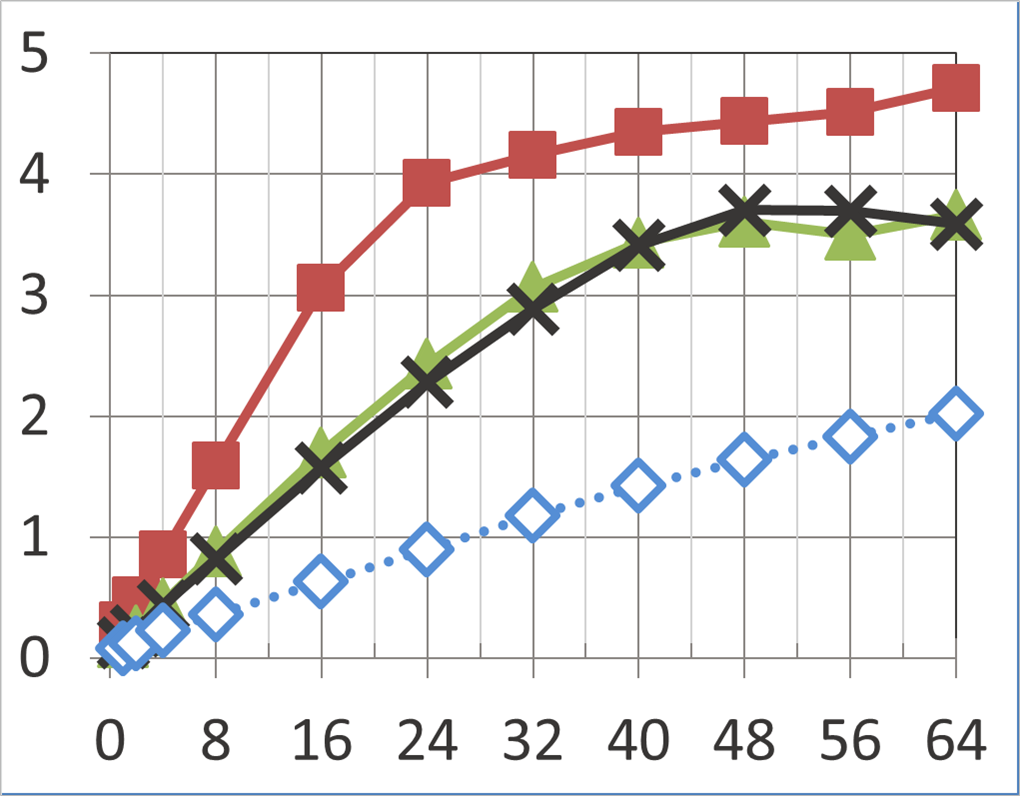}
        \\
        \vspace{-1.5mm}\rotatebox{90}{Array size $2^{20}$} &
        \vspace{-1.5mm}\includegraphics[width=\linewidth]{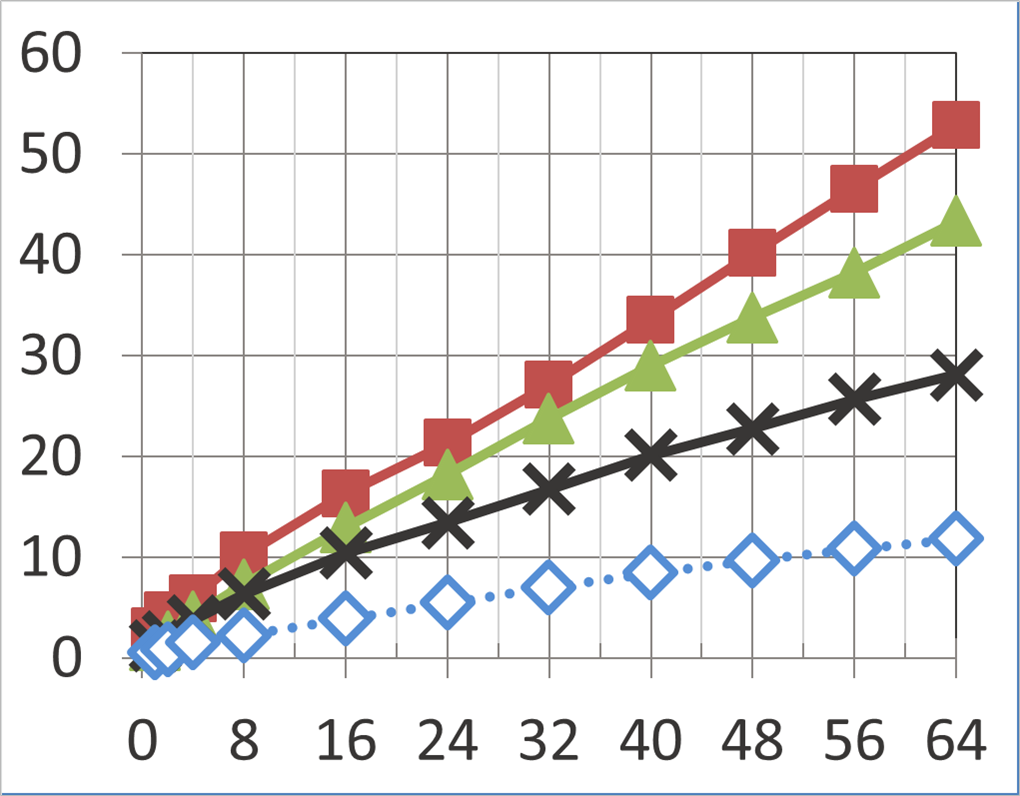} &
        \vspace{-1.5mm}\includegraphics[width=\linewidth]{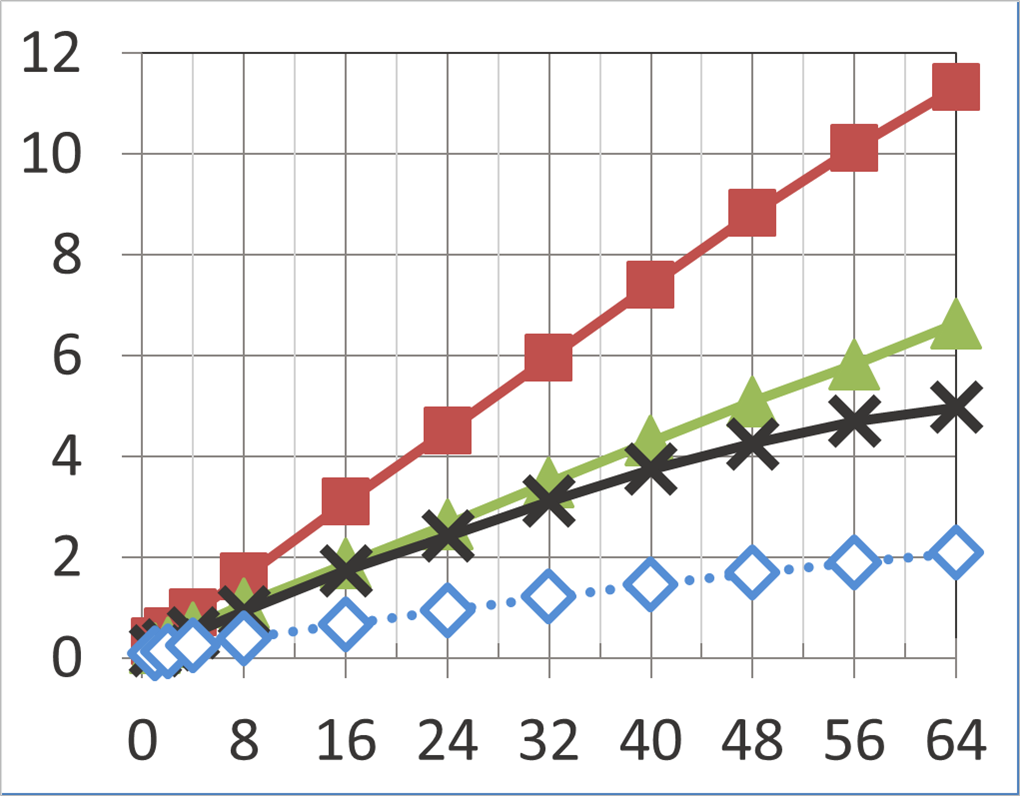}
        \\
        \vspace{-1.5mm}\rotatebox{90}{Array size $2^{14}$} &
        \vspace{-1.5mm}\includegraphics[width=\linewidth]{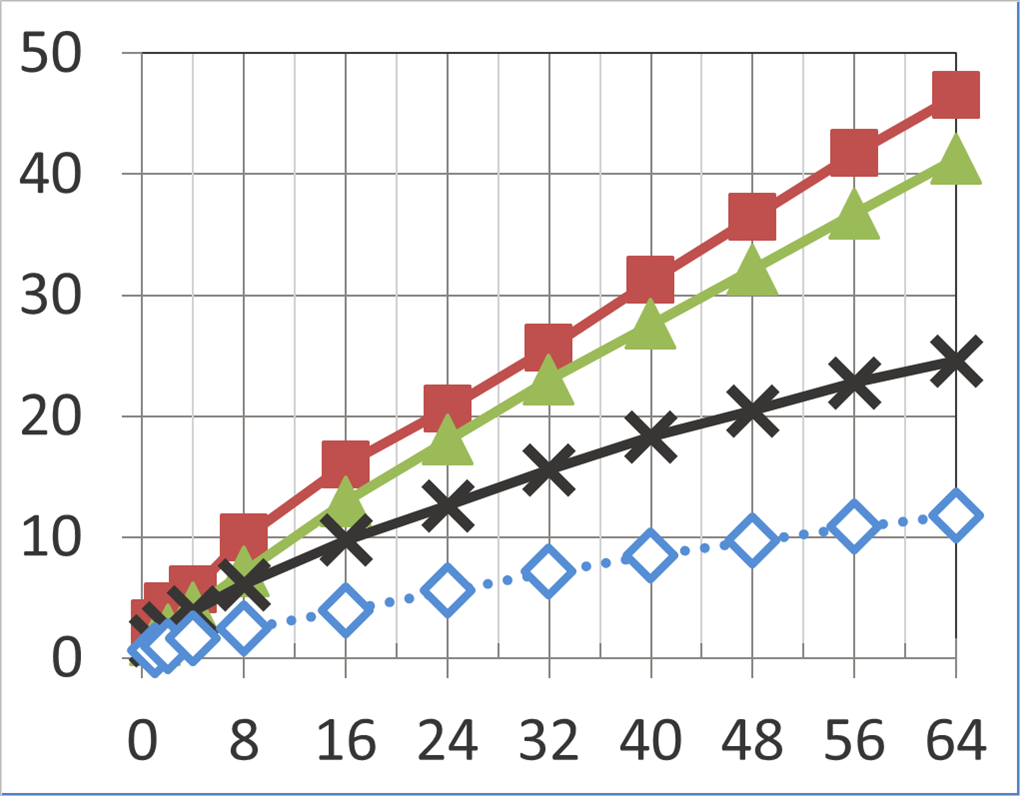} &
        \vspace{-1.5mm}\includegraphics[width=\linewidth]{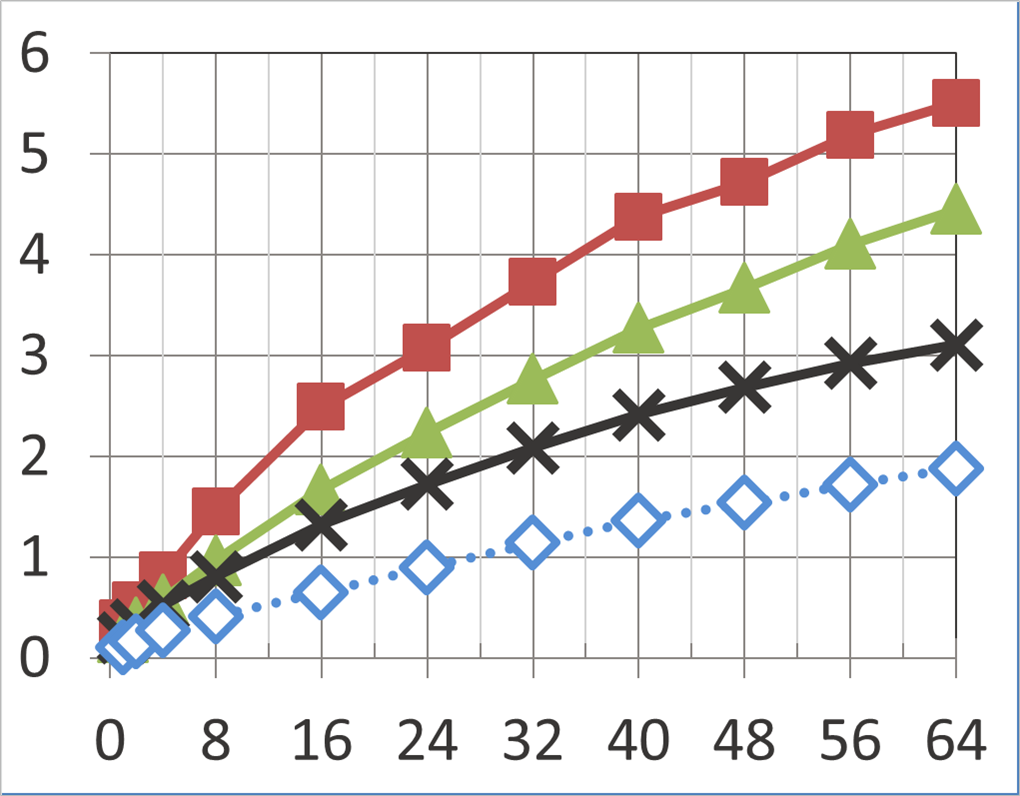}
        \\
    \end{tabular}
    \end{minipage}
\includegraphics[width=0.5\linewidth]{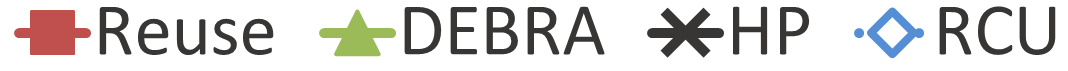}
\vspace{-2mm}
\caption{Results for a \textbf{$k$-CAS microbenchmark}.
The x-axis represents the number of concurrent threads.
The y-axis represents operations per microsecond.}
\label{fig-exp-kcas}
\end{figure}
\end{shortver}

\begin{fullver}
\begin{figure}[t]
    \centering
    \setlength\tabcolsep{0pt}
    \begin{tabular}{m{0.05\linewidth}m{0.31\linewidth}m{0.31\linewidth}m{0.31\linewidth}}
        &
        \multicolumn{3}{c}{
            \fcolorbox{black!80}{black!40}{\parbox{\dimexpr 0.93\linewidth-2\fboxsep-2\fboxrule}{\centering\textbf{2x 24-thread Intel E7-4830 v3}}}
        }
        \\
        &
        \fcolorbox{black!50}{black!20}{\parbox{\dimexpr \linewidth-2\fboxsep-2\fboxrule}{\centering {\footnotesize Array size $2^{26}$}}} &
        \fcolorbox{black!50}{black!20}{\parbox{\dimexpr \linewidth-2\fboxsep-2\fboxrule}{\centering {\footnotesize Array size $2^{20}$}}} &
        \fcolorbox{black!50}{black!20}{\parbox{\dimexpr \linewidth-2\fboxsep-2\fboxrule}{\centering {\footnotesize Array size $2^{14}$}}}
        \\
        \rotatebox{90}{2-CAS} &
        \includegraphics[width=\linewidth]{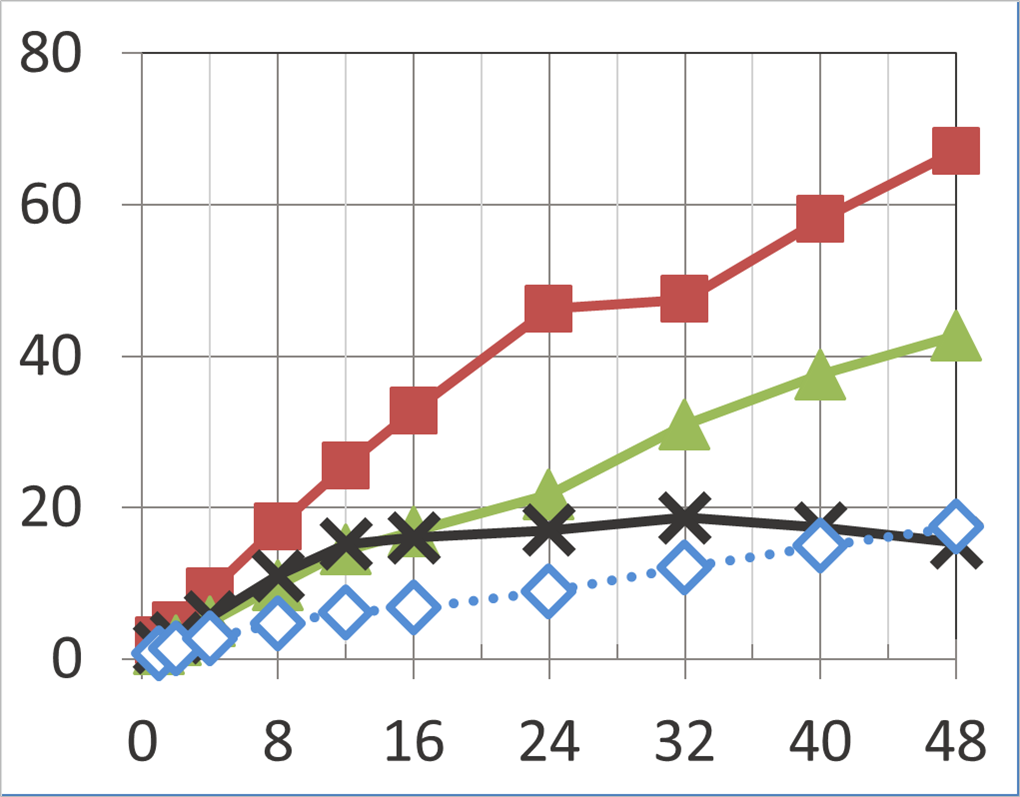} &
        \includegraphics[width=\linewidth]{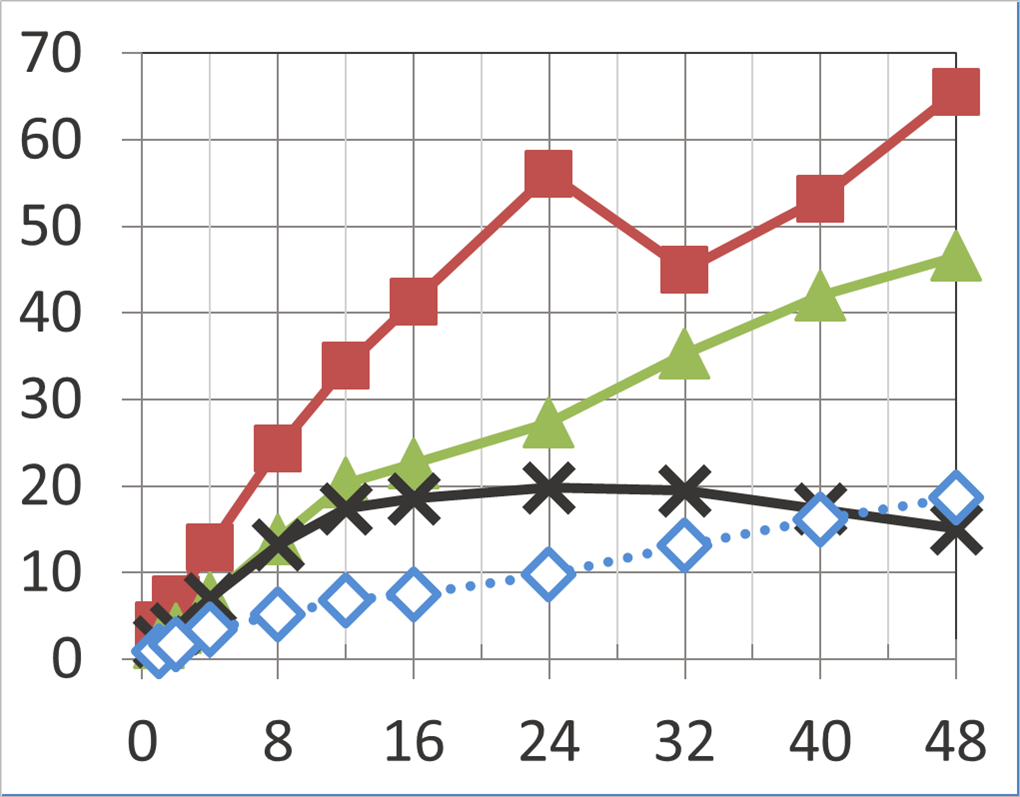} &
        \includegraphics[width=\linewidth]{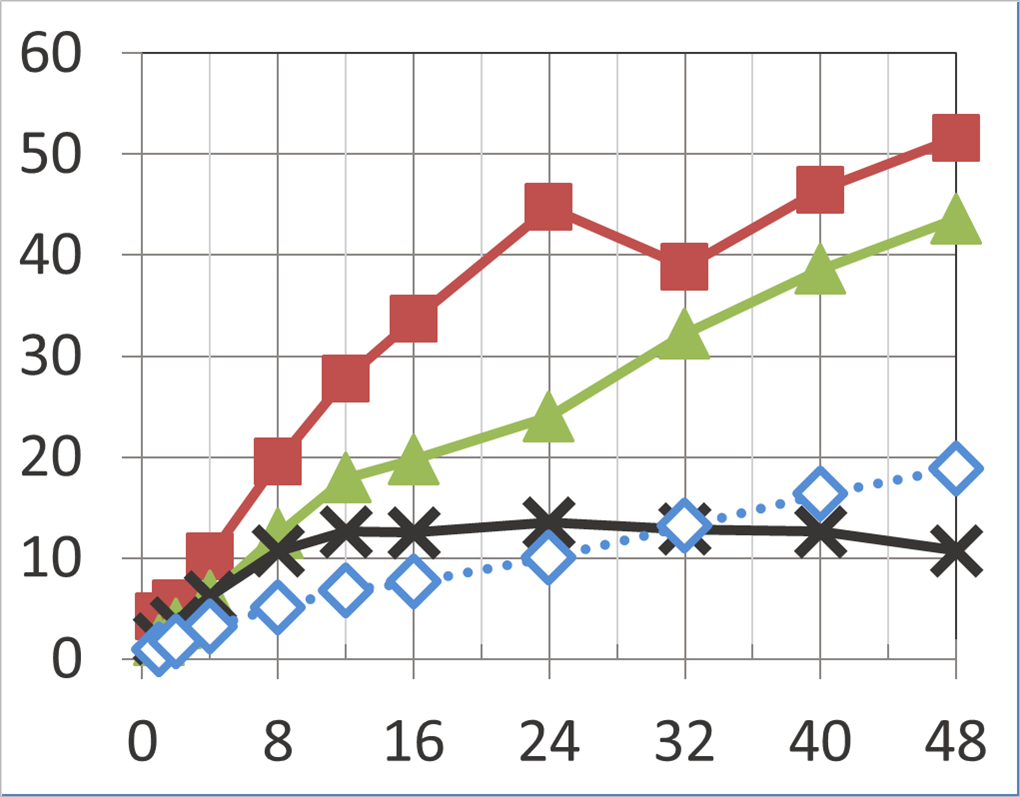}
        \\
        \rotatebox{90}{16-CAS} &
        \includegraphics[width=\linewidth]{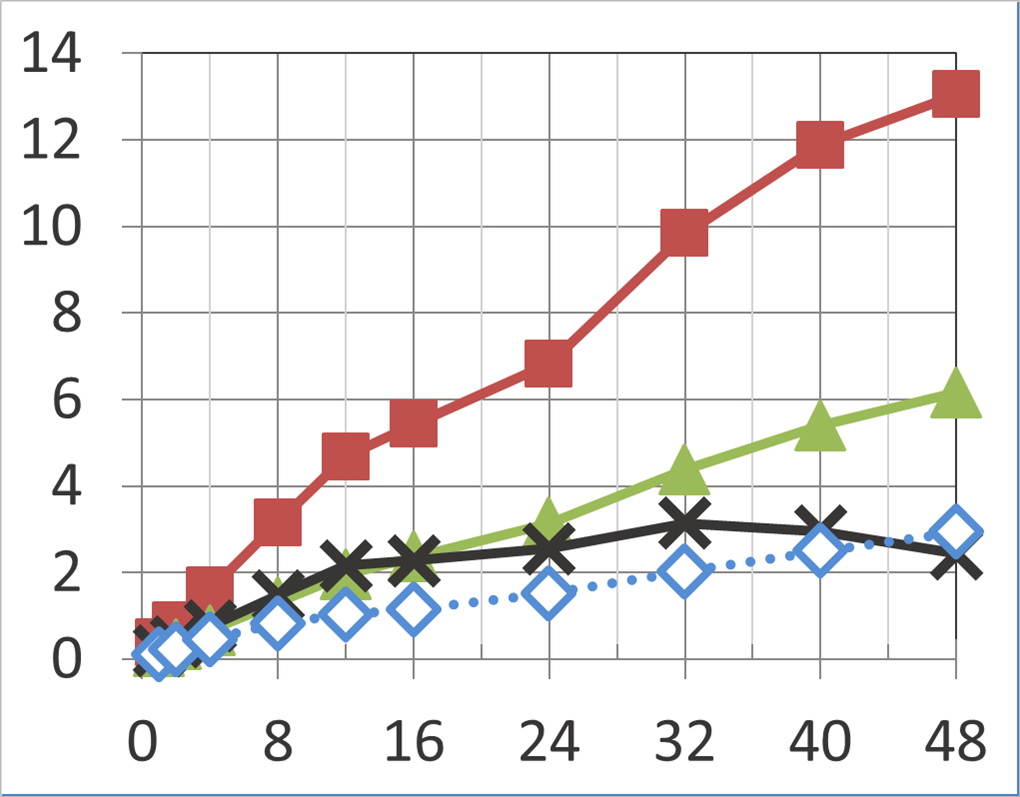} &
        \includegraphics[width=\linewidth]{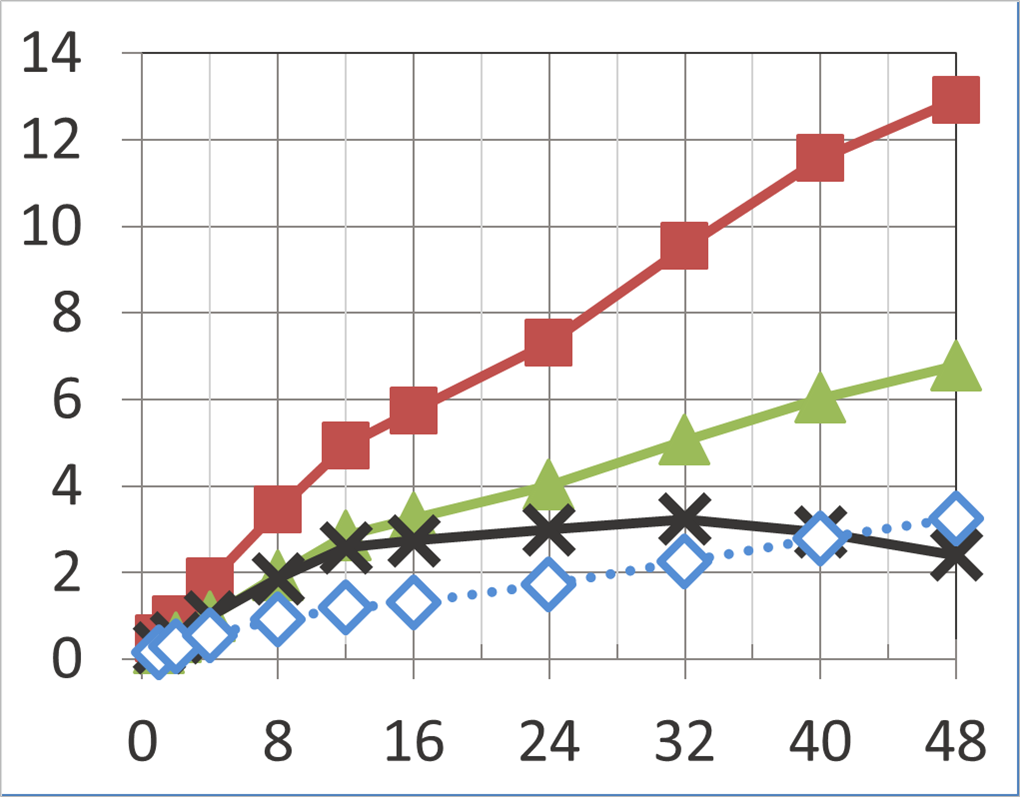} &
        \includegraphics[width=\linewidth]{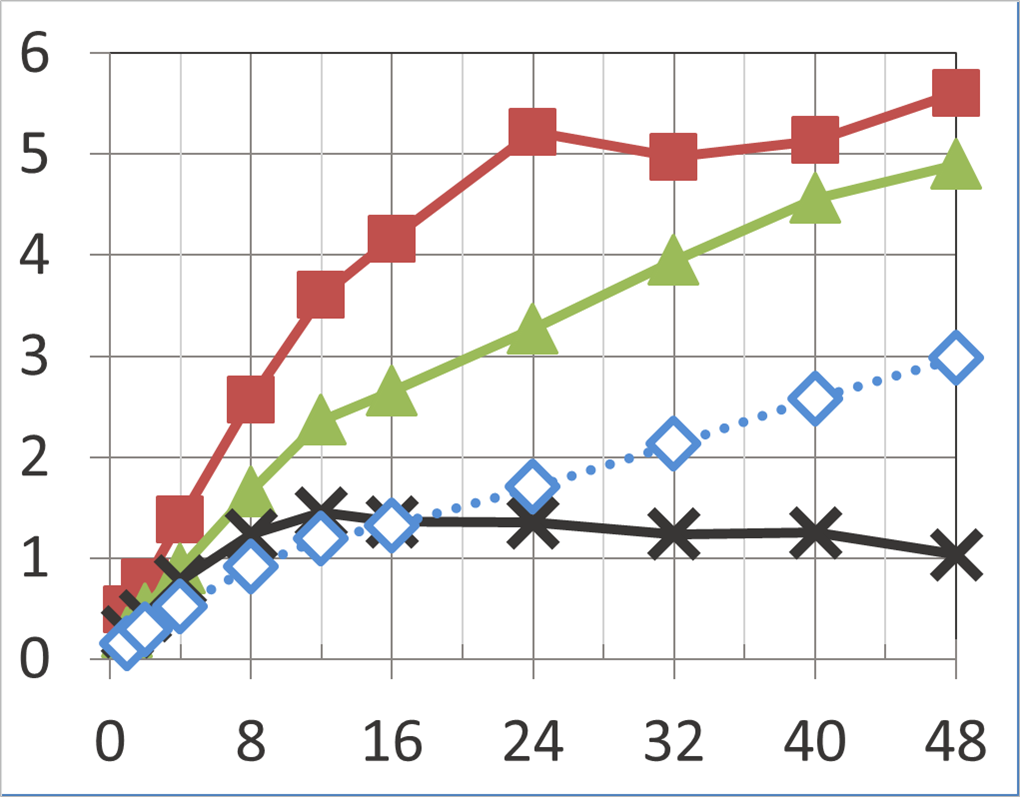}
        \\
    \end{tabular}
    \begin{tabular}{m{0.05\linewidth}m{0.31\linewidth}m{0.31\linewidth}m{0.31\linewidth}}
        &
        \multicolumn{3}{c}{
            \fcolorbox{black!80}{black!40}{\parbox{\dimexpr 0.93\linewidth-2\fboxsep-2\fboxrule}{\centering\textbf{4x 16-thread AMD Opteron 6380}}}
        }
        \\
        &
        \fcolorbox{black!50}{black!20}{\parbox{\dimexpr \linewidth-2\fboxsep-2\fboxrule}{\centering {\footnotesize Array size $2^{26}$}}} &
        \fcolorbox{black!50}{black!20}{\parbox{\dimexpr \linewidth-2\fboxsep-2\fboxrule}{\centering {\footnotesize Array size $2^{20}$}}} &
        \fcolorbox{black!50}{black!20}{\parbox{\dimexpr \linewidth-2\fboxsep-2\fboxrule}{\centering {\footnotesize Array size $2^{14}$}}}
        \\
        \rotatebox{90}{2-CAS} &
        \includegraphics[width=\linewidth]{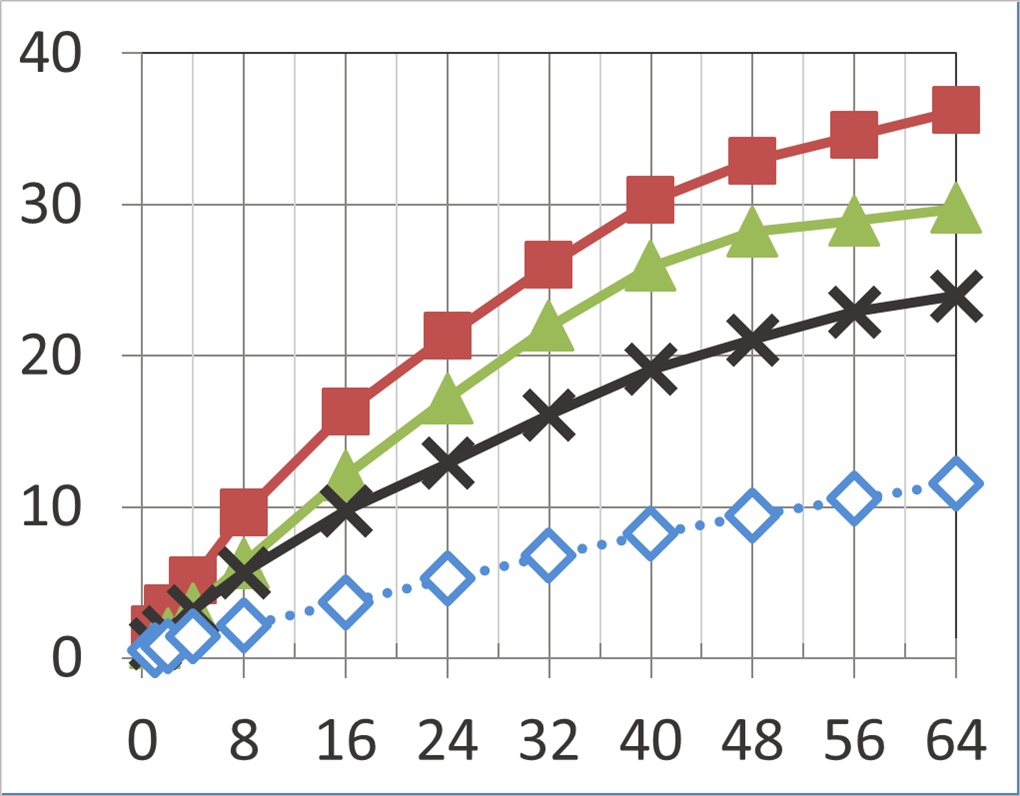} &
        \includegraphics[width=\linewidth]{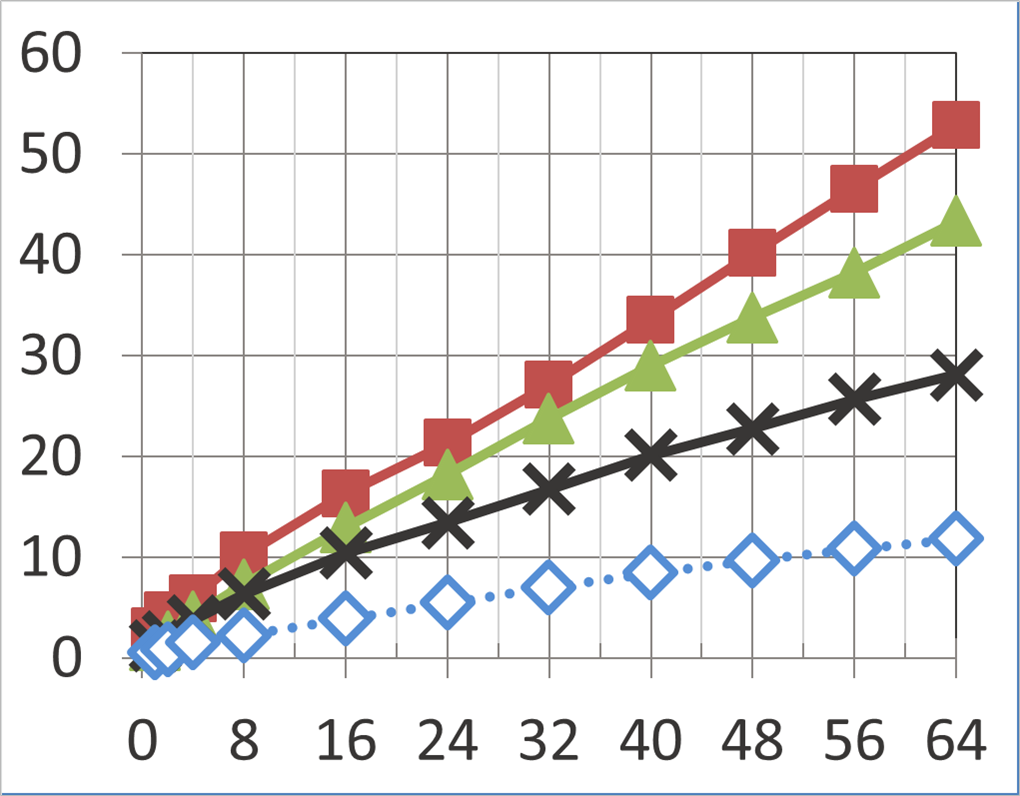} &
        \includegraphics[width=\linewidth]{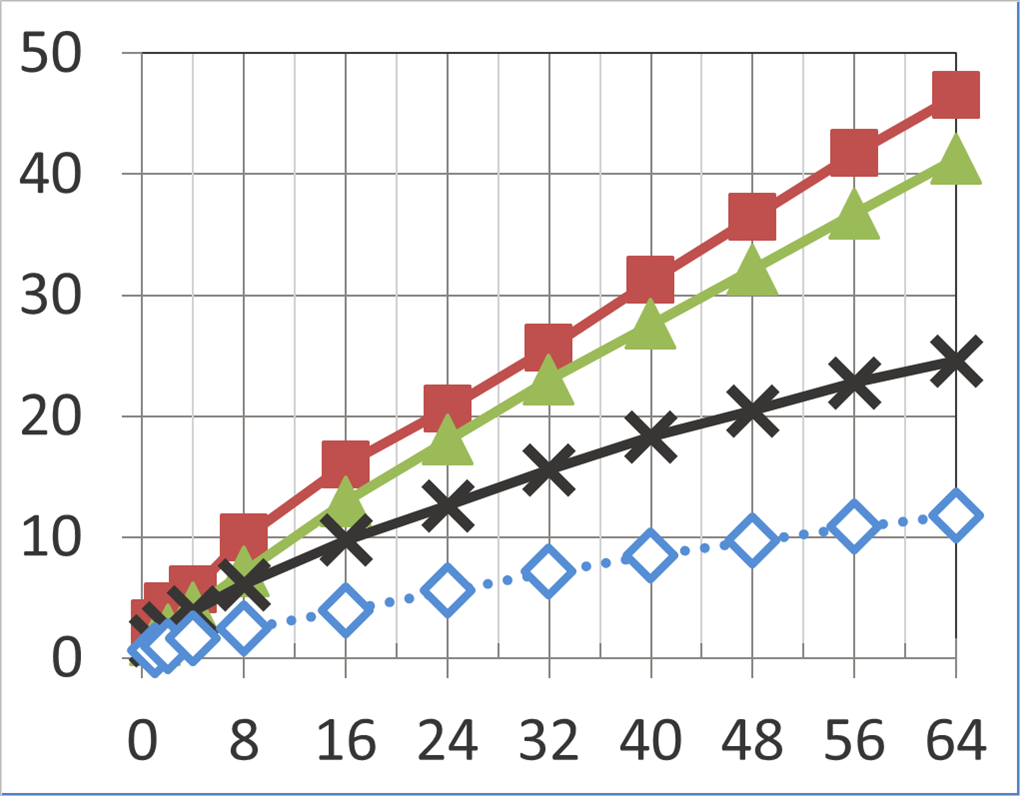}
        \\
        \rotatebox{90}{16-CAS} &
        \includegraphics[width=\linewidth]{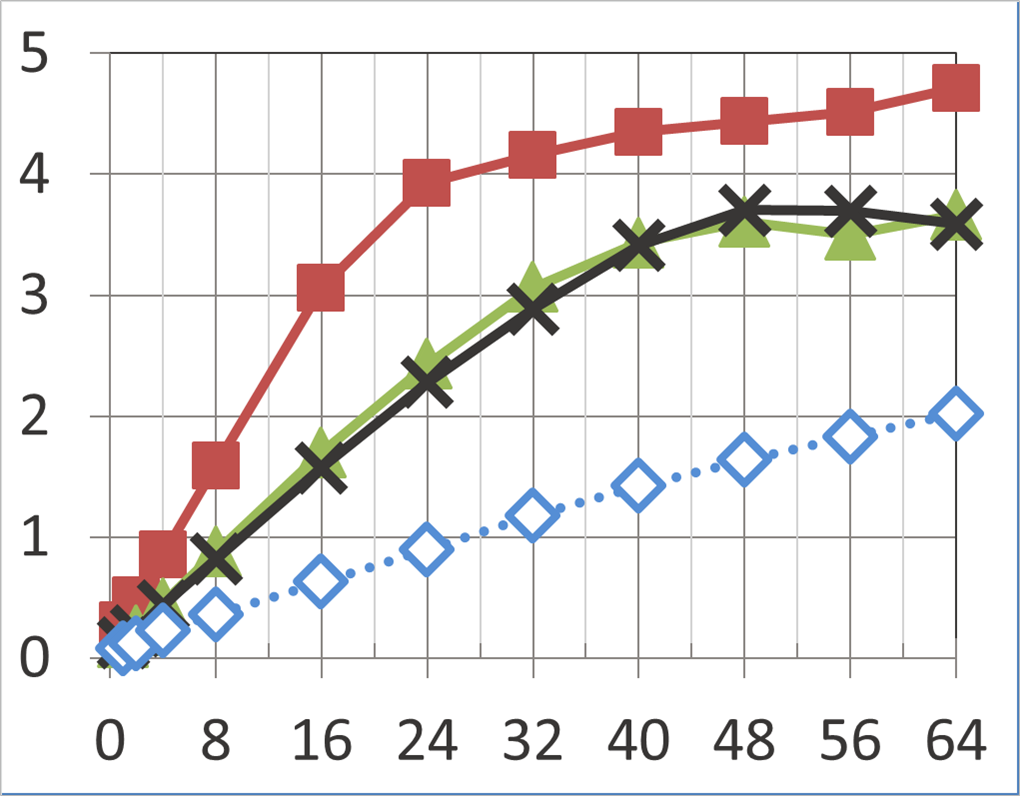} &
        \includegraphics[width=\linewidth]{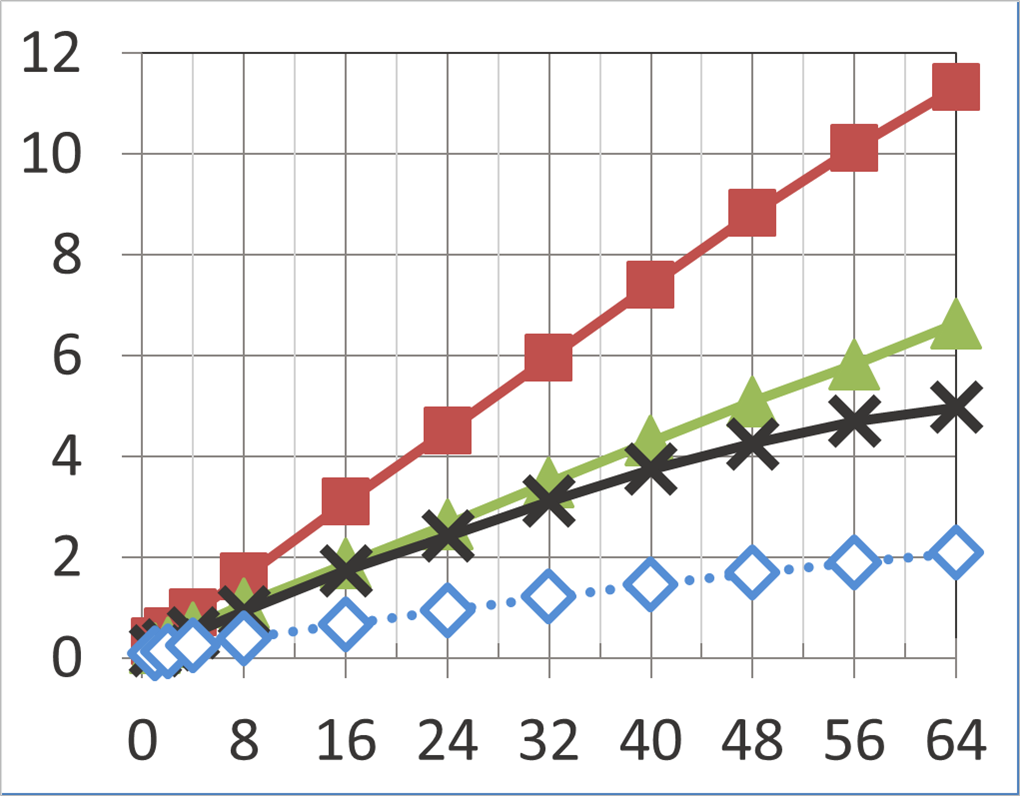} &
        \includegraphics[width=\linewidth]{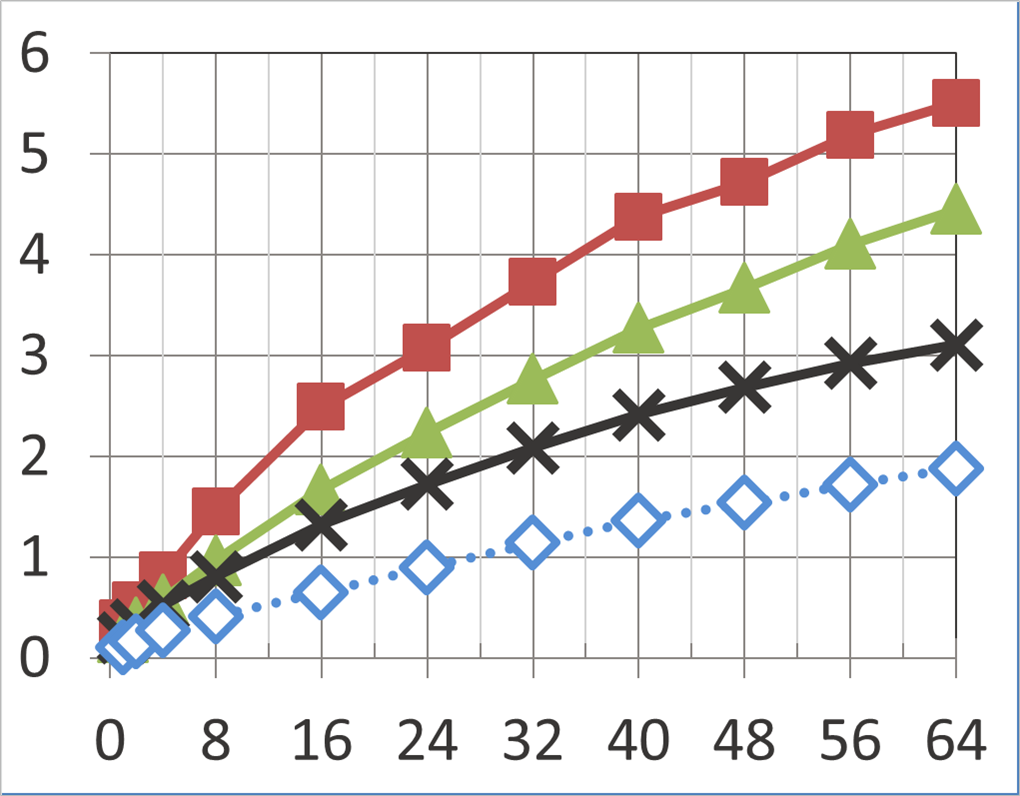}
        \\
    \end{tabular}
	\includegraphics[width=0.8\linewidth]{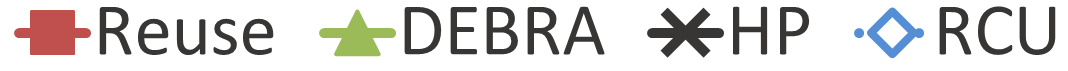}
\caption{Results for a \textbf{$k$-CAS microbenchmark}.
The x-axis represents the number of concurrent threads.
The y-axis represents operations per microsecond.}
\label{fig-exp-kcas}
\end{figure}
\end{fullver}

\fakeparagraph{Results}
The results for this benchmark appear in Figure~\ref{fig-exp-kcas}.
Error bars are not drawn on the graphs, since more than 97\% of the data points have a standard deviation that is less than 5\% of the mean (making them essentially too small to see).

Overall, \textit{Reuse} outperforms every other algorithm, in every workload, on both machines.
Notably, on the Intel machine, its throughput is \textit{2.2 times} that of the next best algorithm at 48 threads with $k=16$ and array size $2^{26}$.
On the AMD machine, its throughput is 1.7 times that of the next best algorithm at 64 threads with $k=16$ and array size $2^{20}$.




On the Intel machine, with $k=2$, NUMA effects are quite noticeable for \textit{Reuse} in the jump from 24 to 32 threads, as threads begin running on the second socket.
According the statistics we collected with PAPI, this decrease in performance corresponds to an increase in cache misses.
For example, with $k=2$ and an array of size $2^{26}$ in the Intel machine, jumping from 24 threads to 25 increases the number of L3 cache misses per operation from 0.7 to 1.6 (with similar increases in L1 and L2 cache misses and pipeline stalls).
We believe this is due to cross-socket cache invalidations.

From the three graphs for $k=2$ on Intel, we can see that the effect is more severe with larger absolute throughput (since the additive overhead of a cache miss is more significant).
Conversely, the effect is masked by the much smaller throughput of the slower algorithms, and by the substantially lower throughputs in the $k=16$ case, except when the array is of size $2^{14}$.
In the array of size $2^{14}$, contention is extremely high, since each of the 48 threads are accessing 16 $k$-CAS addresses, each of which causes contention on the entire cache line of 8 words, for a total of 6144 array entries contended at any given time.
Thus, cache misses become a dominating factor in the performance on two sockets.
These effects were not observed on the AMD machine.
There, the number of cache misses is not significantly different when crossing socket boundaries, which suggests a robustness to NUMA effects that is not seen on the Intel machine.

Interestingly, absolute throughputs on the AMD machine are larger with array size $2^{20}$ than with sizes $2^{14}$ and $2^{26}$.
This is because the $2^{20}$ array size represents a sweet spot with less contention than the $2^{14}$ size and better cache utilization than the $2^{26}$ size.
For example, with 64 threads and $k=16$, \textit{Reuse} incurred approximately 50\% more cache misses with size $2^{26}$ than with size $2^{20}$, and approximately 50\% of operations helped one another with size $2^{14}$, whereas less than 1\% of operations helped one another with size $2^{20}$.

Note, however, that this is not true on the Intel machine.
There, $2^{26}$ is almost always as fast as $2^{20}$, because of the very large shared L3 cache (which is 5x larger than on the AMD machine).
This is reflected in the increased number of cycles where the processor is stalled (e.g., waiting for cache misses to be served) when moving from size $2^{20}$ to $2^{26}$.
On the Intel machine, stalled cycles increase by $85\%$ per operation, whereas on the AMD machine they increase by a whopping $450\%$ per operation.

\begin{figure}[t]
	\centering
	\setlength\tabcolsep{0pt}
	\begin{tabular}{m{0.05\linewidth}m{0.46\linewidth}} 
		&
		\multicolumn{1}{c}{
			\fcolorbox{black!80}{black!40}{\parbox{\dimexpr 0.46\linewidth-2\fboxsep-2\fboxrule}{\centering\textbf{2x 24-thread Intel E7-4830 v3}}}
		}
		\\
		&
		\fcolorbox{black!50}{black!20}{\parbox{\dimexpr \linewidth-2\fboxsep-2\fboxrule}{\centering {\footnotesize Descriptor footprint (in bytes)}}}
		\\
		\rotatebox{90}{$S=2^{26}, k=16$} &
		\includegraphics[width=\linewidth]{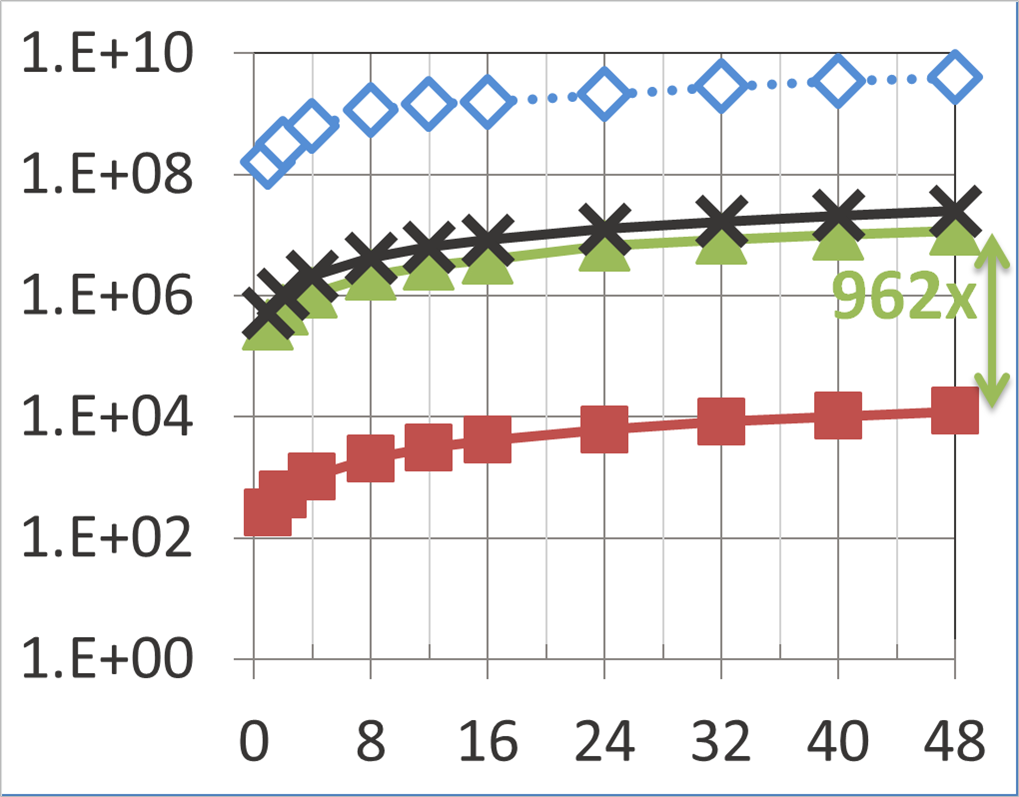}
		\\
	\end{tabular}
	\includegraphics[width=0.5\linewidth]{figures/graphs/legend-small.png}
	\caption{Memory usage for the \textbf{$k$-CAS microbenchmark}.
		The x-axis represents the number of concurrent threads. \textit{Note the logarithmic scale.}} 
	\label{fig-exp-memory}
\end{figure}

\subsubsection{Memory usage in the $k$-CAS benchmark}

We studied memory usage for all algorithms, in all workloads, on both systems, but we only show results for array size $2^{26}$ and $k=16$, because the other graphs are very similar.
These results appear in Figure~\ref{fig-exp-memory}.
%
%
In particular, we are interested in the descriptor footprint, i.e., the maximum amount of memory ever occupied by descriptors in an execution.
Unfortunately, computing the descriptor footprint exactly would require excessive synchronization between threads.
Thus, we approximate the descriptor footprint by computing the descriptor footprint \textit{for each thread}, and then summing those individual footprints.
(This is only an approximation, since different threads may hit their peak memory usage for descriptors at different times.)
The graph in Figure~\ref{fig-exp-memory} contains the results of this approximation.

These results were obtained as follows.
Each thread used three private variables: \textit{totalFree}, \textit{totalMalloc} and \textit{maxFootprint}.
Each time a thread invoked \texttt{free}, it incremented \textit{totalFree} by the size of the descriptor being freed.
Each time a thread invoked \texttt{malloc}, it incremented \textit{totalMalloc} by the size of the descriptor being allocated, and then set $\mbox{\textit{maxFootprint}} = max\{\mbox{\textit{maxFootprint}}, \mbox{\textit{totalMalloc}} - \mbox{\textit{totalFree}}\}$.
The per-thread \textit{maxFootprint}s are then summed to obtain the data points in the graph.

Note that the $y$-axis is a logarithmic scale.
The results show that \textit{DEBRA} and \textit{HPs} use almost \textbf{three orders of magnitude} more memory than \textit{Reuse} at their peaks, and \textit{RCU} uses nearly three orders of magnitude more memory than \textit{DEBRA} and \textit{HPs}.
\textit{RCU}'s memory usage is significantly higher because reclamation is delayed significantly longer than in the other algorithms. 

\subsection{BST microbenchmark}

Unlike in the $k$-CAS algorithm, where memory reclamation was only needed for descriptors, in the BST, memory reclamation is always needed for nodes.
To compare our technique with different memory reclamation options, we implemented four variants of the BST algorithm: \textit{DEBRA/DEBRA}, \textit{DEBRA/Reuse}, \textit{RCU/RCU} and \textit{RCU/Reuse}.
Here, an algorithm named \textit{X/Y} uses \textit{X} to reclaim nodes and \textit{Y} for descriptors.
For example, \textit{DEBRA/Reuse} uses DEBRA to reclaim nodes and has reusable descriptors.

\begin{thesisnot}
Unfortunately, we could not create a variant of the BST using hazard pointers.
As part of the \textit{finalizing} mechanism, this BST implementation \textit{marks} nodes before deleting them.
Furthermore searches are allowed to traverse marked nodes, regardless of whether they have been deleted, and subsequently succeed.
These algorithmic properties make it infeasible to use hazard pointers~\cite{Brown:2015}.
\end{thesisnot}

\medskip

\fakeparagraph{Methodology}
We compared our BST variants using a simple randomized microbenchmark.
For each algorithm $A \in \{$\textit{DEBRA/DEBRA}, \textit{DEBRA/Reuse}, \textit{RCU/RCU}, \textit{RCU/Reuse}$\}$, key range size $K \in \{10^5, 10^6\}$ and update rate $U \in \{100, 0\}$, we run ten timed \textit{trials} for several thread counts $n$.
Each trial proceeds in two phases: \textit{prefilling} and \textit{measuring}.
In the prefilling phase, $n$ concurrent threads perform 50\% \textit{Insert} and 50\% \textit{Delete} operations on keys drawn uniformly randomly from $[0, K)$ until the size of the tree converges to a steady state (containing approximately $K/2$ keys).
Next, the trial enters the measuring phase, during which threads begin counting how many operations they perform.
(These counts are eventually summed over all threads and reported in our graphs.)
In this phase, each thread instead performs $(U/2)$\% \textit{Insert}, $(U/2)$\% \textit{Delete} and $(100-U)$\% \textit{Find} operations on keys drawn uniformly from $[0,K)$ for one second.

As a way of validating correctness in each trial, each thread maintains a \textit{checksum}.
Each time a thread inserts a new key, it adds the key to its checksum.
Each time a thread deletes a key, it subtracts the key from its checksum.
At the end of the trial, the sum of all thread checksums must be equal to the sum of keys in the tree.

\medskip

\fakeparagraph{Results}
\begin{shortver}
\begin{figure}[t]
    \centering
    \setlength\tabcolsep{0pt}
    \begin{minipage}{0.49\linewidth}
    \begin{tabular}{m{0.05\linewidth}m{0.465\linewidth}m{0.465\linewidth}}
        &
        \multicolumn{2}{c}{
            \fcolorbox{black!80}{black!40}{\parbox{\dimexpr 0.93\linewidth-2\fboxsep-2\fboxrule}{\centering\textbf{2x 24-thread Intel E7-4830 v3}}}
        }
        \\
        &
        \fcolorbox{black!50}{black!20}{\parbox{\dimexpr \linewidth-2\fboxsep-2\fboxrule}{\centering {\footnotesize range $[0, 10^5)$}}} &
        \fcolorbox{black!50}{black!20}{\parbox{\dimexpr \linewidth-2\fboxsep-2\fboxrule}{\centering {\footnotesize range $[0, 10^6)$}}}
        \\
        \rotatebox{90}{100\% updates} &
        \includegraphics[width=\linewidth]{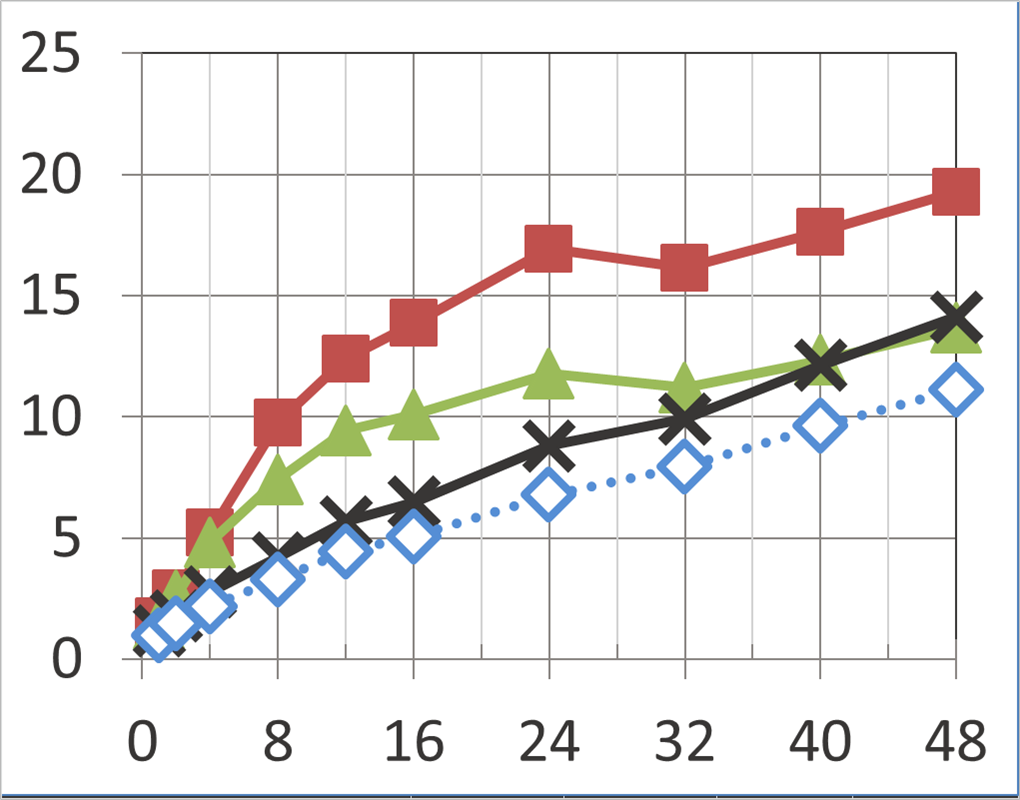} &
        \includegraphics[width=\linewidth]{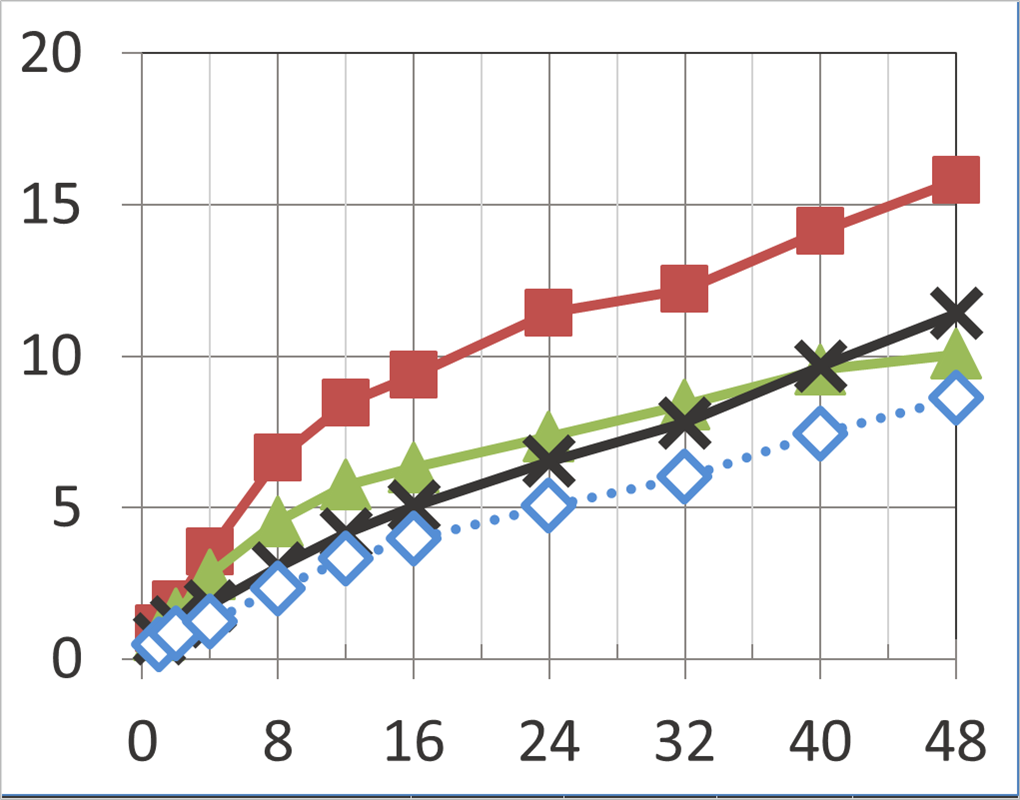}
        \\
        \rotatebox{90}{0\% updates} &
        \includegraphics[width=\linewidth]{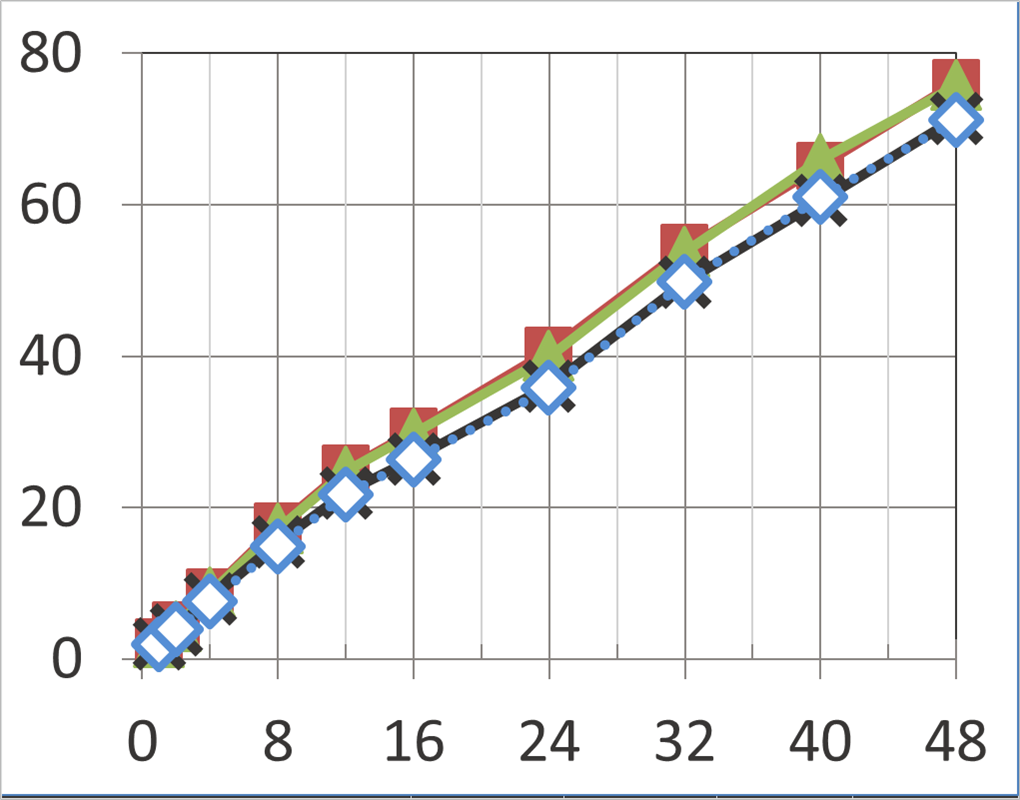} &
        \includegraphics[width=\linewidth]{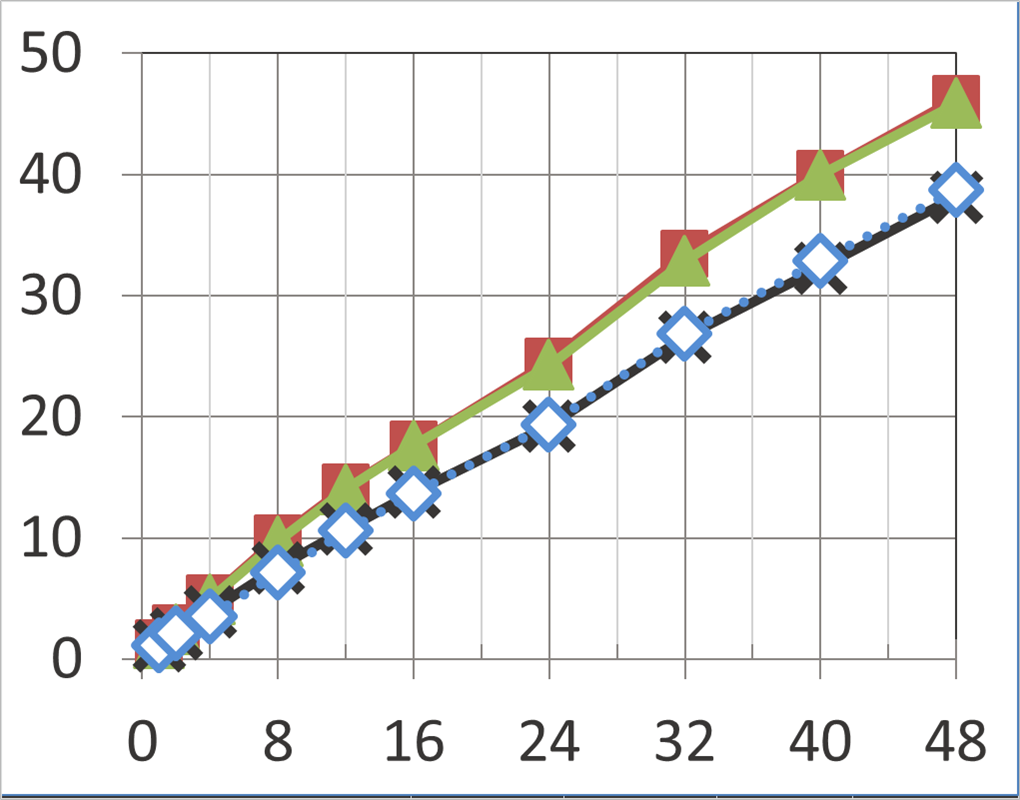}
        \\
    \end{tabular}
    \end{minipage}
    \begin{minipage}{0.49\linewidth}
    \begin{tabular}{m{0.05\linewidth}m{0.465\linewidth}m{0.465\linewidth}}
        &
        \multicolumn{2}{c}{
            \fcolorbox{black!80}{black!40}{\parbox{\dimexpr 0.93\linewidth-2\fboxsep-2\fboxrule}{\centering\textbf{4x 16-thread AMD Opteron 6380}}}
        }
        \\
        &
        \fcolorbox{black!50}{black!20}{\parbox{\dimexpr \linewidth-2\fboxsep-2\fboxrule}{\centering {\footnotesize range $[0, 10^5)$}}} &
        \fcolorbox{black!50}{black!20}{\parbox{\dimexpr \linewidth-2\fboxsep-2\fboxrule}{\centering {\footnotesize range $[0, 10^6)$}}}
        \\
        \rotatebox{90}{100\% updates} &
        \includegraphics[width=\linewidth]{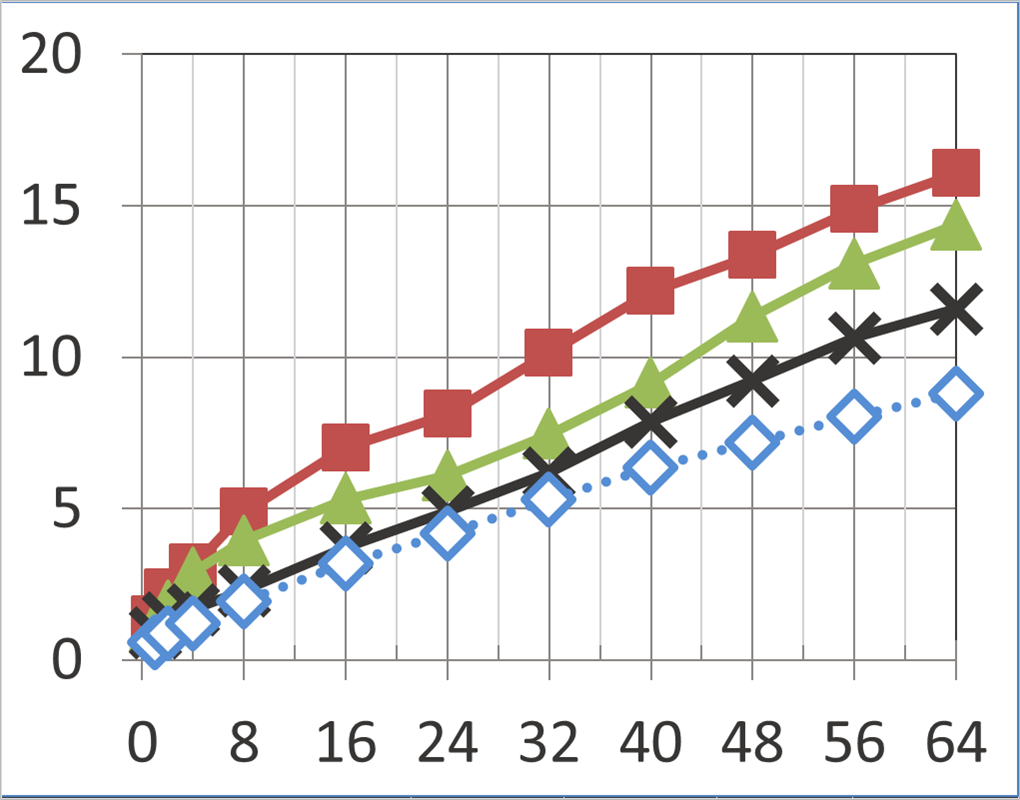} &
        \includegraphics[width=\linewidth]{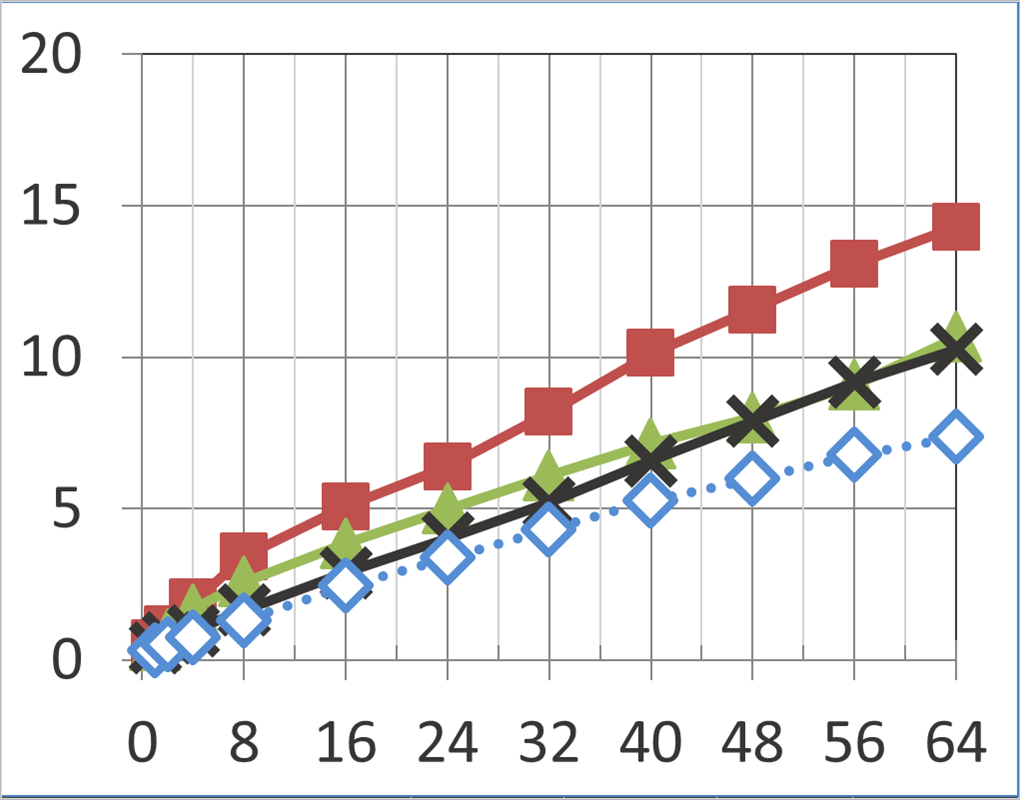}
        \\
        \rotatebox{90}{0\% updates} &
        \includegraphics[width=\linewidth]{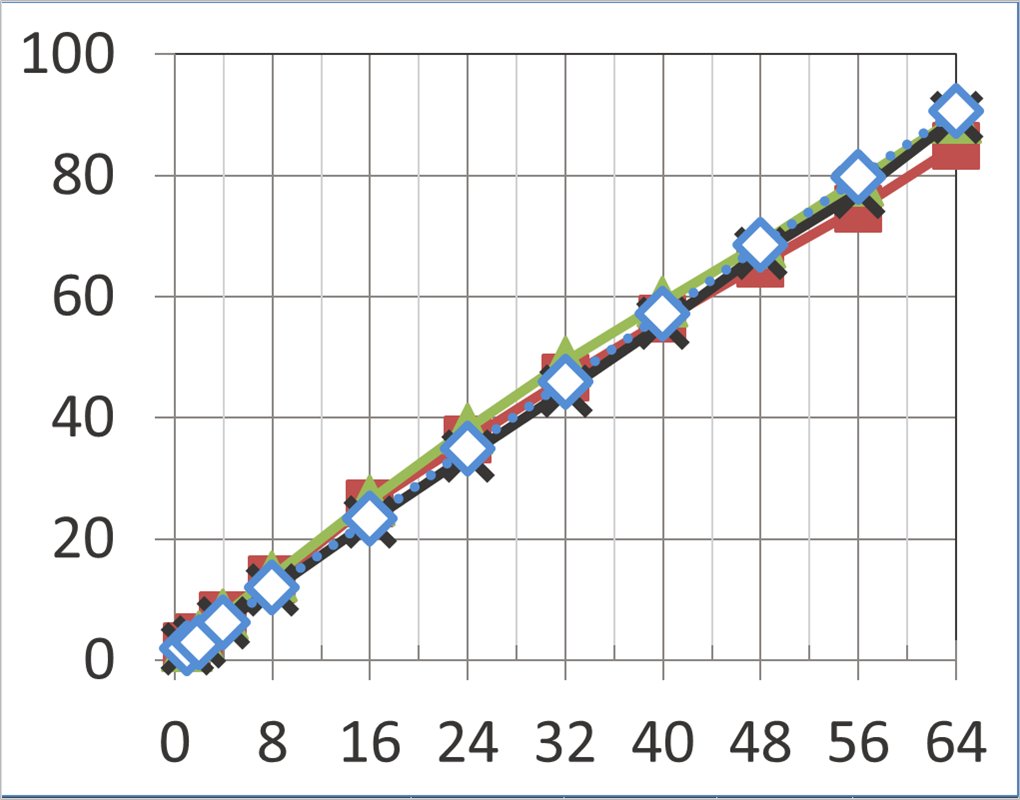} &
        \includegraphics[width=\linewidth]{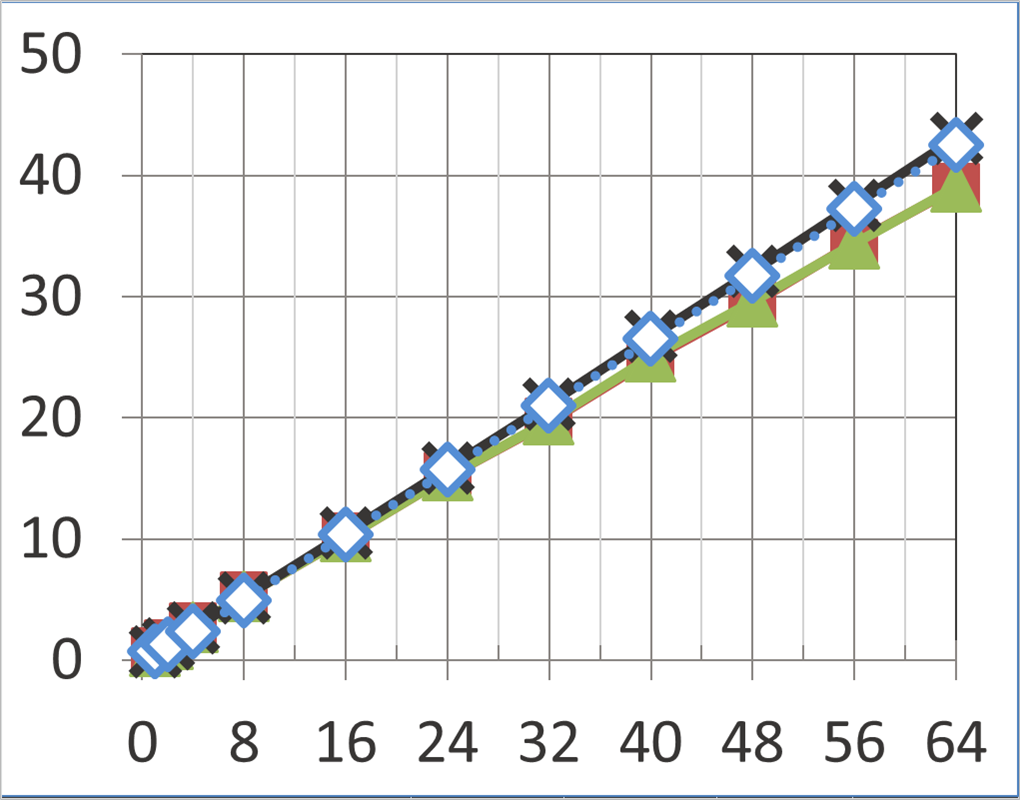}
        \\
    \end{tabular}
    \end{minipage}
	\includegraphics[width=\linewidth]{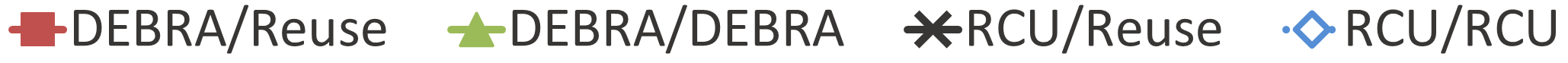}
\caption{Results for a \textbf{BST microbenchmark}.
The x-axis represents the number of concurrent threads.
The y-axis represents operations per microsecond.}
\label{fig-exp-bst}
\end{figure}
\end{shortver}
\begin{fullver}
\begin{figure}[th]
    \centering
    \setlength\tabcolsep{0pt}
    \begin{tabular}{m{0.05\linewidth}m{0.31\linewidth}m{0.31\linewidth}m{0.31\linewidth}}
        &
        \multicolumn{3}{c}{
            \fcolorbox{black!80}{black!40}{\parbox{\dimexpr 0.93\linewidth-2\fboxsep-2\fboxrule}{\centering\textbf{2x 24-thread Intel E7-4830 v3}}}
        }
        \\
        &
        \fcolorbox{black!50}{black!20}{\parbox{\dimexpr \linewidth-2\fboxsep-2\fboxrule}{\centering {\footnotesize 100\% updates}}} &
        \fcolorbox{black!50}{black!20}{\parbox{\dimexpr \linewidth-2\fboxsep-2\fboxrule}{\centering {\footnotesize 10\% updates}}} &
        \fcolorbox{black!50}{black!20}{\parbox{\dimexpr \linewidth-2\fboxsep-2\fboxrule}{\centering {\footnotesize 0\% updates}}}
        \\
        \rotatebox{90}{\small range $[0, 10^6)$} &
        \includegraphics[width=\linewidth]{figures/graphs/reuse-vs-throw8_48_tapuz40_lgchunk21_1000000k-50i-50d.png} &
        \includegraphics[width=\linewidth]{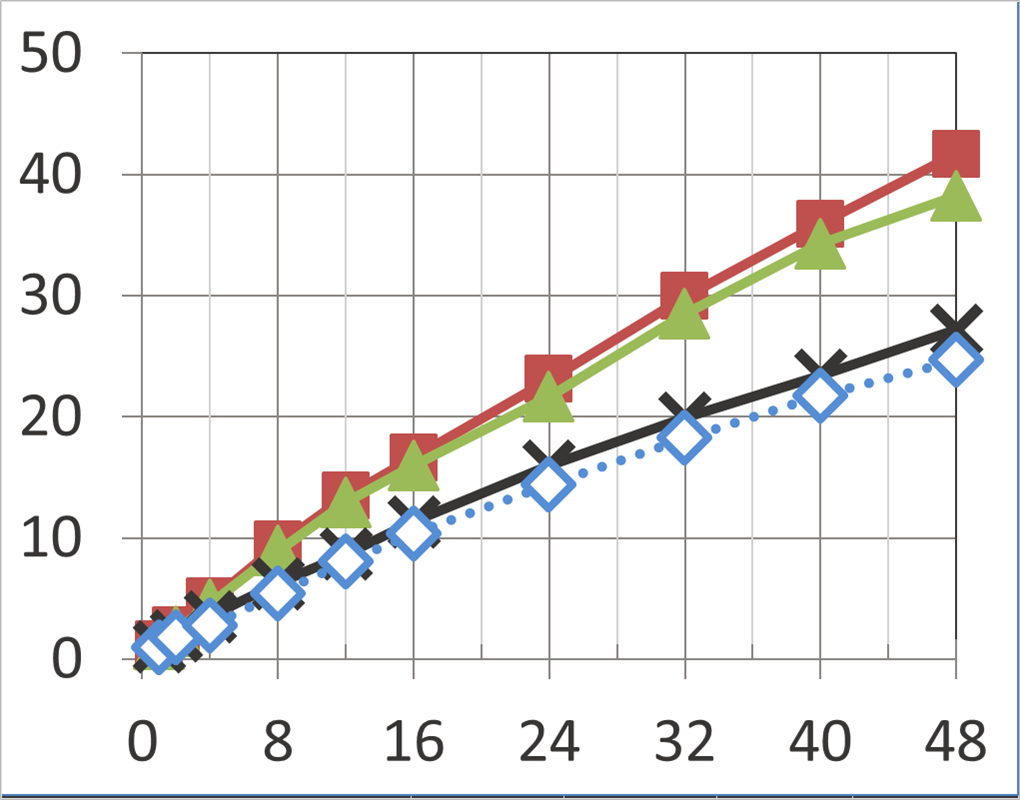} &
        \includegraphics[width=\linewidth]{figures/graphs/reuse-vs-throw8_48_tapuz40_lgchunk21_1000000k-0i-0d.png}
        \\
        \rotatebox{90}{\small range $[0, 10^5)$} &
        \includegraphics[width=\linewidth]{figures/graphs/reuse-vs-throw8_48_tapuz40_lgchunk21_100000k-50i-50d.png} &
        \includegraphics[width=\linewidth]{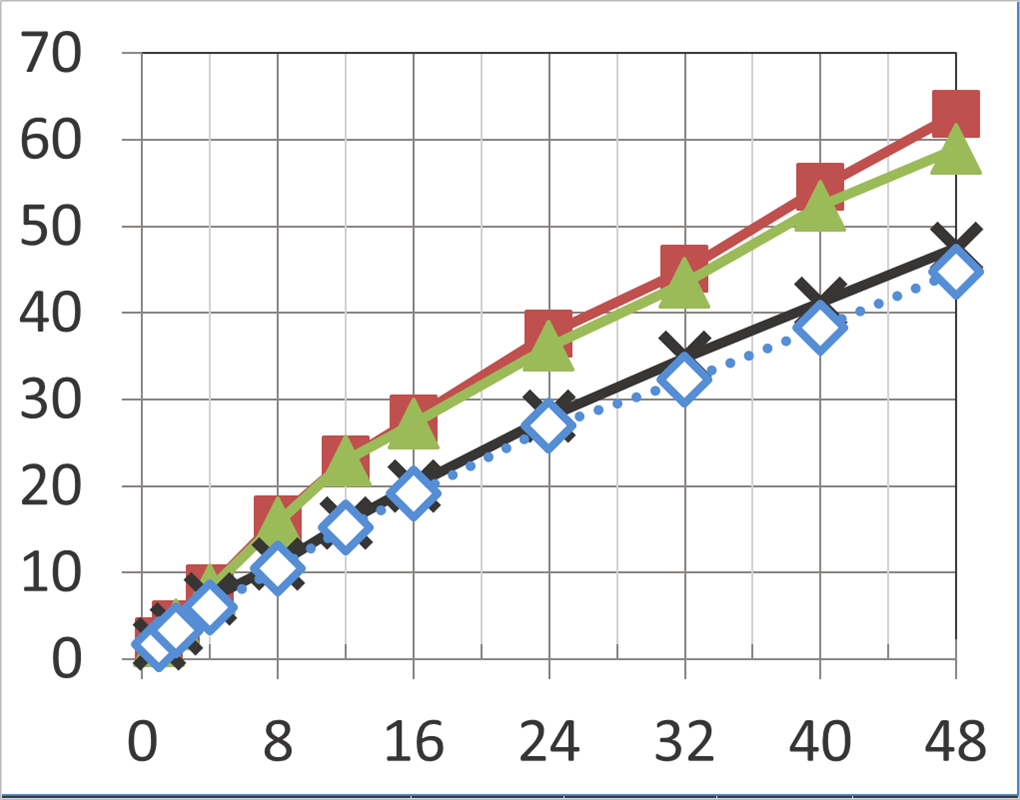} &
        \includegraphics[width=\linewidth]{figures/graphs/reuse-vs-throw8_48_tapuz40_lgchunk21_100000k-0i-0d.png}
        \\
    \end{tabular}

    \begin{tabular}{m{0.05\linewidth}m{0.31\linewidth}m{0.31\linewidth}m{0.31\linewidth}}
        &
        \multicolumn{3}{c}{
            \fcolorbox{black!80}{black!40}{\parbox{\dimexpr 0.93\linewidth-2\fboxsep-2\fboxrule}{\centering\textbf{4x 16-thread AMD Opteron 6380}}}
        }
        \\
        &
        \fcolorbox{black!50}{black!20}{\parbox{\dimexpr \linewidth-2\fboxsep-2\fboxrule}{\centering {\footnotesize 100\% updates}}} &
        \fcolorbox{black!50}{black!20}{\parbox{\dimexpr \linewidth-2\fboxsep-2\fboxrule}{\centering {\footnotesize 10\% updates}}} &
        \fcolorbox{black!50}{black!20}{\parbox{\dimexpr \linewidth-2\fboxsep-2\fboxrule}{\centering {\footnotesize 0\% updates}}}
        \\
        \rotatebox{90}{\small range $[0, 10^6)$} &
        \includegraphics[width=\linewidth]{figures/graphs/reuse-vs-throw8_64_csl-pomela6_lgchunk19_1000000k-50i-50d.png} &
        \includegraphics[width=\linewidth]{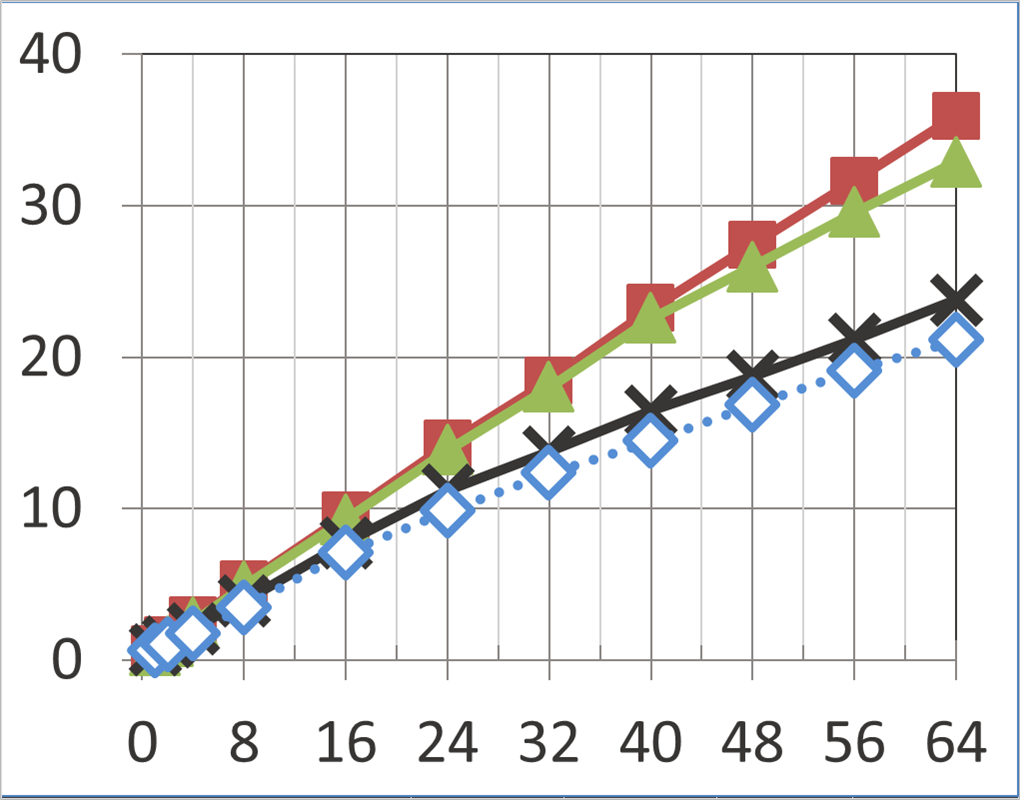} &
        \includegraphics[width=\linewidth]{figures/graphs/reuse-vs-throw8_64_csl-pomela6_lgchunk19_1000000k-0i-0d.png}
        \\
        \rotatebox{90}{\small range $[0, 10^5)$} &
        \includegraphics[width=\linewidth]{figures/graphs/reuse-vs-throw8_64_csl-pomela6_lgchunk19_100000k-50i-50d.png} &
        \includegraphics[width=\linewidth]{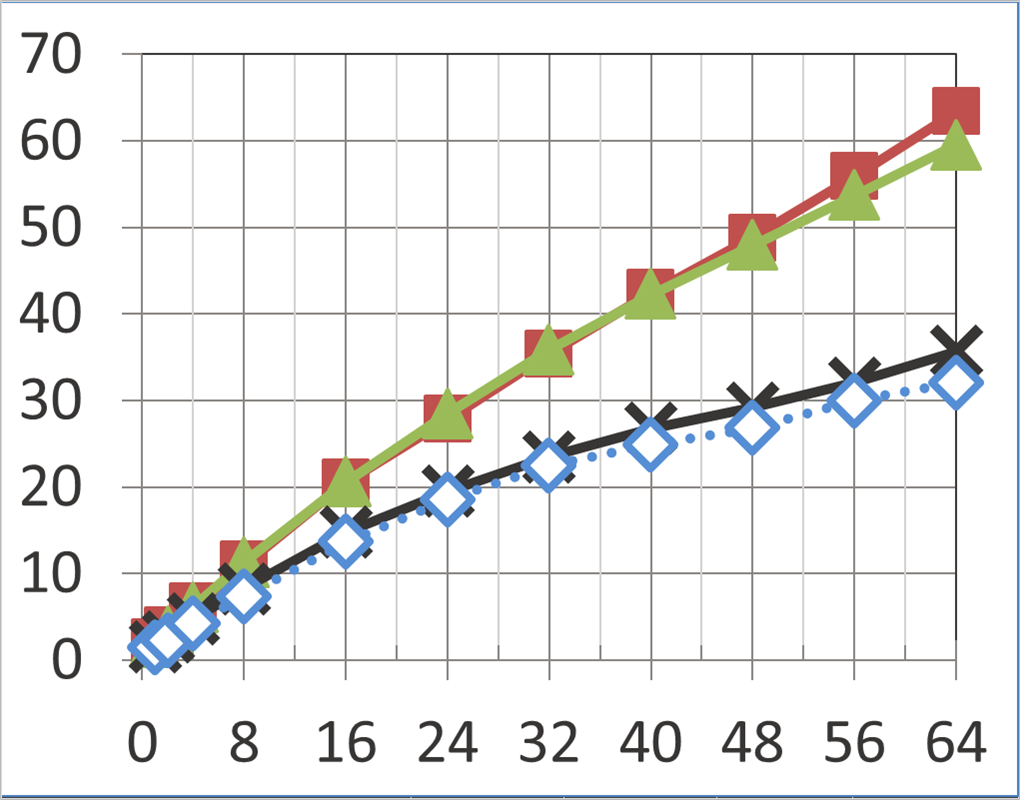} &
        \includegraphics[width=\linewidth]{figures/graphs/reuse-vs-throw8_64_csl-pomela6_lgchunk19_100000k-0i-0d.png}
        \\
    \end{tabular}
	\includegraphics[width=\linewidth]{figures/graphs/legend-bst.png}
\caption{Results for a \textbf{BST microbenchmark}.
The x-axis represents the number of concurrent threads.
The y-axis represents operations per microsecond.}
\label{fig-exp-bst}
\end{figure}
\end{fullver}
The results for this benchmark appear in Figure~\ref{fig-exp-bst}.
The \textit{Reuse} variants perform at least as well as the pure reclamation variants in every case, and significantly outperform the reclamation variants in the 100\% update workload.
Most notably, on the Intel machine with key range $[0, 10^6]$ and 48 threads, \textit{DEBRA/Reuse} outperforms \textit{DEBRA/DEBRA} by 57\%, and \textit{RCU/Reuse} outperforms \textit{RCU/RCU} by 33\%.
As expected, \textit{Reuse} does not perform significantly faster than the reclamation variants in the workloads with no updates.
This is because searches do not create descriptors.
However, crucially, our transformation does not impose any overhead on searches, either. 

\subsection{Studying sequence number wraparound} \label{sec-exp-wraparound}

We performed experiments on the larger AMD machine to study how frequently errors occur when sequence numbers of varying bit-widths experience wraparound.
For each bit-width $B \in \{2,3,4,...,48\}$, we performed 200 trials in which 64 threads run for 100 milliseconds before terminating.
Each trial was the same as a trial in our BST experiments with 100\% updates and key range $[0, 10^5)$.

We identified three different types of errors in these trials.
First, at the end of a trial, the sum of the checksums maintained by all threads would fail to match the sum of keys in the tree.
Second, threads would enter infinite loops due to the tree structure being corrupted, e.g., because a cycle was introduced.
(We identified this type of error by waiting until some thread had run twice as long as it should have.)
Third, an invalid memory access would cause immediate program failure (e.g., due to segmentation fault or bus error).

\begin{figure}[t]
\centering
    \small
    \begin{tabular}{m{0.6\linewidth}m{0.3\linewidth}}
        \includegraphics[width=\linewidth]{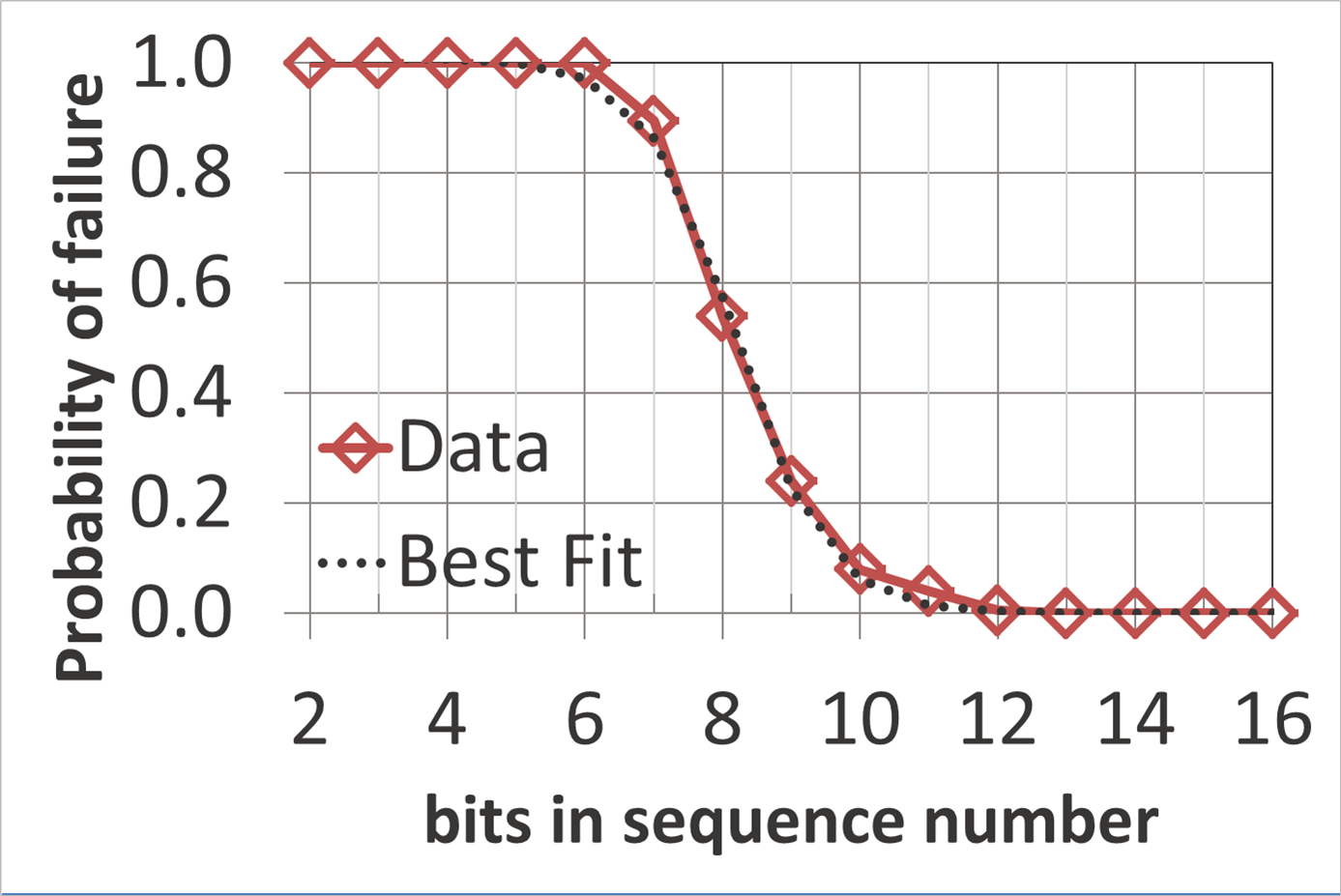} &
        \begin{tabular}{|r|l|}
            \hline
            Bits & E[time until error] \\
            \hline
            16 & 4.5 hours \\
            24 & 116 years \\
            32 & 26078192 years \\
            48 & $> 10^{18}$ years \\
            \hline
        \end{tabular}
    \end{tabular}
    \caption{
Experiment studying sequence number wraparound. 
}
    \label{fig-exp-wraparound}
\vspace{-4mm}
\end{figure}

For each $B$ value, we divided the number of failed runs by 200 to estimate the probability of a trial failing.
A graph showing the resulting estimated probability distribution appears in Figure~\ref{fig-exp-wraparound}.
For small $B$ values, trials frequently experienced errors.
However, for $B \ge 13$, we did not observe a single error in 200 trials (despite the fact that wraparound consistently occurred in every trial).
For $B \ge 16$, trials were not sufficiently long for wraparound to consistently occur.
The results appear in Figure~\ref{fig-exp-wraparound}.

As is common in physics when studying unknown functions, we make an educated guess that the distribution is sigmoidal, 
which means it is of the form $f(x) = a/(1+e^{-b(x-c)})$ for constants $a, b$ and $c$.
We determined a sigmoidal curve of best fit from the data, obtaining the function $f(x) = 1/(1+e^{1.53969(x-8.199181)})$, which
is plotted as the \textit{Best Fit} curve on the graph in Figure~\ref{fig-exp-wraparound}.
As the graph shows, the error between the best fit curve and the measured data is extremely small.
Although we do not have a justification for the shape of this distribution, we think it is worthwhile to put forth a hypothesis and study its consequences. 

We used $f(x)$ to extrapolate on the data to estimate the expected time until an error occurs in this workload for several bit-widths that would be impractical to test experimentally.
These extrapolations appear in the table on the right of Figure~\ref{fig-exp-wraparound}.
They should be taken with a grain of salt, since the error in our estimation likely grows quickly with $B$.
However, the extrapolations suggest that even $B = 32$ would be quite safe for this workload.
To our knowledge, this kind of experimental exploration of the practicality of unbounded sequence numbers has not previously been done.


\section{Related Work} \label{sec-related}

Several papers have presented universal constructions or strong primitives for non-blocking algorithms in which operations create descriptors~\cite{Israeli:1994, Anderson:1995, Afek:1995, Moir:1997, Harris:2002, LMS09, JP05:opodis, Marathe:2008, Attiya:2011, Brown:2013}. 
A subset of these algorithms employ ad-hoc techniques for reusing descriptors~\cite{Israeli:1994, Anderson:1995, Afek:1995, Moir:1997, Marathe:2008, LMS09, JP05:opodis}.
The rest assume descriptors will be allocated for each operation and eventually reclaimed.

Most of the ad-hoc techniques for reusing descriptors have significant downsides.
Some are complex and tightly integrated into the underlying algorithm, or rely on highly specific algorithmic properties (e.g., that descriptors contain only a single word). 
Others use synchronization primitives that atomically operate on large words, which are not available on modern systems, and are inefficient when implemented in software.
Yet others introduce high space overhead (e.g., by attaching a sequence number to \textit{every} memory word).
Some techniques also incur significant runtime overhead (e.g., by invoking expensive synchronization primitives just to \textit{read fields} of a descriptor).
Furthermore, these techniques give, at best, a vague idea of how one might reuse descriptors for arbitrary algorithms, and it would be difficult 
to determine how to use them in practice.
Our work avoids all of these downsides, and provides a concrete approach for transforming a large class of algorithms.


Barnes~\cite{Barnes:1993} introduced a technique for producing non-blocking algorithms that can be more efficient (and sometimes simpler) than the universal constructions described above.
With Barnes' technique, each operation creates a new descriptor.
Creating a new descriptor for each operation allows his technique to avoid the ABA problem while remaining conceptually simple.
Each operation conceptually locks each location it will modify by installing a pointer to its descriptor, and then performs it modifications and unlocks each location.
Barnes' technique is the inspiration for the class WCA. 
Many algorithms have since been introduced using variants of this technique~\cite{Harris:2002, Ellen:2010, Attiya:2011, Howley:2012, Shafiei:2013, Brown:2013, Brown:2014}.
Several of these algorithms are quite efficient in practice despite the overhead of creating and reclaiming descriptors.
Our technique can significantly improve the space and time overhead of such algorithms.

Recent work has identified ways to use hardware transactional memory (HTM) to reduce descriptor allocation~\cite{Brown:2017:TIF,Liu2015}.
Currently, HTM is supported only on recent Intel and IBM processors.
Other architectures, such as AMD, SPARC and ARM have not yet developed HTM support.
Thus, it is important to provide solutions for systems with no HTM support.
Additionally, even with HTM support, our approach is useful.
Current (and likely future) implementations of HTM offer no progress guarantees, so one must provide a lock-free fallback path to guarantee lock-free progress.
The techniques in~\cite{Brown:2017:TIF,Liu2015} accelerate the HTM-based code path(s), but do nothing to reduce descriptor allocations on the fallback path.
In some workloads, many operations run on the fallback path, so it is important for it to be efficient.
Our work provides a way to accelerate the fallback path, and is orthogonal to work that optimizes the fast path.

The \textit{long-lived renaming} (LLR) problem is related to our work (see~\cite{Brodsky2011} for a survey), but its solutions do not solve our problem.
LLR provides processes with operations to \textit{acquire} one unique resource from a pool of resources, and subsequently \textit{release} it.
One could imagine a scheme in which processes use LLR to reuse a small set of descriptors by invoking \textit{acquire} instead of allocating a new descriptor, and eventually invoking \textit{release}.
Note, however, that a descriptor can safely be released only once it can no longer be accessed by any other process.
Determining when it is safe to release a descriptor is as hard as performing general memory reclamation, and would also require delaying the release (and subsequent acquisition) of a descriptor (which would increase the number of descriptors needed).
In contrast, our weak descriptors eliminate the need for memory reclamation, and allow immediate reuse.

\begin{thesisonly}
\section{Summary}
\end{thesisonly}
\begin{thesisnot}
\section{Conclusion} 
\end{thesisnot}
\label{sec-desc-conclusion}

We presented a novel technique for transforming algorithms that throw away descriptors into algorithms that reuse descriptors.
Our experiments show that our transformation yields significant performance improvements for a lock-free $k$-CAS algorithm. 
Furthermore, our transformation reduces peak memory usage 
by nearly three orders of magnitude over the next best implementation.

We also applied our transformation to a lock-free implementation of \llt\ and \sct, and studied its performance by doing rigorous experiments on a lock-free binary search tree that uses \llt\ and \sct. 
These experiments demonstrated a significant performance advantage for our transformed algorithm in workloads that perform many updates.
Our transformed \llt\ and \sct\ algorithm has the potential to improve the performance of many algorithms that use \llt\ and \sct.

We believe our transformation can be used to improve the performance and memory usage of many other algorithms that throw away descriptors.
Moreover, we hope that our extended weak descriptor ADT will aid in the design of more efficient, complex algorithms, by allowing algorithm designers to benefit from the conceptual simplicity of throwing away descriptors without paying the practical costs of doing so.

\medskip

\section*{Acknowledgments}

We thank Faith Ellen for her gracious help in proving correctness for our transformations, and her insightful comments.
Some of this work was done while Trevor was a student at the University of Toronto, and while Maya was visiting him there.
This work was supported by the Israel Science Foundation (grant 1749/14), the Yad-HaNadiv foundation, the Natural Sciences and Engineering Research Council of Canada, and Global Affairs Canada.
Maya is supported in part by the Technion Hasso Platner Institute Research School.

%
\bibliography{bibliography}  
%

\end{document}